\documentclass[10pt]{article}

\usepackage{amssymb}
\usepackage{amsmath}
\usepackage{relsize}

\def\re{{\rm Re}}
\def\im{{\rm Im}}
\def\dim{\rm dim}
\def\Ext{\rm Ext}
\def\wt#1{{{\widetilde #1} }}
\def\wh#1{{{\,\widehat #1\,} }}
\def\graph{{\rm gr\,}}
\def\ran{{\rm ran\,}}
\def\dom{{\rm dom\,}}
\def\ker{{\rm ker\,}}
\def\supp{{\rm supp\,}}
\def\diag{{\rm diag\,}}

\newcommand\dN{{\mathbb{N}}}
\newcommand\dR{{\mathbb{R}}}
\newcommand\dC{{\mathbb{C}}}
\newcommand{\bO}{{\mathbb{O}}}

\newcommand\dZ{{\mathbb{Z}}}

\newcommand\gotD{{\mathfrak{D}}}
\newcommand\gotH{{\mathfrak{H}}}
\newcommand\gotK{{\mathfrak{K}}}

\newcommand\gotN{{\mathfrak{N}}}

\newcommand\gotS{{\mathfrak{S}}}
\newcommand\gotT{{\mathfrak{T}}}

\newcommand{\ga}{{\alpha}}
\newcommand{\gb}{{\beta}}
\newcommand{\gd}{{\delta}}
\newcommand{\gD}{{\Delta}}
\newcommand{\gga}{{\gamma}}
\newcommand{\gG}{{\Gamma}}

\newcommand{\gl}{{\lambda}}

\newcommand{\gO}{{\Omega}}

\newcommand{\gs}{{\sigma}}
\newcommand\gS{{\Sigma}}
\newcommand{\gT}{{\Theta}}
\newcommand{\gt}{\tau}

\newcommand\cB{{\mathcal{B}}}
\newcommand\cC{{\mathcal{C}}}
\newcommand\cD{{\mathcal{D}}}
\newcommand\cH{{\mathcal{H}}}

\newcommand\cN{{\mathcal{N}}}

\newcommand\cT{{\mathcal{T}}}

\newcommand\cZ{{\mathcal{Z}}}

\newtheorem{theorem}{Theorem}[section]
\newtheorem{proposition}[theorem]{Proposition}
\newtheorem{corollary}[theorem]{Corollary}
\newtheorem{lemma}[theorem]{Lemma}
\newtheorem{definition}[theorem]{Definition}
\newtheorem{example}[theorem]{Example}
\newtheorem{remark}[theorem]{Remark}

\newcommand{\ba}{\begin{array}}
\newcommand{\ea}{\end{array}}

\newcommand{\bea}{\begin{eqnarray}}
\newcommand{\eea}{\end{eqnarray}}

\newcommand{\bead}{\begin{eqnarray*}}
\newcommand{\eead}{\end{eqnarray*}}

\newcommand{\be}{\begin{equation}}
\newcommand{\ee}{\end{equation}}

\newcommand{\bed}{\begin{displaymath}}
\newcommand{\eed}{\end{displaymath}}

\newcommand{\bl}{\begin{lemma}}
\newcommand{\el}{\end{lemma}}

\newcommand{\bt}{\begin{theorem}}
\newcommand{\et}{\end{theorem}}

\newcommand{\bc}{\begin{corollary}}
\newcommand{\ec}{\end{corollary}}

\newcommand{\bd}{\begin{definition}}
\newcommand{\ed}{\end{definition}}

\newcommand{\bspi}{\begin{split}}
\newcommand{\espi}{\end{split}}

\newcommand{\la}{\label}

\newenvironment{proof}%
{\begin{sloppypar}\noindent{\bf Proof.}}%
{\hspace*{\fill}$\square$\end{sloppypar}}

\newcommand{\slim}{\,\mbox{\rm s-}\hspace{-2pt} \lim}

\newcommand{\transpose}[1]{\ensuremath{#1^{\scriptscriptstyle t}}}

\numberwithin{equation}{section}

\title{Boundary triplets, tensor products\\ and point contacts to reservoirs}

\author{A.A.~Boitsev
\and
J.F.~Brasche
\and
M.M.~Malamud
\and
H.~Neidhardt
\and
I.Yu.~Popov
}

\begin{document}

\maketitle

\begin{abstract}
\noindent
We consider symmetric operators of the form $S := A\otimes I_{\gotT} + I_{\gotH} \otimes
T$ where $A$ is symmetric and $T = T^*$ is (in general) unbounded. Such operators
naturally arise in problems of simulating  point contacts to
reservoirs. We construct a boundary triplet $\Pi_S$ for $S^*$ preserving
the tensor structure. The corresponding $\gamma$-field and Weyl function are
expressed by means of the $\gamma$-field and Weyl function corresponding to the
boundary  triplet $\Pi_A$ for  $A^*$ and the spectral measure of $T$.
Applications to 1-D Schr\"odinger and Dirac operators  are given. A model of electron
transport through a quantum dot assisted by cavity photons is proposed. In this
model the boundary operator is chosen to be the well-known Jaynes-Cumming operator which is regarded
as the Hamiltonian of the quantum dot.\\


\end{abstract}

\newpage
\tableofcontents

\section{{Introduction}}

In the following we are interested in the description of point contacts of quantum systems to quantum reservoirs.
Let us recall the general philosophy of modeling of point contacts in quantum mechanics. Let $\{\gotK,S_0\}$ be  a quantum system
where $S_0$ is a self-adjoint operator acting on the separable Hilbert space $\gotK$. To describe point interactions one restricts the self-adjoint operator $S_0$ to a densely defined
closed symmetric operator $S$ and extends it subsequently to another self-adjoint operator $S'$. The new self-adjoint operator $S'$ is regarded as the Hamiltonian
taking into account point interactions. Which extension one has to choose depends on the physical problem. Typical examples for instance are $\gd$ and $\gd'$-point interactions, cf.
\cite{AlbGeszHoegHolden2005}.
From the mathematical point of view it is interesting to note that the problem of describing point interactions fits into the framework of extension theory for symmetric operators.

To describe point contacts of a quantum system with a reservoir one has to specify the approach. At first, one considers the
compound system consisting of the quantum system $\{\gotH,A_0\}$ and the reservoir $\{\gotT,T\}$ where $A_0$ and $T$ are self-adjoint operators on the separable Hilbert spaces
$\gotH$ and $\gotT$, respectively. Its Hamiltonian is given
by the self-adjoint operator
\be\la{eq:1.1}
S_0 := A_0 \otimes I_{\gotT} + I_\gotH \otimes T
 \ee
where $S_0$ acts in the Hilbert space $\gotK := \gotH \otimes \gotT$. To model a contact
to  the quantum reservoir,  the Hamiltonian $S_0$ is usually additively perturbed in a
suitable manner, cf. \cite{Davies1976,DerezJaks2004,DerezFrueb2006}. On the other hand,
to model point contacts to reservoirs we use the restricting-extension procedure: We
restrict the  operator $A_0 = A_0^*$ to a densely defined closed symmetric operator $A$
and consider the closed symmetric operator
\be\la{eq:1.2}
S := A \otimes I_{\gotT} + I_\gotH \otimes T \subset S_0.
\ee
From the physical point of view, the restriction of $A_0$ to $A$ and the subsequent extension to a self-adjoint operator $S'$
can be regarded as the opening of the quantum system $\{\gotH,A_0\}$ and the subsequent coupling of it
to a reservoir. The self-adjoint extension $S'$ should be different from $S_0$.
However,   self-adjoint extensions preserving tensor product form $\widetilde S =
\widetilde A \otimes I_{\gotT} + I_{\gotH} \otimes T$ with $\widetilde A = \widetilde
A^*$ being an extension of $A$, do not describe any interaction with the reservoir. From
the physical point of view it is very important to describe all those extensions, which
really describe point interactions with the reservoir.

In this paper  we investigate  operator  \eqref{eq:1.2}  in the framework of boundary
triplets and the corresponding Weyl functions. This is a new approach to the extension
theory of symmetric operators that has been  developed during the last three decades
(see e.g. \cite{DHMS12,DM87,DM91,DM95,GG91,Koch79,Schmuedgen2012}).

A boundary triplet $\Pi_A = \{\cH^A,\gG^A_0,\gG^A_1\}$ for the adjoint operator $A^*$ of
a densely defined closed symmetric operator $A$ consists of an auxiliary Hilbert space
$\cH^A$ and linear mappings $\gG_0^A,\gG_1^A: \dom(A^*) \longrightarrow \cH^A$ such that
the abstract Green's identity
\bed (A^*f,g) - (f,A^*g) = (\gG^A_1f,\gG^A_0g) - (\gG^A_0f,\gG^A_1g),\qquad f,g\in
\dom(A^*), \eed
holds  and the mapping
\bed \gG^A :=
\begin{pmatrix}
\gG^A_0\\
\gG^A_1
\end{pmatrix}: \dom(A^*)  \longrightarrow
\begin{matrix}
\cH^A\\
\oplus\\
\cH^A
\end{matrix}
\eed
is surjective.

A boundary triplet for $A^*$ exists whenever $A$ has equal deficiency indices.
It  plays the role of a "coordinate system" for the quotient space
$\dom(A^\ast)/ \dom(A)$  and leads to a natural  parametrization   of the
self-adjoint extensions of $A$ by means of self-adjoint linear relations (multi-valued
operators) in $\cH$, see \cite{GG91} and \cite{Schmuedgen2012} for details.
More precisely,  any  self-adjoint extension $\wt A$ of $A$ defines  a self-adjoint
relation $\gT := \gG^A \dom(\wt A)$ in $\cH^A$ and vice versa. We write $A_{\gT} = \wt
A$, i.e.
   \bed
\dom(\wt A) = \dom(A_{\gT}) :=\{f \in \dom(A^*): \gG^A f \in \gT\}.
  \eed
If $\gT$ is an operator $\gT= B$, this relation  takes the form
   \be\la{1.5A}
\dom(\wt A) = \dom(A_{B}) :=\{f \in \dom(A^*): \gG^A_1 f = B\gG^A_0 f \}
  \ee
and  looks like  an abstract boundary condition. Among all self-adjoint extensions
there are two augmented ones: $A_0 := A^*\upharpoonright\ker(\gG^A_0)$ and $A_1 :=
A^*\upharpoonright\ker(\gG^A_1)$ which correspond to self-adjoint relations $\gT_0 :=
(0,f)^t$ and $\gT_1 := (f,0)^t$, $f \in \cH$, respectively. Clearly, $\gT^{-1}_1 =
\gT_0$.

The main analytical tool in this approach is  the abstract Weyl function $M^A(\cdot)$
which was introduced and studied in \cite{ DM91,DM87}. This abstract  Weyl
function $M^A(\cdot)$ plays a similar role in the theory of  boundary triplets as the
classical Weyl-Titchmarsh function does  it in the  theory of Sturm-Liouville operators. In
particular, it allows one to investigate spectral properties of extensions (see
\cite{BMN02, DM91, Mal92, MalNei09}).  The Weyl function is defined by
\bed
M^A(z) := \gG_1^A\gga^A(z), \quad z \in \rho(A_0),
\eed
where
   \bed
\gga^A(z) := (\gG_0\upharpoonright \gotN_z)^{-1}\quad, \quad \gotN_z := \ker(A^* -
z),\quad z \in \rho(A_0).
   \eed
Here $\gga^A(\cdot)$ is the so-called $\gamma$-field,  the second important quantity
related to a boundary triplet $\Pi_A$.

Emphasize that a boundary triplet for $A^*$ is not unique. Its role in extension theory
is similar to that of a coordinate system in analytic geometry. The problem is to
construct an adequate  ("good") boundary triplet such that the corresponding Weyl
function and the boundary operator corresponding to the extension of interest have
"good" properties. To demonstrate the later point we mention boundary triplets for
direct sum of maximal  1-D Schr\"odinger operators $-d^2/x^2 +q$ in $L^2[x_{n-1}, x_n]$
constructed in  \cite{KosMal10,KosMalNat2016, CarlMalPos2013}, where the boundary
operator corresponding to 1-D Schr\"odinger operator with infinitely many point
interactions in $L^2(\dR_+) = \bigoplus_{n=1}^{\infty} L^2[x_{n-1}, x_n]$ is a
Jacobi matrix.

This approach has successfully been applied to the characterization of the absolutely
continuous spectrum of self-adjoint realizations  \cite{BMN02, MalNei09}, as well as to
the investigation of the spectral properties of  1-D Schr\"odinger and 1-D Dirac
operators with  point interactions
\cite{EckKosMalTesch2014,KosMal10,KosMal2013,KosMal2014,KosMalNat2016, CarlMalPos2013},
3-D Schr\"odinger operators with point interactions \cite{MalSchmu2012}, and  elliptic
boundary value problems in domains with compact boundaries \cite{Gru68, Gru09, Gru11},
\cite{BGW2009, BMNW2008},  \cite{Mal10}, \cite{BLL13, BLLLP10}, to the scattering theory
\cite{BehMalNei08, BehMaNei2017}, etc. Especially we mention the works \cite{Gor71,
GG91,MalNei2012} (see also the literature quoted therein) where the
Sturm-Liouville operator with unbounded operator potential was treated
as an operator admitting the tensor structure \eqref{eq:1.2}.

Our goal in this paper is to  apply  the boundary triplet  approach to the
problem of coupling of a quantum system to a reservoir by point interactions. More
precisely, the mathematically rigorous problem is: Given a boundary triplet $\Pi_A =
\{\cH^A,\gG^A_0,\gG^A_1\}$ for $A^*$   one should construct an adequate ("good")
boundary triplet $\Pi_S$ for $S^*$ with $S$ given by \eqref{eq:1.2}  and such that
\bed
S_0 := S^*\upharpoonright\ker(\gG_0^S) = A_0 \otimes I_{\gotT} + I_\gotH \otimes T
\eed
and compute the corresponding Weyl function and
$\gamma$-field.

So starting with a given boundary triplet $\Pi_A$ for $A^*$ a ``good'' candidate for a
boundary triplet $\Pi_S := \{\cH^S,\gG^S_0,\gG^S\}$ for $S^*$ would be
       \be\la{eq:1.8}
\cH^S = \cH^A \otimes \gotT, \quad  \gG^S_0 := \gG^A_0 \otimes I_{\gotT}
\quad \gG^S_1 := \gG^A_1 \otimes I_{\gotT}.
   \ee
This triplet feels   the tensor structure of the problem.  For instance, according to
\eqref{1.5A} and \eqref{eq:1.8}  an  extension  $S' = A' \otimes I_{\gotT} + I_{\gotH}
\otimes T\in \Ext_S$ admits a representation $S' = S_{B'_S}$ with the boundary operator
$B'_S = (B'_S)^*$ having  the tensor form,  $B'_S = B'_A \otimes I_{\gotT}$, where $B'_A
= (B'_A)^*$ is the boundary operator of  $A' = A'^*\in \Ext_A$ in the triplet $\Pi_A$,
i.e. $A' = A_{B'_A }$. In particular, we have $S_0 = A_0 \otimes I_\gotT + I_\gotH \otimes T$.
Hence any  extension $S_{B'} = S_{B'}^*\in \Ext_S$ with the boundary operator $B'$ not
admitting tensor structure, can be regarded as  a Hamiltonian  describing  a point
interaction with a reservoir.

It is shown in \cite{BoiNeiPop2013,MalNei2012}
that the triplet \eqref{eq:1.8} is a boundary triplet for $S^*$ whenever $T$  is bounded.
However this fails for  unbounded $T$,  a case naturally arising  in physical
problems. This case requires new ideas and is much more technically involved.

Let  $\Pi_A = \{\cH^A,\gG^A_0,\gG^A_1\}$ be a boundary triplet for $A^{*}$ and let
$M^A(\cdot)$ and $\gamma^A(\cdot)$  be the corresponding  Weyl function and
$\gamma$-field, respectively.  Using the regularization procedure introduced and
developed in \cite{MalNei2012,KosMal10, CarlMalPos2013} we construct a special boundary
triplet $\Pi_S=  \{\cH^S,\gG^S_0,\gG^S_1\}$  for $S^*$ such that  $S_0 :=
S^*\upharpoonright\ker(\gG^S_0) = A_0 \otimes I_{\gotT} + I_{\gotH} \otimes T$.
Moreover, we show in  Theorem \ref{th:2.5} that the corresponding   $\gamma$-field
$\gga^S (\cdot)$ and Weyl function $M^S(\cdot)$ are given by
  \be\la{Gamma-field_for_S_Intro}
\gga^S (z)f  := \int_\dR \left(\gga^A(z-\gl)\ \frac{1}{\sqrt{\im(M^A(i-\gl))}} \otimes
I_{\gotT}\right) \widehat E_T(d\gl)f,
  \ee
and
  \be\la{WFunc_for_S_Intro}
\begin{split}
M^S(z)f  := \int_\dR \left(L^A(z-\gl,i-\gl)\otimes I_{\gotT}\right)\widehat E_T(d\gl)f
\qquad z \in \dC_\pm.
\end{split}
\ee
where
  \bed\la{L-A_Intro}
L^A (z,\zeta) := \frac{1}{\sqrt{\im(M^A(\zeta))}}(M^A(z)
-\re(M^A(\zeta)))\frac{1}{\sqrt{\im(M^A(\zeta))}},
     \eed
$z \in \dC_\pm$, $\zeta \in \dC_+$, and $\widehat E_T(\cdot) := I_{{\mathcal H}^A} \otimes E_T(\cdot)$, where  $E_T(\cdot)$ is the spectral
measure of $T=T^*$. Here both improper integrals exist for every  $f \in \cH^A \otimes \gotT$.

We apply  formula \eqref{WFunc_for_S_Intro} for $M^S(\cdot)$  to show that for
non-negative $A\ge 0$ and $T \ge 0$  the Friedrichs and Krein extensions of $S\ge 0$ are
given by
\begin{equation}\label{Intro_Fried_Kre_ext}
\widehat S_F =  \widehat A_F \otimes I_{\gotT} + I_\gotH \otimes T,  \qquad  \widehat
S_K = \widehat A_K \otimes I_{\gotT} + I_\gotH \otimes T,
    \end{equation}
where $\widehat A_F$ and $\widehat A_K$ are the Friedrichs and Krein extensions of $A$,
respectively. In turn, we apply these formulas to show that if $T\in\mathcal B(\gotT)$,
then the operator $S$ has LSB-property (each semibounded boundary operator $B$ defines a
semibounded extension $S_B$ of $S$) if and only if the operator $A\ge 0$ has the LSB-property.

This  approach  can be used to propose a model describing rigorously the electron
transport through a quantum dot assisted by photons, a topic which is of great interests
for physicists, cf.
\cite{AbduTangManoGudm2016,Hu1993,KnouJauhMcCoDixoMcEuNazaVaarFoxon1994,PederHolmBuett1994}
etc. In this case we start from  the  operator \eqref{eq:1.1} with $A_0$ being
Sturm-Liouville operator on the line with piece-wise constant potential and unbounded
$T$ given by $T= b^*b,$ where $b^*$ and $b$ are the creation and annihilation operators,
respectively. We define $A$ as a restriction of  $A_0$ to the domain $\dom(A) =
W^{2,2}_0(\dR_-)\oplus W^{2,2}_0(\dR_+)$  and then define the operator $S$  by
\eqref{eq:1.2}.   We construct a boundary triplet for $S^*$ feeling  the tensor
structure \eqref{eq:1.2}, and compute the  Weyl function (a special case of
\eqref{WFunc_for_S_Intro}).

An interesting feature of our approach is  to define Hamiltonian describing the
point contact to the reservoir  by means of  a boundary operator.
To this end  we choose the boundary operator to be  well-known Jaynes-Cumming operator
(see  \cite{jaynes1963})
  \bed
C_{JC} = B \otimes I_\gotT+ I_{\cH^A} \otimes T + \gt V_{JC}, \quad \gt \in \dR,
  \eed
The operator  $C_{JC}$ has a physical meaning and is regarded as the Hamiltonian of the
quantum dot. From the mathematical point of view it is interesting to note that the
Jaynes-Cumming Hamiltonian is in fact a Jacobi matrix.

In this connection we mention the papers by Pavlov
\cite{Pavlov1984,Pavlov87,PavlShus88} treating several solvable physical  models in the
framework of extension theory.

In a forthcoming paper we plan to express  the scattering matrix
for a naturally related scattering system by means of the Weyl function,  using results
from \cite{BehMalNei08,BehMaNei2017}.
 Explicit knowledge of the scattering matrix allows to
calculate the current going through the quantum dot using the so-called Landauer-B\"uttiker
formula invented in \cite{Buettiker1985,Landauer1957}, see also
\cite{AschJakPauPil2007,CorNeiWilZag2014} for a mathematically rigorous proof of this
formula. Using this approach our final goal is to compute explicitly the electron and photon current
going through the quantum dot of the proposed model.

The paper is organized as follows. In Section 2 we give a short introduction into the
boundary triplet approach. In particular, we consider the case of a direct sum of
symmetric operators. Section 3 is devoted to operator-spectral integrals needed in the
following. In Section 4 we consider boundary triplets for tensor products. First we
compute explicitly the Weyl function and $\gamma$-field for the triplet \eqref{eq:1.8}
with a bounded $T\in \cB(\gotH)$ by using the functional calculus developed in Section
3. In Section 4.2 we construct an adequate  boundary triplet for $S^*$ assuming $T$ to
be unbounded and prove formulas \eqref{WFunc_for_S_Intro} and
\eqref{Gamma-field_for_S_Intro}. Section 5 is devoted to the case of non-negative
operators $A$ and $T$, a situation typical in physics. In particular, formulas
\eqref{Intro_Fried_Kre_ext} are proved here.  In Section 6  we illustrate  the abstract
results of  Sections 4,5, by typical  physical examples. In particular, we consider
Schr\"odinger and Dirac operators and a reservoir of bosons on half-line and bounded
intervals. Finally, in Section 7 we use the previous examples to propose a simple model
describing a photon assisted electronic transport through a quantum dot.

{\bf Notation.}
Let  $\gotH_1$, $\gotH_2$ be separable Hilbert spaces. By  $\cC(\gotH_1, \gotH_2)$
($\cB(\gotH_1,\gotH_2)$) we denote  the set of closed (bounded) linear operators from
$\gotH_1$ to $\gotH_2$; $\cC(\gotH) := \cC(\gotH, \gotH)$,  $\cB(\gotH) :=
\cB(\gotH,\gotH)$. By $\gotS_p(\gotH)$, $p \in (0, \infty]$, we denote the
Schatten-Neumann ideals of order $p$ on  $\gotH$; in particular, $\gotS_\infty(\gotH)$
is the ideal of  compact operators on $\gotH$. By $\dom(T)$, $\ran(T)$, $\rho(T)$ and
$\gs(T)$ we denote the domain, range resolvent set and spectrum of the operator $T$,
respectively.

\section{Preliminaries}

\subsection{Linear relations}

A linear relation $\Theta$ in $\cH$ is a closed linear subspace
of $\cH \oplus \cH$.
The set of  all linear relations in $\cH$ is
denoted by $\widetilde\cC(\cH)$.
Denote also by $\cC(\cH)$ the set of all closed linear (not necessarily densely defined) operators in $\cH$.
Identifying  each  operator  $T\in \cC(\cH)$ with its graph $\graph(T)$
we  regard  $\cC(\cH)$  as a subset of $\widetilde\cC(\cH)$.

The role of the set $\widetilde\cC(\cH)$ in extension theory becomes
clear from Proposition \ref{prop2.1}.
However, it's role in the operator theory
is substantially   motivated  by the following circumstances: in contrast to
$\cC(\cH)$, the set $\widetilde\cC(\cH)$ is closed with respect to taking
inverse  and adjoint  relations $\Theta^{-1}$ and  $\Theta^*$.
Here $\Theta^{-1}= \{\{g,f\}: \{f,g\}\in \Theta\}$  and
\bed
\Theta^*= \left\{
\begin{pmatrix} k\\k^\prime
\end{pmatrix}: (h^\prime,k)=(h,k^\prime)\,\,\text{for all}\,
\begin{pmatrix} h\\h^\prime\end{pmatrix}
\in\Theta\right\}.
\eed
A linear relation $\Theta$ is called symmetric if
$\Theta\subset\Theta^*$ and self-adjoint if $\Theta=\Theta^*$.

\subsection{Boundary triplets and proper extensions}

Following \cite{GG91} and \cite{DM91} we briefly recall some basic facts on boundary
triplets. Let $S$ be a densely defined closed symmetric operator in a separable Hilbert
space $\gotH$ and  let $n_\pm(S) := \dim(\gotN_{\pm i})$ be its deficiency indices,
where  $\gotN_z := \ker(S^* - z)$, $z \in \dC_\pm$, is the defect subspace of $S$.
\bd\la{def:2.1}
 A closed extension $\wt S$ of $S$ is called proper
if $\dom(S) \subsetneqq \dom(\wt S)  \subsetneqq \dom(S^*)$.
\ed

We denote by  $\Ext_S$ the set of all proper extensions of $S$ completed by the
non-proper extensions $S$ and $S^*$. It is known that any dissipative (accumulative), in
particular symmetric, extension of $S$ is proper.
\bd[cf. \cite{GG91}]\la{def:3.1}
{\rm
A triplet $\Pi_S = \{\cH^S,\Gamma_0^S,\Gamma_1^S\}$,
where $\cH^S$ is an auxiliary
Hilbert space and $\Gamma_0^S, \Gamma_1^S : \dom(S^{*}) \to \cH^S$ are
linear
mappings, is called a boundary triplet for $S^{*}$ if the 'abstract Green's identity'
\be\la{2.1}
(S^*f,g) - (f,S^*g) = (\gG_1^Sf,\gG_0^Sg) - (\gG_0^Sf,\gG_1^Sg), \quad f,g \in \dom(S^*),
\ee
holds  and the mapping $\gG^S := (\Gamma_0^S, \Gamma_1^S)^{t}: \dom(S^{*}) \rightarrow
(\cH^S \oplus \cH^S)^{t}$ is surjective, i.e. $\ran(\gG^S) = (\cH^S \oplus \cH^S)^{t}$.
}
\ed
A boundary triplet $\Pi_S=\{\cH^S,\gG_0^S,\gG_1^S\}$ for $S^*$ always exists
whenever  $n_+(S) = n_-(S)$. Note also that  $n_\pm(S) =
\dim(\cH^S)$ and $\ker(\Gamma_0^S) \cap \ker(\Gamma_1^S)=\dom(S)$.

The linear maps $\gG_j^S : \dom(S^*) \longrightarrow \cH^S$, $j = 0,1$, are neither
bounded nor closable. However, equipping the domain $\dom(S^*)$ with the graph norm
\bed
\|f\|^2_{S^*} := \|S^*f\|^2 + \|f\|^2, \qquad f \in \dom(S^*),
\eed
one obtains a Hilbert space  $\gotH_+(S^*)$.  and regarding

It turns out that the mappings  $\gG_j^S : \dom(S^*) \longrightarrow \cH^S$, $j \in
\{0,1\}$, as the mappings from $\gotH_+(S^*)$ into $\cH^S$ are already  bounded.

 In what follows we denote by $\wh \gG_j^S$   the operator $\gG_j^S$ treated as the
 mapping  $\wh \gG_j^S : \gotH_+(S^*) \longrightarrow \cH^S$, $j \in \{0,1\}$.
If $J_{S^*} : \gotH_+(S^*) \longrightarrow \dom(S^*)$ denotes the embedding operator,
then  $\wh \gG_j^S  = \gG_j^SJ_{S^*}$, $j \in \{0,1\}$.
It follows from Definition \ref{def:3.1} that $\ran(\wh \gG^S) = \cH^S \oplus \cH^S$,
where $\wh \gG^S := (\wh \gG_0^S,\wh \gG_1^S)^t$. Notice that the abstract Green's
identity \eqref{2.1} can be written as
\bed
(S^*J_{S^*}f,J_{S^*}g) - (J_{S^*}f,S^*J_{S^*}g) = (\wh \gG_1^Sf,\wh \gG_0^Sg) - (\wh \gG_0^Sf,\wh \gG_1^Sg),
\eed
$f,g \in \gotH_+(S^*)$.
With any boundary triplet $\Pi_S$ one associates  two canonical
self-adjoint extensions $S_j:=S^*\!\upharpoonright\ker(\gG_j^S)$, $j \in \{0,1\}$.
Conversely, for any extension $S_0=S_0^*\in \Ext_S$ there exists
a (non-unique) boundary triplet $\Pi_S=\{\cH^S,\gG_0^S,\gG_1^S\}$ for $S^*$
such that $S_0:=S^*\!\upharpoonright\ker(\gG_0^S)$.

Using the concept of boundary triplets one can parametrize all
proper extensions of $A$ in the following way.
\begin{proposition}[cf. \cite{DM91, Mal92}]\label{prop2.1}
Let  $\Pi_S=\{\cH^S,\gG_0^S,\gG_1^S\}$  be a
boundary triplet for  $S^*.$  Then the mapping
\be\la{eq:2.2}
\Ext_S \ni \widetilde S \to  \Gamma^S \dom(\widetilde S)
=\{(\Gamma_0^S f,\Gamma_1^S f)^{t} : \  f\in \dom(\widetilde S) \} =:
\Theta \in \widetilde\cC(\cH^S)
\ee
establishes  a bijective correspondence between the sets $\Ext_S$
and  $\widetilde\cC(\cH^S)$. We write  $\wt S = S_\Theta$
if $\wt S$ corresponds to $\Theta$ by \eqref{eq:2.2}.
Moreover, the following holds:

\item[\;\;\rm (i)] $S^*_\Theta= S_{\Theta^*}$, in particular,
  $S^*_\Theta= S_\Theta$ if and only if $\Theta^* = \Theta$.

\item[\;\;\rm (ii)]  $S_\Theta$ is symmetric (self-adjoint) if
  and only if so is  $\Theta$.
\end{proposition}

In particular, $S_j:=S^*\!\upharpoonright\ker(\gG_j^S) =
S_{\Theta_j},\ j\in \{0,1\},$ where $\gT_0:=
\begin{pmatrix} \{0\}\\ \cH^S \end{pmatrix}$ and
$\gT_1 := \begin{pmatrix} \cH^S \\ \{0\} \end{pmatrix} = \graph(\bO)$ where $\bO$
denotes the zero operator in $\cH^S$.
Note also that  the trivial linear relations $\{0\} \times \{0\}$ and $\cH^S \times
\cH^S\in \wt\cC(\cH^S)$ parametrize the extensions  $S$ and $S^*$, respectively, in any
triplet $\Pi_S$.

\subsection{$\gamma$-field and Weyl function}

It is well known that the Weyl function is an important tool in the direct and inverse
spectral theory of Sturm-Liouville operators. In \cite{DM87,DM91} the concept of Weyl
function was generalized to the case of an arbitrary symmetric operator $S$ with $n_+(S)
= n_-(S)\le \infty$. {Following} \cite{DM91},  we briefly recall basic facts on Weyl
functions and $\gamma$-fields, associated with a boundary triplet  $\Pi.$ For further
properties and applications see \cite{BGP07,DM91,DM95}, \cite{Schmuedgen2012} (and
references therein).
\bd[cf. \cite{DM87,DM91}]\label{Weylfunc}
{\em
Let $\Pi^S=\{\cH^S,\gG_0^S,\gG_1^S\}$ be a boundary triplet  for $S^*$ and
$S_0=S^*\!\upharpoonright\ker(\gG_0^S)$. The operator valued
functions $\gamma^S(\cdot) : \rho(S_0) \rightarrow  \cB(\cH^S,\gotH)$ and
$M^S(\cdot) :\ \rho(S_0)\rightarrow  \cB(\cH^S)$ defined by
\be\label{2.3A}
\gamma^S(z):=\bigl(\Gamma_0^S\!\upharpoonright\gotN_z\bigr)^{-1}
\qquad\text{and}\qquad M^S(z):=\Gamma_1^S\gamma^S(z), \quad
z\in\rho(S_0), \ee
are called the $\gamma$-field and the  Weyl function,
respectively, corresponding to the boundary triplet $\Pi_S.$
}
\ed
Clearly, the Weyl function can equivalently be defined by
\bed
M^S(z)\Gamma_0^S f_z = \Gamma_1^Sf_z,\qquad  f_z \in \gotN_z, \quad z\in\rho(S_0).
\eed
The $\gamma$-field $\gamma^S(\cdot)$ and the Weyl function $M^S(\cdot)$ in \eqref{2.3A}
are well defined. Moreover, both $\gamma^S(\cdot)$ and $M^S(\cdot)$ are holomorphic on
$\rho(S_0)$ and satisfy the following relations
\be\la{2.5}
\gamma^S(z)=\bigl(I+(z-\zeta)(S_0-z)^{-1}\bigr)\gamma^S(\zeta),
\qquad z,\zeta\in\rho(S_0),
\ee
and
\begin{equation}\label{mlambda}
M^S(z)-M^S(\zeta)^*=(z-\overline\zeta)\gamma^S(\zeta)^*\gamma^S(z), \qquad
z,\zeta\in\rho(S_0),
\end{equation}
hold. Identity \eqref{mlambda} yields that
$M^S(\cdot)$  is  an $\cB(\cH^S)$-valued Nevanlinna function ($M^S(\cdot)\in R[\cH^S]$), i.e. $M^S(\cdot)$ is an $\cB(\cH^S)$-valued holomorphic function on
$\dC_\pm$ satisfying
\bed
M^S(z)=M^S(\overline z)^*\qquad\text{and}\qquad
\frac{\im(M^S(z))}{\im(z)}\geq 0, \qquad
z\in \dC_\pm.
\eed
It follows also from \eqref{mlambda} that
$0\in \rho(\im(M^S(z)))$
for all $z\in\dC_\pm$.

Being an $R[\cH^S]$-function the Weyl function $M^S(\cdot)$ admits an integral
representation
    \begin{equation}\label{W-F+Integral_rep-n}
M^S(z)= C_0 + \int_{\dR}\left(\frac{1}{t-z} - \frac{t}{1+t^2}\right)\, d\Sigma_S(t),
\qquad \int_{\dR}\frac{d\Sigma_S(t)}{1+t^2}\in \cB(\cH^S),
    \end{equation}
where $C_0 = C_0^*$ and $\Sigma_S(\cdot)$ is a left continuous $(\Sigma_S(t)=
\Sigma_S(t-0))$ monotone operator-valued function. Emphasize that the linear term $C_1z$
is missing in  \eqref{W-F+Integral_rep-n} because the operator $A$  is densely defined
(see  \cite{DM91}).

A Weyl function $M^S(\cdot)$ is said to be of scalar type if there exists a scalar
Nevanlinna function $m^S(\cdot)$ such that the representation
\bed
M^S(z) = m^S(z)I_{\cH^S}, \quad z \in \dC_\pm,
\eed
holds where $I_{\cH^S}$ is the identity operator in $\cH^S$, see
\cite{ABMN05}. Obviously, $M^S(\cdot)$ is of scalar type if $n_\pm(A) = 1$.

Next we extract from \eqref{W-F+Integral_rep-n}  lower  and upper bounds for
$\im(M^S(i-\gl))$ which will be useful in the sequel.  It follows from
\eqref{W-F+Integral_rep-n}  that
\begin{equation}\label{W-F_Imag_part_Int_rep}
\im\bigl(M^S(i-\lambda)\bigr)=\int_{\dR}\frac{d\Sigma_S(t)}{(t-\lambda)^2+1}, \quad
\lambda\in\dR
 \end{equation}
Note that with certain positive constants $C_1, C_2>0$ the following estimate holds
    \begin{equation*}
\frac{C_1}{1 + |\lambda|^2}\le\frac{1+t^2}{(t-\lambda)^2+1}\le C_2(1+|\lambda|^2),
\qquad \lambda\in\dR.
    \end{equation*}
Combining these  estimates with the identity  $\im M(i) = \int_{\dR}(1+t^2)^{-1}{d\Sigma_S(t)}$ one derives  from \eqref{W-F_Imag_part_Int_rep}  that
      \begin{equation}\label{joh4}
C_1(1+|\lambda|^2)^{-1} \im M(i)\le \im M^S(i-\lambda) \le  C_2(1+|\lambda|^2) \im M(i),
\;\; \lambda\in\dR.
      \end{equation}
Emphasize that since the proof of estimates  \eqref{joh4} is based only on integral
representation \eqref{W-F+Integral_rep-n}, these estimates are valid for any
$R[\cH^S]$-function not necessarily being a  Weyl function.

\subsection{Krein-type formula for resolvents}\la{sec.II.1.4}

Let $\Pi_S=\{\cH^S,\gG_0^S,\gG_1^S\}$ be a boundary triplet for $S^*,$ and
$M^S(\cdot)$ and  $\gamma^S(\cdot)$ the corresponding Weyl function
and   $\gamma$-field, respectively.
For any proper (not necessarily self-adjoint)  extension
${\widetilde S}_{\Theta}\in \Ext_S$ with non-empty resolvent set
$\rho({\widetilde S}_{\Theta})$ the following Krein-type formula
holds (cf. \cite{DM87,DM91,DM95})
\be\label{2.30}
(S_\Theta - z)^{-1} - (S_0 - z)^{-1} = \gamma^S(z) (\Theta -
M^S(z))^{-1} (\gamma^{S}({\overline z}))^*,
\ee
$z\in \rho(S_0)\cap\rho(S_\Theta)$. Formula  \eqref{2.30} extends  the known Krein formula for canonical resolvents to the
case of any $S_\Theta\in \Ext_S$ with $\rho(S_\Theta)\not = \emptyset.$  Moreover, due
to relations   \eqref{eq:2.2} and \eqref{2.3A} all objects in formula  \eqref{2.30} are
expressed by means of  the boundary triplet $\Pi_S$. We emphasize, that this connection
makes it possible to apply the Krein-type formula \eqref{2.30}  to boundary value
problems.

\subsection{Normalized boundary triplets}

Let $S_n$ be a  densely defined closed symmetric operator  in $\gotH_n$, ${n\in\dZ}$,
and let $S:= \bigoplus_{n \in \dZ}S_n$. Clearly,
   \bed
   \begin{split}
S^* &= \bigoplus_{n \in \dZ}S^*_n,\\
\dom(S^*) &=  \left\{f = \mathlarger{\mathlarger{\mathlarger{\oplus}}}_{\dZ} f_n\in
\mathfrak{H}:\  f_n\in\dom(S^*_n),\ \  \sum^{\infty}_{n=1}\|S^*_n
 f_n\|^2<\infty\right\}.
 \end{split}
   \eed
Let $\Pi_{S_n} = \{\cH^{S_n},\Gamma^{S_n}_0,\Gamma^{S_n}_1\}$  be  a boundary triplet
for $S^*_n$,  $n \in \dZ$. Define  mappings $\Gamma_j^S$, $j\in\{0,1\},$  by setting
   \begin{equation}\label{III.1_02}
   \begin{split}
\Gamma^{S}_j  &:= \bigoplus_{n\in\dZ}  \Gamma^{S_n}_j,\\
\dom(\gG^{S}_j) &:= \left\{\mathlarger{\mathlarger{\mathlarger{\oplus}}}_{n\in\dZ}f_n
\in \dom(S^*): \sum_{n\in\dZ}\|\gG^{S_n}_jf_n\|^2 < \infty \right\}.
\end{split}
   \end{equation}
       \begin{definition}\label{def_III.1_01}
Let  $\Gamma_j^S$ be given  by \eqref{III.1_02} and $\cH^S :=
\bigoplus_{n\in\dZ}\cH^{S_n}$. A collection $\Pi_S = \{\cH^S,\gG^S_0,\gG^S_1\}$ is
called a \emph{direct sum of boundary triplets} and is  assigned as $\Pi_S =
\bigoplus_{n\in\dZ}\Pi_{S_n}$.
     \end{definition}

It was first discovered by A. Kochubei \cite{Koch79} that the direct sum $\bigoplus\Pi_n$ of
boundary triplets $\Pi_n$ is not a boundary triplet in general. Later on simple
examples were constructed in \cite{MalNei2012}, \cite{KosMal10}, \cite{CarlMalPos2013}.
Moreover, it was shown in \cite[Theorem 3.2]{KosMal10} that  $\Pi_S$ is only a
generalized boundary triplet (a boundary relation in the sense of \cite{DHMS06}).
Moreover,  according to \cite{DHMS12}  $\Pi_S$ is a so called ES-generalized boundary
triplet for $S^*$, since the operator  $S_0 := S^*\upharpoonright\ker(\gG^S_0)$ is
essentially self-adjoint.

The reason is that the domain $\dom(\gG^S_j)$, $j\in\{0,1\},$ might be narrower  than
$\dom(S^*)$
and the range of the mapping $\gG^S := (\gG^S_0, \gG^S_1)^t: \dom(S^{*}) \rightarrow
(\cH^S \oplus \cH^S)^t$ might be a proper subset of $(\cH^S \oplus \cH^S)^t$.
Nevertheless, $\dom(\gG^S_j)$, $j \in \{0,1\}$, is always dense in $\gotH_+(S^*)$  and
its range $\ran(\gG^S)$ is dense in $(\cH^S \oplus \cH^S)^t$.
Moreover, by \cite[Proposition 5.3]{DHMS06}, $\Pi_S$
is a boundary triplet whenever  $\ran(\gG^S) = (\cH^S \oplus \cH^S)^t$.  Besides, in
accordance with \cite[Proposition 3.8]{KosMal10} the conditions
  \be\label{Converg_of_Gamma-j}
\sum_{n\in\dZ}\|\gG^{S_n}_jf_n\|^2 < \infty, \quad f = \mathlarger{\mathlarger{\mathlarger{\oplus}}}_{n\in\dZ}f_n \in
\dom(S^*), \quad  j \in\{0,1\},
  \ee
imply that  $\Pi_S = \bigoplus_{n\in\dZ}\Pi_{S_n}$ is an ordinary  boundary triplet,
while the sole first condition in  \eqref{Converg_of_Gamma-j}  (with $j=0$) ensures only
that  $\Pi_S$ is a $B$-generalized boundary triplet in the sense of \cite{DM95},
\cite{DHMS12}.  Moreover,  according to \cite{DHMS12}  $\Pi_S$ is a so called
ES-generalized boundary triplet for $S^*$, since the operator  $S_0 :=
S^*\upharpoonright\ker(\gG^S_0)$ is essentially self-adjoint.

A regularization procedure described below was first proposed  in \cite{MalNei2012}
and has been applied  to construct a boundary triplet for Sturm-Liouville operators
    \bed
-d^2/dx^2 \otimes I_{\gotT}+I_{\gotH} \otimes T, \qquad  \gotH = L^2(\dR_+;\gotT)
= L^2(\dR_+)\otimes \gotT,
    \eed
with unbounded potential $T=T^*\in\cC(\gotT)$. Further generalizations of regularization
procedures as well as applications to Schr{\"o}dinger and Dirac operators with
$\delta$-interactions  were obtained in  \cite{KosMal10} and  \cite{CarlMalPos2013},
respectively.

Let $\Pi_S = \{\cH^S,\gG^S_0,\gG^S_1\}$ be a boundary triplet for $S^*$ with Weyl
function $M^S(\cdot)$. We call $\Pi_S$ a normalized boundary triplet for $S^*$ if the
condition $M^S(i) = iI_{\cH^S}$ is satisfied.
  \bl[\cite{MalNei2012}]\la{II.5}
Let $\Pi_S = \{\cH^S,\gG^S_0,\gG^S_1\}$ be a boundary triplet for $S^*$, let
$\gga^S(\cdot)$  and $M^S(\cdot)$ be the  $\gga(\cdot)$-field  and Weyl function,
respectively. Let $R_S := \sqrt{\im(M^S(i))}$ and $Q_S := \re(M^S(i))$. Then $\wt \Pi_S
= \{\wt \cH^S,\wt \gG^S_0,\wt \gG^S_1\}$, where
      \be\la{eq:2.10}
\wt \cH^S := \cH^S, \quad \wt \gG^S_0 := R_S\gG^S_0 \quad \mbox{and} \quad \wt \gG^S_1
:= R^{-1}_S(\gG^S_1 - Q_S\gG^S_0),
   \ee
is a normalized boundary triplet for $S^*$ such that
\bed
S_0 := S^*\upharpoonright\ker(\gG^S_0) = S^*\upharpoonright\ker(\wt \gG^S_0). \eed
The  $\gamma$-field $\wt \gga^S(\cdot)$ and  Weyl function $\wt M^S(\cdot)$
corresponding to the triplet $\wt \Pi_S$ are given by
\bed
\wt \gga^S(z) = \gga^S(z)R^{-1}_S \;\; \mbox{and} \;\; \wt M^S(z) =
R^{-1}_S(M^S(z) - Q_S)R^{-1}_S, \;\; z \in \dC_\pm.
  \eed

\el
Lemma \ref{II.5} shows that with any boundary triplet one can associate a normalized
boundary triplet such that $S_0$ remains unchanged.
The following theorem presents  a regularization procedure for direct sums  $\Pi_S =
\bigoplus_{n\in\dZ}\Pi_{S_n}$  to define an ordinary boundary triplet.
 \bt[Theorem 3.3, \;\cite{MalNei2012}]\la{th:2.6}
Let $S_n$ be a  densely defined closed symmetric operator in $\gotH_n$, ${n\in\dZ}$, and
$S:= \oplus_{n \in \dZ}S_n$. Let  $\Pi_{S_n} =
\{\cH^{s_n},\Gamma_{0}^{S_n},\Gamma_{1}^{S_n}\}$ be a boundary triplet for $S^*_n$,
$S_{0n} := S^*_n\upharpoonright\ker(\Gamma_{0}^{S_n})$, $n \in \dZ$, and let
$\gga^{S_n}(\cdot)$ and  $M^{S_n}(\cdot)$ be the corresponding  $\gamma$-field  and Weyl
function, respectively.  Finally, let  $R_{S_n} := \sqrt{\im(M^{S_n}(i))}$ and $Q_{S_n}
:= \re(M^{S_n}(i))$, $n \in \dZ$. Then the triplet $\wt\Pi_S =
\{\wt\cH^S,\wt\gG^S_0,\wt\gG^S_1\}$ with
  \be\label{2.16_reg_triplet}
    \begin{split}
\wt\cH^S &:= \bigoplus_{n\in\dN} \cH^{S_n}, \quad  \wt\gG^{S}_0 :=
\bigoplus_{n\in\dZ}R_{S_n}\gG^{S_n}_0, \\
\wt\gG^S_1 &:= \bigoplus_{n \in \dZ}
R^{-1}_{S_n}\left(\gG^{S_n}_1 - Q_{S_n}\gG^{S_n}_0\right),
\end{split}
  \ee
is  a (normalized) boundary triplet for $S^*$ satisfying
  \bed
\wt S_0 = S^*\upharpoonright\ker(\wt\gG^S_0) = \bigoplus_{n\in\dZ}{\widetilde S}_{0n} =
\bigoplus_{n\in\dZ}S_{0n}, \qquad \widetilde S_{0n} =
S^*_n\upharpoonright\ker({\widetilde\gG}^{S_n}_0).
  \eed
Moreover, the   $\gamma$-field $\wt \gga^S(\cdot)$ and  Weyl function $\wt M^S(\cdot)$
corresponding to  $\wt \Pi_S$ are given by
  \be\label{W-fun_and_gam-field_for_dir_sum}
  \begin{split}
 \wt \gga^S(z) &= \bigoplus_{n\in\dZ} \gga^{S_n}(z)R^{-1}_{S_n} \quad \mbox{and} \quad\\
\wt M^S(z) &= \bigoplus_{n\in\dZ}R^{-1}_{S_n}\left(M^{S_n}(z) -
Q_{S_n}\right)R^{-1}_{S_n},\quad z \in \dC_\pm.
\end{split}
   \ee
  \end{theorem}
Next we assume that the operator $S = \bigoplus^{\infty}_{n=1}S_n$ has a regular real
point, i.e., there exists $a={\overline a}\in{\hat\rho}(A)$. The latter is equivalent to
the existence of $\varepsilon> 0$ such that
   \begin{equation}\label{III.2.2_01}
(a-\varepsilon, a + \varepsilon) \subset \cap^{\infty}_{n=1} {\widehat\rho}(S_n).
       \end{equation}
Emphasize that  condition $a\in\cap^{\infty}_{n=1}{\widehat\rho}(S_n)$ is not sufficient
 for the inclusion ${a\in\widehat\rho}(A).$

It is known (see e.g. \cite{K47}, \cite{AG81}) that under condition \eqref{III.2.2_01}
for every $k\in{\dN}$ there exists a self-adjoint extension ${\widetilde S}_{k} =
{\widetilde S}^*_{k}$ of $S_k$ preserving the gap $(a-\varepsilon, a+\varepsilon)$. The
latter amounts to saying that the Weyl function of the pair $\{S_k, {\widetilde
S}_{k}\}$ is regular within the gap $(a-\varepsilon, a+\varepsilon)$.

For operators $S = \bigoplus_{n=1}^\infty S_n$ satisfying \eqref{III.2.2_01} we complete
Theorem \ref{th:2.6}   by presenting a regularization procedure for $\Pi =
\bigoplus_{n=1}^\infty\Pi_n$ leading to a boundary triplet (cf. \cite[Theorem
3.13]{KosMal10}, \cite[Theorem 2.12 and Corollary 2.13]{CarlMalPos2013}). In
applications to symmetric operators with a gap this regularization  is  more appropriate
and simpler than the one described in Theorem \ref{th:2.6}.
   \begin{proposition}[\cite{KosMal10, CarlMalPos2013}] \label{cor_III.2.2_02}
Let $\{S_n\}^{\infty}_{n=1}$ be a sequence of symmetric operators satisfying
\eqref{III.2.2_01}. Let also ${\Pi}_n = \{\cH_n, {\Gamma}^{(n)}_0, {\Gamma}^{(n)}_1\}$
be a boundary triplet for $S^*_n$ such that $(a-\varepsilon,
a+\varepsilon)\subset\rho(S_{0n})$, $S_{0n}=
S_n^*\upharpoonright\ker({\Gamma}^{(n)}_0)$. Let also $\gga^{S_n}(\cdot)$ and
$M_n(\cdot) := M^{S_n}(\cdot)$ be the corresponding $\gamma$-field and Weyl function,
respectively. Assume also that for some operators $R_n$  such that $R_n, R^{-1}_n
\in \cB(\cH_n)$, the following conditions are satisfied
       \bed
       \begin{split}
&\sup_n \|R_n^{-1}({M}'_n(a))(R_n^{-1})^*\|_{\cH_n} < \infty \quad \text{and}\quad \\
&\sup_n\|R_n^*({M}'_n(a))^{-1}R_n\|_{\cH_n} < \infty, \quad n\in\dN.
\end{split}
       \eed
Then the direct sum  $\wt \Pi_S = \bigoplus_{n=1}^\infty {\wt\Pi}_n$ of boundary
triplets where
    \begin{equation}\label{III.2.2_08}
    \begin{split}
{\wt\Pi}_n &=\{\cH_n, \wt\Gamma^{(n)}_0, \wt\Gamma^{(n)}_1\},\\
\wt\Gamma^{(n)}_0 &:= R_n{\Gamma}^{(n)}_0 \quad \mbox{and} \quad \wt\Gamma^{(n)}_1:=
(R^{-1}_n)^*\bigl({\Gamma}^{(n)}_1 -
  M_n(a){\Gamma}^{(n)}_0\bigr),
  \end{split}
     \end{equation}
forms a boundary triplet for $S^* = \bigoplus_{n=1}^\infty S_n^*$.

Moreover, the  corresponding  $\gamma$-field $\wt \gga^S(\cdot)$ and  Weyl function $\wt
M^S(\cdot)$  are given by
  \bed
  \begin{split}
 \wt \gga^S(z) &= \bigoplus_{n\in\dZ} \gga^{S_n}(z)R^{-1}_n \quad \mbox{and} \quad\\
\wt M^S(z) &= \bigoplus_{n\in\dZ}(R^{-1}_{n})^*\left(M_{n}(z) -
M_{n}(a)\right)R^{-1}_{n},\quad z \in \dC_\pm.
\end{split}
   \eed
In particular one can set $R_{n} = \sqrt{{M}'_n(a)}$, $n\in \dN$.
    \end{proposition}
Emphasize  that  $M'_n(a)$ is a positive definite operator whenever $a\in \rho(S_{0n})$.

\section{Operator-spectral integrals}\la{App.I}

Let $F(\cdot)$ be an orthogonal operator measure with compact support $\supp(F)
\subseteq \gD:= [a,b)$, $-\infty < a < b < \infty$, and with values in $\cB(\cH)$. Further,
let $\gO(\cdot): [a,b) \longrightarrow \cB(\cH,\cH_1)$ be an operator-valued function. We
consider partitions $\cZ$ of $[a,b)$ of the form $[a,b) = [\gl_0,\gl_1) \cup
[\gl_1,\gl_2) \cup \ldots \cup [\gl_{n-1},\gl_n)$, $\gl_0 = a$, $\gl_n = b$ and set
$\gD_m := [\gl_{m-1},\gl_m)$, $m =1,2,\ldots,n$.
 Thus $[a,b) = \bigcup^n_{m=1}\gD_m$ and the $\gD_m$ are pairwise disjoint. Let
$|\cZ| := \max_{m = 1,2,\ldots,n}|\gD_m|$, where $|\gD_m| := \gl_m -\gl_{m-1}$.
We define the operator $\gS_\cZ\gO$ by
\bed
\gS_\cZ\gO = \sum^n_{m=1}\gO(x_m)F(\gD_m), \qquad x_m \in \gD_m.
\eed
The sum $\gS_\cZ\gO$ is called the Riemann-Stieltjes sum of $\gO(\cdot)$ with respect to the operator measure $F(\cdot)$. If there is an operator
$\gS_0 \in \cB(\cH,\cH_1)$ such that $\lim_{|\cZ|\to 0}\|\gS_\cZ\gO - \gS_0\|=0$ independent of the special choice of
$\cZ$ and $\{x_m\}^n_{m=1}$, then
$\gS_0$ is called the operator spectral integral of $\gO(\cdot)$ with respect to $F(\cdot)$ and is denoted by
   \begin{equation}\label{Bochner_Int-l}
\gS_0 =:\int_{\gD} \gO(\gl)F(d\gl).
  \end{equation}
Obviously, in a similar way one can define for operator-valued functions $\gO:\gD\longrightarrow \cB(\cH_1,\cH)$
the operator spectral integral $\int_{\gD}F(d\gl)\gO(\gl)$ as the limit of the
Riemann-Stieltjes sums $\sum_m F(\gD_m)\gO(x_m)$.
It is clear that the operator spectral integral is linear with respect to $\gO(\cdot)$. If $B$ is a bounded operator, then
\bed
B\int_{\gD} \gO(\gl)F(d\gl) = \int_{\gD} B\gO(\gl)F(d\gl).
\eed
\begin{definition}
The operator-valued mapping $\gO:[a,b)\longrightarrow \cB(\cH)$ will be  called
$F$-admissible, if the integral $\int_{\gD} \gO(\gl)F(d\gl)$ exists and
   \begin{equation}\label{commut-relation}
\gO(\gl)F(\gd) = F(\gd) \gO(\gl) F(\gd),\qquad \gd\in \cB([a,b)),\quad  \gl \in \gD.
  \end{equation}
  \end{definition}
\begin{proposition}\la{prop:A.-2}
Let $\gO:[a,b)\longrightarrow \cB(\cH)$ be $F$-admissible, $\gO_1:[a,b)\longrightarrow \cB(\cH,\cH_1)$, and assume that
$\int_{\gD} \gO_1(\gl)F(d\gl)$ exists. Then
$\int_{\gD} \gO_1(\gl)\gO(\gl)F(d\gl)$ exists and
  \bed
\int_{\gD} \gO_1(\gl)\gO(\gl)F(d\gl) = \int_{\gD} \gO_1(\gd)F(d\gd)\,\int_{\gD}
\gO(\mu)F(d\mu).
  \eed
\end{proposition}
\begin{proof}
It is easily seen that
  \bed
\gS_\cZ\gO_1\gS_\cZ\gO \longrightarrow  \int_{\gD} \gO_1(\gl)F(d\gl)\,\int_{\gD}
\gO(\mu)F(d\mu)\qquad\text{as}\qquad \vert \cZ \vert \longrightarrow 0.
  \eed
On the other hand, since the measure $F(\cdot)$  is  orthogonal, $F(\gD_j)F(\gD_{k}) =
F(\gD_j) \delta_{jk}$, $j,k\in \{1,\ldots, m\}$. Combining these relations with the
$F$-admissibility of   $\gO$ yields
   \bed
\begin{split}
\gS_\cZ\gO_1\gS_\cZ\gO &= \sum^n_{m,m'=1}\gO_1(x_m)F(\gD_m)\gO(x_{m'})F(\gD_{m'}) \\
 &= \sum^n_{m,m'=1}\gO_1(x_m)F(\gD_m)F(\gD_{m'}) \gO(x_{m'}) F(\gD_{m'})  \\
&= \sum^n_{m=1}\gO_1(x_m)F(\gD_m) \gO(x_{m}) F(\gD_{m})\\
& =\sum^n_{m=1}\gO_1(x_m)\gO(x_m)F(\gD_m) \longrightarrow  \int_{\gD} \gO_1(\gl)\gO(\gl)F(d\gl)
 \end{split}
    \eed
as $\vert \cZ \vert \longrightarrow 0$. Combining both relations completes the proof.
  \end{proof}

In what follows  we assume that  $\cH = \cH_1$.
   \begin{proposition}\la{prop:A.-1}
Let $X:[a,b)\longrightarrow \cB(\cH)$ be an $F$-admissible  function, and
assume, in addition, that there exist real numbers $c_1,c_2$, such that
$X(\gl)$ is self-adjoint and $c_1\le X(\gl) \le c_2$,  $\gl \in \gD$.
Let  $\varphi\in
C[c_1,c_2]$.
Then the following
holds:

\item[\;\;\rm (i)] The operator $\widehat{X}:= \int_{\gD} X(\gl) F(d\gl)$ is self-adjoint and
satisfies $c_1\le \widehat{X} \le c_2$;

\item[\;\;\rm (ii)]  The following estimate holds
$\parallel \varphi(\widehat{X}) \parallel \le \parallel \varphi \parallel_{\infty}$

\item[\;\;\rm (iii)] The operator-valued function $\varphi(X(\cdot))$ is $F$-admissible and
    \begin{equation}\label{expectation_of_fun-n}
\int_{\gD} \varphi(X(\gl)) F(d\gl) = \varphi(\widehat{X}).
    \end{equation}
\end{proposition}
\begin{proof}
{\;\;(i)} Let $\cZ$ be any partition as above. Then for any  $h\in \cH$ one gets
  \bed
\langle \gS_{\cZ} X h, h \rangle =  \sum_{m=1}^n \langle F(\gD_m) X(x_m) F(\gD_m) h, h
\rangle \ge  \sum_{m=1}^n c_1 \parallel F(\gD_m) h \parallel^2.
  \eed
Thus $\langle \gS_{\cZ}h,h \rangle \in \dR$ and $\langle \gS_{\cZ}h,h \rangle \ge c_1
\parallel h \parallel^2$. In the same way one shows that $\langle \gS_{\cZ}h, h\rangle
\le c_2 \parallel h \parallel^2$. By passing to the limit, as $\vert \cZ \vert
\longrightarrow 0$, we get that $\langle \hat{X}h, h \rangle  \in [c_1 \parallel h
\parallel^2 , c_2 \parallel h \parallel^2]$ for every $h\in \cH$, and (i)
is proved.

{\;\;(ii)} By the functional calculus, both inequalities  $\parallel\varphi(\widehat{X})
\parallel \le
\parallel \varphi \parallel_{\infty}$ and $\parallel \varphi(X(\gl)) \parallel \le
\parallel \varphi \parallel_{\infty}$ hold for every $\gl \in \gD$ and each  continuous
function $\varphi\in C[c_1,c_2]$.

{\;\;(iii)} First we prove, by induction, the assertion (iii) in the special case, when
$\varphi(\gl) = \gl^n$. By the assumption,  the assertion is true for $n=1$. Suppose
that it is true for $n =k$. Let us prove it for $n = k+1$. One has
  \begin{equation}\label{F-admis-le_powers}
  \begin{split}
X^{k+1}&(\gl)F(\gd) =  X(\gl) F(\gd) X^k(\gl) F(\gd) \\
&= F(\gd) X(\gl)F(\gd) \cdot F(\gd)
X^k(\gl) F(\gd)
= F(\gd) X^{k+1}(\gl) F(\gd),\
\end{split}
  \end{equation}
$\gl \in \gD, \gd\in \cB(\gD)$.
Therefore Proposition \ref{prop:A.-2} ensures that the integral  $\int_{\gD} X^{k +
1}(\gl) F(d\gl)$ exists and
  \bed
\int_{\gD} X^{k + 1} (\gl) F(d\gl) = \widehat{X}^{k + 1}.
    \eed
By linearity, these equalities are  easily extended  for polynomials in $\lambda$.

Let $\varphi$ be a continuous function,  $\varphi\in C[c_1,c_2]$.
By the Weierstrass Theorem, there exists  a sequence $\{p_k\}_1^\infty$ of polynomials
approaching  $\varphi$ in $ C[c_1,c_2]$. In accordance with  the functional calculus for
self-adjoint operators,
\bed
\parallel \varphi(\widehat{X}) - p_k(\widehat{X}) \parallel\  \le \  \parallel \varphi-p_k \parallel_{\infty} \longrightarrow
0\quad \text{as}\quad k\to \infty.
    \eed
and
   \bed
\parallel \varphi(X(\gl)) - p_k(X(\gl)) \parallel \  \le \  \parallel \varphi-p_k \parallel_{\infty} \longrightarrow 0
\quad \text{as}\quad k\to \infty, \qquad \gl \in \gD.
   \eed
Combining  this relation with equalities \eqref{F-admis-le_powers} for polynomials we
obtain that  for every $\gl\in \gD$ and any  Borel  subset $\gd\subset \gD$ the
following  holds
  \bed
  \begin{split}
\varphi(X(\gl)) F(\gd) &= \lim_{k\longrightarrow \infty} p_k(X(\gl)) F(\gd)\\
 &=
\lim_{k\longrightarrow \infty} F(\gd) p_k(X(\gl)) F(\gd) = F(\gd)\varphi(X(\gl)) F(\gd),
\quad \gl \in \gD.
\end{split}
  \eed
This relation means that the function $\varphi(X(\cdot))$ satisfies
commutation relation \eqref{commut-relation}.  To prove its $F$-admissibility it remains
to prove the existence of the integral  $\int_{\gD} \varphi(X(\gl)) F(d\gl)$.  We prove
it together  with  relation \eqref{expectation_of_fun-n}. To this end  for each
partition $\cZ$  of $[a,b]$ we  prove the following estimate
  $$
\parallel \gS_{\cZ} (\varphi(X) - p_k(X))
\parallel \ \le \  \parallel \varphi- p_k \parallel_{\infty}.
  $$
Since the measure $F$ is orthogonal, one gets
  \bed
  \begin{split}
\parallel \gS_{\cZ} &\varphi(X)f - \gS_{\cZ} p_k(X)f \parallel^2 =
\left\| \sum_{m=1}^n (\varphi-p_k)(X) (x_m) F(\gD_m) f \right\|^2 \\
&= \sum_{m=1}^n \parallel (\varphi-p_k)(X)(x_m) F(\gD_m)f  \parallel^2 \\
&\le \sum_{m=1}^n \parallel \varphi - p_k \parallel_{\infty}^2\parallel F(\gD_m) f \parallel^2 =
\parallel \varphi - p_k \parallel_{\infty}^2 \parallel f \parallel^2.
\end{split}
\eed
Let $\{\cZ_j\}$ be a sequence of partitions satisfying  $\vert \cZ_j \vert
\longrightarrow 0$  and  let $\varepsilon >0$. Choose $k$ such that $\parallel \varphi -
p_k \parallel_{\infty} < \varepsilon$ and  $j_0$ such that
  \bed
\parallel \gS_{\cZ_j}  p_k(X) - \gS_{\cZ_{j'}} p_k(X) \parallel < \varepsilon,\qquad j,j'\ge j_0,
  \eed
and hence $\parallel \gS_{\cZ_j}  \varphi(X) - \gS_{\cZ_{j'}} \varphi(X) \parallel <
3\varepsilon$ for all $j,j'\ge j_0$. Thus the limit $\lim_{j\longrightarrow \infty}
\gS_{\cZ_j}  \varphi(X)$ exists, and
  \begin{equation}\label{3.5-int-sum-estimate}
  \begin{split}
\parallel &\lim_{j\longrightarrow \infty} \gS_{\cZ_j}  \varphi(X) - \varphi(\widehat{X}) \parallel
\\
&\le
\parallel \lim_{j\longrightarrow \infty} \gS_{\cZ_j}  (\varphi(X) - p_k(X)) \parallel +
\parallel p_k(\widehat{X}) - \varphi(\widehat{X}) \parallel < 2 \varepsilon.
\end{split}
   \end{equation}
Since  $\varepsilon>0$ is arbitrary, this inequality ensures  the existence of the
integral $\int_{\gD} \varphi(X(\gl)) F(d\gl)$, thus  proves $F$-admissibility of
$\varphi(X(\cdot))$.   Moreover, estimate \eqref{3.5-int-sum-estimate} proves  equality
\eqref{expectation_of_fun-n}.
  \end{proof}
 Denote by  $dm$ the  Lebesgue measure on $\dR$.
  \begin{corollary}\label{Cor_Int_with_Lip_fun-n}
Assume that $\gO(\cdot) = \gO(\cdot)^*$  is  a self-adjoint $\cB(\cH)$-valued Lipschitz
function  in $\gD = [a,b)$ and  $c_1 \le \gO(\cdot)\le c_2$. Assume also that $F(\cdot)$
is  a spectral measure in $\cH$ with compact support, $\supp(F) \subseteq \gD:= [a,b)$,
and $\varphi\in C[c_1,c_2]$.  If in addition, commutation relation
\eqref{commut-relation} holds, then  the operator-valued function $\varphi(\gO(\cdot))$
is $F$-admissible and
    \begin{equation}\label{expectation_of_fun-n_2}
\int_{\gD} \varphi(\gO(\gl)) F(d\gl) = \varphi(\widehat{X}).
    \end{equation}
  \end{corollary}
\begin{proof}
It is shown in \cite[Lemma 7.2]{AdaMen96}   that the integral \eqref{Bochner_Int-l}
exists whenever $\gO(\cdot)$ is  Lipschitz function.  By Proposition
\ref{prop:A.-1}(iii) $\varphi(\gO(\cdot))$ is $F$-admissible and equality
\eqref{expectation_of_fun-n_2} holds.
  \end{proof}
  \begin{corollary}
Let  $\gO(\cdot) = \gO(\cdot)^*$  be  differentiable with respect to the operator norm
$m$-almost everywhere in $\gD = [a,b)$, $c_1 \le \gO(\cdot)\le c_2$,   and let
$\gO(\cdot)$ be  expressed by means of its derivative $\gO'(\cdot)$ via the Bochner
integral  on $[a,b)$, i.e.
  \begin{equation}\label{integ_repr-n_of_Omega}
\gO(\gl) = \gO(a) + \int_a^\gl \gO'(x)dx, \quad \gl \in [a,b).
   \end{equation}
Assume also that  $F(\cdot)$ is  a spectral measure in $\cH$ with compact support,
$\supp(F) \subseteq \gD:= [a,b)$,  and $\varphi\in C[c_1,c_2]$. Assume also that
commutation relation \eqref{commut-relation} holds. Then the operator-valued function
$\varphi(\gO(\cdot))$ is $F$-admissible and
    \begin{equation}\label{eq:3.8}
\int_{\gD} \varphi(\gO(\gl)) F(d\gl) = \varphi(\widehat{X}).
    \end{equation}
   \end{corollary}
\begin{proof}
It is known (see \cite[Proposition 5.1.4]{BW83}) that the integral  \eqref{Bochner_Int-l}   exists whenever
$\gO(\cdot)$ admits representation  \eqref{integ_repr-n_of_Omega}. By Proposition
\ref{prop:A.-1}(iii)  $\varphi(X(\cdot))$ is $F$-admissible and equality
\eqref{eq:3.8} holds.
  \end{proof}
  \begin{remark}
    {\em
    Absolute continuity of  $\gO(\cdot)$ (and even its Lipschitz property)
does  not ensure  representation \eqref{integ_repr-n_of_Omega}  (see \cite[Chapter
5]{Yosida1980}). Thus, the conditions in both corollaries are different.
}
  \end{remark}
If $F(\cdot)$ is a  spectral measure on $\dR$ with non-compact  support, then we
define improper operator spectral integrals by
\bed
\begin{split}
\int_\dR \gO(\gl)F(d\gl) &:= \slim_{\begin{subarray}{l} b\to +\infty\\ a\to -\infty\end{subarray}}\int_{\gD}\gO(\gl)F(d\gl),\\
\int_\dR F(d\gl)\gO(\gl) &:= \slim_{\begin{subarray}{l} b\to +\infty\\ a\to -\infty\end{subarray}}\int_{\gD} F(d\gl)\gO(\gl).
\end{split}
\eed
Obviously, the improper operator spectral integral $\int_\dR \gO(\gl)F(d\gl)$ exists if
and only if the following conditions
   \bed
\slim_{b\to \infty}\int^{b + \varepsilon}_b\gO(\gl)F(d\gl) = 0 \quad\text{and}\quad
\slim_{a\to -\infty}\int^a_{a - \varepsilon}\gO(\gl)F(d\gl) = 0,
   \eed
are satisfied for any $\varepsilon > 0$. Similar results hold true for $\int_\dR
F(d\gl)\gO(\gl)$.

\begin{proposition}\la{prop:A.-3}
Let $\gO:\dR\longrightarrow \cB(\cH)$. Assume that $\gO\upharpoonright \gD$ is $F$-admissible for every compact interval $\gD$ and
  \bed\
\parallel \gO(\lambda) \parallel \le C_0 (1+ \vert \gl \vert)^{\ga}, \qquad \gl \in \dR,
  \eed
for some constants $\ga\ge 0$,  $C_0>0$. Then the improper spectral integral $\int_{\dR}
\gO(\gl) F(d\gl) f$ exists for any $f\in \cH$ satisfying
\be\la{joh5}
\int_{\dR} \vert \gl \vert^{2\ga} d\parallel F(\gl) f \parallel^2 < \infty.
\ee
\end{proposition}
\begin{proof}
Let $b,c>0$. Let $n\in \dN$. Put $x_m:= b+ \dfrac{m-1}{n}c$, $\gD_m:= [x_m,x_m+\dfrac{c}{n} )$, $\cZ:= \bigcup_{m=1}^n \gD_m$. Then
\begin{eqnarray*}
\parallel \gS_{\cZ} \gO f \parallel^2 & = & \gS_{m=1}^n \parallel \gO(x_m) F(\gD_m) f \parallel^2 \\
& \le & C_0^2 \gS_{m=1}^n (1+x_m)^{2\ga} \parallel F(\gD_m) f \parallel^2 \\
& \le & C_0^2 \int_{[b,b+c)} (1+\gl)^{2\ga} d \parallel F(\gl) f \parallel^2.
\end{eqnarray*}
Passing to the limit, as $n$ tends to infinity, we get that
\bed
\parallel \int_{[b,b+c)} \gO(\gl) F(d\gl) f \parallel^2 \le C_0^2\int_{[b,b+c)}(1+\gl)^{2\ga} d \parallel F(\gl) f \parallel^2.
\eed
The integral on the right hand side tends to zero, as $b$ tends to infinity, provided
\eqref{joh5} holds. The case $a\longrightarrow -\infty$ is  treated similarly.
  \end{proof}
\section{Boundary triplets for tensor products}

\subsection{Bounded case}

Let $A$ be a densely defined symmetric operator with equal deficiency indices acting in
the separable Hilbert space $\gotH$ and let $T$ be a bounded self-adjoint operator
acting on the separable Hilbert space $\gotT$. Let us consider the closed symmetric
operator $S:= A \otimes I_\gotT + I_\gotH \otimes T$ in $\gotH_S := \gotH \otimes
\gotT$. We recall that the operator $S$ is defined as the closure of $S_\odot := A \odot
I_\gotT + I_\gotH \odot T$,
   \bed
\dom(A \odot I_\gotT + I_\gotH \odot T) := \left\{f = \sum^n_{k=1} g_k \otimes h_k: g_k
\in \dom(A), h_k \in \gotT \right\}
  \eed
and
  \bed
S_\odot f := \sum^n_{k=1}
(Ag_k \otimes h_k + g_k \otimes Th_k), \qquad f \in \dom(A \odot I_\gotT + I_\gotH \odot
T).
  \eed
Obviously, the operator $S_\odot$ is densely defined and symmetric.

Let $\Pi_A = \{\cH^A,\gG_0^{A}, \gG_1^{A}\}$ be a boundary triplet for $A^*$ with
$\gamma$-field $\gga^A(\cdot)$ and Weyl function $M^A(\cdot)$. Let $J_{A^*}$ be the embedding
operator $J_{A^*}: \gotH_+(A^*) \longrightarrow \dom(A^*)$. Obviously, $\ran(J_{A^*}) =
\dom(A^*)$ and $\ker(J_{A^*}) = \{0\}$ as well as $\gG^A_j = \wh{\gG^A_j} J^{-1}_{A^*}$,
$j = 0,1$. Notice that
$\gotH_+((A \otimes I_{\gotT})^*) = \gotH_+(A^* \otimes I_{\gotT}) = \gotH_+(A^*)
\otimes \gotT$
and $J_{(A \otimes I_{\gotT})^*} = J_{A^* \otimes I_{\gotT}}.$
Moreover, one has
\bed
\ran(J_{(A \otimes I_{\gotT})^*})  = \dom((A \otimes I_{\gotT})^*) = \dom(A^* \otimes
I_{\gotT}). \eed
We set
     \bed\la{eq:3.6}
(\gG^A_j \wh\otimes I_\gotT) f := (\wh{\gG^A_j} \otimes I_{\gotT})\,J^{-1}_{(A \otimes
I_{\gotT})^*}, \quad j \in \{0,1\}, \quad f \in \dom(A^* \otimes I_{\gotT}).
     \eed
It turns out that $\Pi_A \wh\otimes I_\gotT := \{\cH^A \otimes \gotT,\gG^A_0 \wh\otimes
I_\gotT,\gG^A_1 \wh\otimes I_\gotT\}$ is a boundary triplet for $(A\otimes I_\gotT)^* =
A^* \otimes I_{\gotT}$.

\bt\la{th:3.1}
Let $\Pi_A = \{\cH^A,\gG^A_0,\gG^A_1\}$ be a boundary triplet for $A^{*}$ with
$\gamma$-field $\gamma^A(\cdot)$ and Weyl function $M^A(\cdot)$.   Let also $T = T^* \in
\cB(\gotT)$, and let $\gD$ be the smallest  closed  interval containing the spectrum
$\gs(T)$.
Finally, let  $\widehat E_T(\gd) := I_{\cH^A} \otimes E_T(\gd)$, $\gd \in \cB(\dR)$, where
$E_T(\cdot)$ is the spectral measure of $T$. Then:
\begin{enumerate}

\item[\rm (i)]
$\Pi_S = \{\cH^S,\gG^S_0,\gG^S_1\} := \Pi_A \wh\otimes I_\gotT$ is a boundary triplet
for $S^*$ such that $S_0 := S^*\upharpoonright\ker(\gG_0^S) = A_0 \otimes I_{\gotT} +
I_{\gotH} \otimes T$.

\item[\rm (ii)]
The $\gamma$-field $\gga^S(\cdot)$ and the Weyl function $M^S(\cdot)$ of $\Pi_S$ admit the
following representations
  \be\la{eq:3.7}
\gamma^S(z) = \int_{\gD} \left(\gamma^A(z - \lambda) \otimes I_{\gotT}\right) \widehat
E_T(d\lambda), \qquad z \in \dC_\pm,
  \ee
and
\be\label{weyl}
\begin{split}
M^S(z) &= \int_{\gD} \widehat E_T(d\lambda) \left(M^A(z - \lambda) \otimes I_{\gotT}\right)\\
&=  \int_{\gD}  \left(M^A(z - \lambda) \otimes I_\gotT\right)\widehat E_T(d\lambda),
\qquad z \in \dC_\pm.
\end{split}
\ee
In particular,
  \begin{equation}\label{4.3_def_subspace}
\ran \left( \int_{\Delta}\bigl(\gamma^A(z-\lambda)\otimes I_{\gotT}\bigr){\widehat
E}_T(d\lambda)\right) = \gotN_z(S^*)=\ker(S^*-z).
  \end{equation}

\item[\rm (iii)]
If the Weyl function $M^A(\cdot)$  is of scalar type, $M^A(\cdot) =
m^A(\cdot)I_{\cH^A}$,  then
   \bed
M^S(z) = I_{\cH^A} \otimes m^A(z - T), \qquad z \in \dC_\pm.
   \eed
In particular, the latter holds  whenever  $n_\pm(A) = 1$.
\end{enumerate}
\et

Note that the integrals \eqref{weyl} and \eqref{eq:3.7}  exist due to Corollary
\ref{Cor_Int_with_Lip_fun-n}  since both the Weyl function $M^S(z- \cdot)$ and
$\gamma$-field  $\gamma^S(z- \cdot)$ are holomorphic in $\gl$, hence Lipschitz
functions.

\begin{proof}
(i) The proof is straightforward.

(ii)   In accordance with  \cite[Lemma 7.2]{AdaMen96} both integrals \eqref{eq:3.7}
and \eqref{weyl}  exist since $\gamma^A(\cdot)$ and $M^A(\cdot)$ are
Lipschitz.   Let $\pi =\{a = \gl_0 <\gl_1 < \gl_2 < \ldots  < \gl_n = b\}$  be a
partition of $\gD= [a,b]$, $\Delta_k := [\gl_{k-1}, \gl_{k})$,  and let
   \bed
   \begin{split}
T_k &:= \lambda_k E(\Delta_k), \quad  T_{\pi} := \bigoplus^n_{k=1}T_k  =
\sum^n_{k=1}\lambda_k E_T(\Delta_k),\\
S_{\pi} &:=A\otimes I_{\gotT} +I_{\gotH}\otimes T_{\pi},
\end{split}
   \eed
and $\gotT_k:=\ran E(\Delta_k)$.   $T_k$ is regarded as an operator in $\gotT_k$.
It is easily seen  that $\gotT = \bigoplus^n_1\gotT_k$
and
   \bed
S_{\pi}=\bigoplus^n_{k=1} S_k, \qquad S_{k} := A\otimes I_{\gotT_k} + I_{\gotH}\otimes
T_{k}\in \cC({\gotH}\otimes {\gotT_k}).
  \eed
Clearly,  $S_{k}^* := A^*\otimes I_{\gotT_k} + I_{\gotH}\otimes T_{k}$.
Moreover for every $k$ such that $\gotT_k \not= \{ 0 \}$ we have $\sigma(T_k)=\{\lambda_k\}$ and hence
$\sigma(S^*_k)=\sigma(A^*\otimes I_{\gotT_k}) +\lambda_k $  and
$\gotN_z(S_k) = \gotN_{z- \gl_k}(A)\otimes \gotT_k$.
Clearly,
    \bed
    \begin{split}
S^*_{\pi} &= A^*\otimes \left(\bigoplus^n_{k=1} E_T(\Delta_k)\right) + I_{\gotH}\otimes
\left(\bigoplus^n_{k=1}\lambda_k E_T(\Delta_k)\right) \\
&= \bigoplus^n_{k=1}(A^*+\lambda_k I_{\gotH})\otimes E_T(\Delta_k).
\end{split}
    \eed
Hence   $\gotN_z(S_{\pi}) = \ker(S_{\pi}^* - z I_{\gotH}) = \text{ran}
\left(\sum^n_{k=1}\gamma^A(z-\lambda_k) \otimes E_T(\Delta_k)\right)$

Noting that  $\Gamma^{S_\pi}_0 = \Gamma^{S}_0 = \Gamma^{A}_0\otimes I_{\gotT}$  and
using definition \eqref{2.3A} one gets
  \begin{equation}\label{5.24_gamma-field_S-pi}
  \begin{split}
\Gamma^{S_\pi}_0 &\left(\sum^n_{k=1} \gamma^A(z-\lambda)\otimes E_T(\Delta_k)\right) =
\sum^n_{k=1}\Gamma_0^A\gamma^A(z-\lambda_k)\otimes E_T(\Delta_k)\\
&=\sum^n_{k=1}
I_{\cH}\otimes E_T(\Delta_k) = I_{\cH}\otimes I_{\gotT} = I_{{\cH}\otimes{\gotT}}.
\end{split}
  \end{equation}
Combining this relation with   definition \eqref{2.3A} of the $\gamma$-field   one
derives
   \bed
\gamma^{S_\pi}(z) =  \left(\Gamma^{S_\pi}_0\upharpoonright \gotN_z(S_{\pi})\right)^{-1}
= \sum^n_{k=1} \gamma^A(z - \gl_k)\otimes E_T(\Delta_k).
   \eed

Applying operator $\Gamma_1$ to this equality and using Definition \ref{Weylfunc}  we
arrive at the Weyl function $M^{S_\pi}(\cdot)$ corresponding to the triplet
$\Pi^{S_{\pi}}$ of  $S^*_{\pi}$,
    \bed
    \begin{split}
M^{S_\pi}(z) &=  \Gamma^{S_\pi}_1 \gamma^{S_\pi}(z)  = \sum^n_{k=1} \Gamma^A_1\gamma^A(z
- \gl_k)\otimes E_T(\Delta_k) \\
&=   \sum^n_{k=1} M^A(z-\lambda_k)\otimes E_T(\Delta_k).
\end{split}
    \eed
Since the integrals \eqref{eq:3.7New} and \eqref{weyl_New}  exist, the following uniform
convergence holds
    \begin{equation}\label{5.27_wt-gamma}
\gamma^{S_{\pi}}(z)\to \int_{\gD} \left(\gamma^A(z - \lambda) \otimes I_{\gotT}\right)
\widehat E_T(d\lambda) =: \wt{\gamma}^{S}(z) \qquad \text{as}\quad  |\pi| \to 0,
      \end{equation}
and
\begin{equation}\label{5.27_wt-weyl}
M^{S_{\pi}}(z)\to \int_{\gD} \left(M^A(z - \lambda) \otimes I_\gotT\right)\widehat
E_T(d\lambda) =: \wt {M}^{S}(z)  \qquad \text{as}\quad  |\pi| \to 0,
    \end{equation}
where as usual   $|\pi| = \max_{k = 1,2,\ldots,n}|\gD_k|$.

Next we show that  $\wt{\gamma}^{S}(z) = {\gamma}^{S}(z)$ and  $\wt {M}^{S}(z) =
{M}^{S}(z)$ for $z\in \dC_{\pm}$. One gets
\be
\begin{split}
\bigl((&A^*-z)\otimes
I_{\gotT}\bigr)\gamma^{S_{\pi}}(z)g = \sum^n_{k=1} (A^*-z)\gamma^A(z-\lambda_k)\otimes E_T (\Delta_k)g \nonumber \\
&= -\sum^n_{k=1} \gl_k\gamma^A(z - \gl_k)\otimes E_T(\Delta_k)g
 \to  -\int_{\Delta}\bigl(\lambda\gamma^A(z-\lambda)\otimes I_{\gotT}\bigr){\widehat
 E}_T(d\lambda)g
 \end{split}
\ee
as $|\pi| \to 0$.  Since $A^*$ is closed, one gets by combining this relation with
\eqref{5.27_wt-gamma} that  $\int_{\Delta}\bigl(\lambda\gamma^A(z-\lambda)\otimes
I_{\gotT}\bigr){\widehat E}_T(d\lambda)g \in \dom(A^*\otimes I_{\gotT})$  for each $g\in
{\cH}\otimes {\gotT}$  and
    \bed
        \begin{split}
\bigl((A^*-z)\otimes I_{\gotT} \bigr) &\int_{\Delta}\bigl(\gamma^A(z-\lambda)\otimes
I_{\gotT}\bigr){\widehat E}_T(d\lambda)\\
 &= - \int_{\Delta}\bigl(\lambda\gamma^A(z-\lambda)\otimes I_{\gotT}\bigr){\widehat
 E}_T(d\lambda).
\end{split}
  \eed
In turn, using this relation and applying Proposition \ref{prop:A.-2}  we derive
  \bed
  \begin{split}
(&S^*-z)\int_{\Delta}\bigl(\gamma^A(z-\lambda)\otimes I_{\gotT}\bigr){\widehat
E}_T(d\lambda) \\
&= \bigl((A^*-z)\otimes I\bigr)
\int_{\Delta}\bigl(\gamma^A(z-\lambda)\otimes I_{\gotT}\bigr){\widehat E}_T(d\lambda)\\
 & \qquad\qquad\qquad\qquad+\int_{\Delta}\lambda {\widehat E}_T(d\lambda)\cdot
\int_{\Delta}\bigl(\gamma^A(z-\lambda)\otimes I_{\gotT}\bigr){\widehat E}_T(d\lambda)\nonumber \\
&= - \int_{\Delta}\bigl(\lambda\gamma^A(z-\lambda)\otimes I_{\gotT}\bigr){\widehat E}_T(d\lambda)
 + \int_{\Delta}\bigl(\lambda\gamma^A(z-\lambda)\otimes I_{\gotT}\bigr){\widehat
E}_T(d\lambda) = 0.
\end{split}
  \eed
It follows that $\ran \left( \int_{\Delta}\bigl(\gamma^A(z-\lambda)\otimes
I_{\gotT}\bigr){\widehat E}_T(d\lambda)\right) \subset \gotN_z(S^*)=\ker(S^*-z).$

   Let us show that the convergence in \eqref{5.27_wt-gamma}
holds in $\gotH_+(S)$, i.e. in the graph norm.

Choose a sequence $\{\pi_n\}^{\infty}_1$ of partitions of $[a,b]$ such that
$\lim_{n\to\infty}|\pi_n|=0$. Since the convergence in \eqref{5.27_wt-gamma}  is
uniform, there exists a constant $C(z) >0$ depending on $z$ and not depending on $n$ and
such that $\|\gamma^{S_{\pi_n}}(z)\|\le C(z)$ for all $n$.  Besides,  for any
$\varepsilon > 0$ there exists $N = N(\varepsilon) \in\dN$ such that
$\|T_{\pi_n}-T\|\le\varepsilon$ for $n\ge N$.  Taking  these relations into account one
gets
     \bed
        \begin{split}
\|&(S^* - z)\gamma^{S_{\pi}}(z)g\| = \|(S^* - z)\gamma^{S_{\pi}}(z)g -(S^*_{\pi}-z)\gamma^{S_{\pi}}(z)g\|\\
&= \|\bigl(I\otimes(T-T_{\pi})\bigr)\gamma^{S_{\pi}}(z)g \|
 \le \varepsilon \|\gamma^{S_{\pi}}(z)\| \cdot \|g\| \le \varepsilon C(z)\|g\|
\end{split}
     \eed
for any $\pi\in\{\pi_n\}^{\infty}_N$, hence  $\|\lim_{n\to\infty}(S^* -
z)\gamma^{S_{\pi_n}}(z)g\| =0$ for any $g\in {\cH}\otimes {\gotT}$. In turn, combining
this relation with \eqref{5.27_wt-gamma} yields
    \begin{equation}\label{converg_gamma-field__H+}
\lim_{n\to\infty}\left\|\gamma^{S_{\pi_n}}(z) -  \int_{\gD} \left(\gamma^A(z - \lambda)
\otimes I_{\gotT}\right) \widehat E_T(d\lambda)\right\|_{\gotH_+(S)} =0.
      \end{equation}

It follows from \eqref{5.24_gamma-field_S-pi}  that  $\Gamma^{S}_0 \gamma^{S_\pi}(z)=
\Gamma^{S_\pi}_0 \gamma^{S_\pi}(z) = I_{\cH}\otimes I_{\gotT} \to I_{\cH}\otimes
I_{\gotT}$ as $|\pi| \to 0$.   Therefore relation \eqref{converg_gamma-field__H+}
implies
$$
\Gamma^{S}_0 \wt{\gamma}^{S}(z) =  I_{\cH}\otimes I_{\gotT},
$$
i.e.  $\wt{\gamma}^{S}(z) = {\gamma}^{S}(z)$. This proves \eqref{eq:3.7}.
In  turn, \eqref{eq:3.7} implies  \eqref{4.3_def_subspace}.

Further, combining just established relation $\wt{\gamma}^{S}(\cdot) =
{\gamma}^{S}(\cdot)$ with
relation \eqref{converg_gamma-field__H+}
and using the boundedness of the operator ${\Gamma}^{S}_1\in [{\gotH_+(S)}, \cH]$ we
obtain
    \bed
    \begin{split}
\lim_{n\to\infty} M^{S_{\pi_n}}(z) &= \lim_{n\to\infty} \Gamma^{S_{\pi_n}}_1
\gamma^{S_{\pi_n}}(z) \\
&= \lim_{n\to\infty} \Gamma^{S}_1 \gamma^{S_{\pi_n}}(z) =
\Gamma^{S}_1 \gamma^{S}(z) = M^{S}(z), \quad z\in \dC_{\pm},
\end{split}
    \eed
where the convergence  is uniform.
 In turn, combining this relation with   \eqref{5.27_wt-weyl} yields \eqref{weyl}.
  \end{proof}
  \begin{remark}
    {\em
Another proof of Theorem \ref{th:3.1}  can be found in \cite[cf. Proposition 3.1 and
3.2]{BoiNeiPop2013}.
}
    \end{remark}
\begin{example}
{\em
Let us illustrate the theorem above. To this end we consider the case that $A$ is a closed symmetric operator with deficiency indices $n_\pm = 2$.
In particular, let $\Pi_A = \{\cH^A,\gG^A_0,\gG^A_1\}$ where $\cH^A = \transpose{(\cH^A_1 \oplus \cH^A_2)}$, $\cH^A_j = \dC$, $j = 1,2$. We use the representation
\bed
\gG^A_j  =
\begin{pmatrix}
\gG^A_{j1}\\
\gG^A_{j2}
\end{pmatrix} : \dom(A^*) \longrightarrow
\begin{matrix}
\cH^A_1\\
\oplus\\
\cH^A_2
\end{matrix},
\quad j = 0,1.
\eed
 For the Gamma field
$\gga^A(\cdot)$ we use the representation $\gga^A(z) = (\gga^A_1(z),\gga^A_2(z))$, $\gga_j^A(z): \cH^A_j \longrightarrow \gotH$, $j = 1,2$, $z \in \dC_\pm$. The Weyl function $M^A(\cdot)$ admits the representation
\bed
M^A(z) =
\begin{pmatrix}
m^A_{11}(z) & m^A_{12}(z)\\
m^A_{21}(z) & m^A_{22}(z)
\end{pmatrix}, \quad z \in \dC_\pm,
\eed
where $m^A_{ij}(\cdot)$ are holomorphic functions in $\dC_\pm$.

We consider the closed symmetric operator $S = A \otimes I_\gotT + I_{\gotHÂ} \otimes T$, where $T$ is bounded and self-adjoint.
Let $\Pi_S = \Pi_A \wh \otimes I_\gotT$, cf. Theorem \ref{th:3.1}\,(i). Obviously, the boundary value space $\cH^S = \cH^A \otimes \gotT$ can be decomposed by
$\cH^S  = \transpose{(\cH^S_1 \oplus \cH^S_2)}$, $\cH^S_j := \gotT$, $j = 1,2$.
The boundary value maps $\gG^S_0 = \gG^A_0 \wh \otimes I_\gotT$ and $\gG^S_1 = \gG^A_1 \wh \otimes I_\gotT$  will be represented by
\bed
\gG^S_0 =
\begin{pmatrix}
\gG^S_{01}\\
\gG^S_{02}
\end{pmatrix}:\dom(S^*) \longrightarrow
\begin{matrix}
\cH^S_1\\
\oplus\\
\cH^S_2
\end{matrix}
\;\mbox{and} \;
\gG^S_1 =
\begin{pmatrix}
\gG^S_{11}\\
\gG^S_{12}
\end{pmatrix}: \dom(S^*) \longrightarrow
\begin{matrix}
\cH^S_1\\
\oplus\\
\cH^S_2
\end{matrix}
\eed
where $\gG^S_{0j} := \gG^A_{0j} \wh \otimes I_\gotT$ and $\gG^S_{1j} := \gG^A_{0} \wh \otimes I_\gotT$, $j = 0,1$.
From \eqref{eq:3.7} we get the representation $\gga^S(z) =  (\gga^S_1(z),\gga^S_2(z))$, $z \in \dC_\pm$, where $\gga^S_j(z): \cH^S_j \longrightarrow \gotH$,
\bed \gga^S_j(z) := \int_\gD \gga^A_j(z-\gl)\; \widehat E_T(d\gl), \quad j = 1,2. \eed
The Weyl function $M^S(\cdot)$ admits the representation
\bed
M^S(z) =
\begin{pmatrix}
m^A_{11}(z - T) & m^A_{12}(z-T)\\
m^A_{21}(z - T) & m^A_{22}(z-T)
\end{pmatrix}:
\begin{matrix}
\cH^S_1\\
\oplus\\
\cH^S_2
\end{matrix}
\longrightarrow
\begin{matrix}
\cH^S_1\\
\oplus\\
\cH^S_2
\end{matrix}, \quad z \in \dC_\pm.
\eed
The representation of the Weyl function becomes very simple if $M^A(\cdot)$ is diagonal. In this case we have
$M^S(z) = \diag(m^A_{11}(z-T),m^A_{22}(z-T))$, $z \in \dC_\pm$.
}
\end{example}
Let us compute the normalized boundary triplet $\wt \Pi_S$ associated with $\Pi_S$ in
accordance with Lemma \ref{II.5}.
\begin{proposition}\la{prop:3.2}
Let $\Pi_A = \{\cH^A,\gG^A_0,\gG^A_1\}$ be a boundary triplet for $A^{*}$ with the
$\gamma$-field  $\gamma^A(\cdot)$ and Weyl function $M^A(\cdot)$. Let also $A_0 :=
A^*\upharpoonright\ker(\gG_0^A)$,  $T = T^* \in \cB(\gotT)$, and let $\gD$ be the smallest
closed  interval containing the spectrum $\gs(T)$,
and let $\Pi_S = \Pi_A \wh \otimes I_\gotT$.
Finally, let  $\widehat E_T(\gd) := I_{\cH^A} \otimes E_T(\gd)$, $\gd \in \cB(\dR)$, where
$E_T(\cdot)$ is the spectral measure of $T$. Then:
\item[\;\;{\rm (i)}] The triplet $\wt \Pi_S = \{\wt \cH^S,\wt \gG^S_0,\wt \gG^S_1\}$  with
$\wt \cH^S := \cH^A \otimes \gotT$  and
\be\la{eq:3.10}
\begin{split}
\wt \gG^S_0 :=& \left(\int_{\gD} \widehat E_T(d\gl) \sqrt{\im(M^A(i-\gl))} \otimes I_{\gotT}\right)\cdot(\gG^A_0 \wh\otimes I_\gotT),\\
\wt \gG^S_1 :=& \left(\int_{\gD} \widehat E_T(d\gl) \frac{1}{\sqrt{\im(M^A(i-\gl))}} \otimes I_{\gotT}\right)\cdot\\
& \cdot\left(\gG^A_1 \wh\otimes I_\gotT - \left(\int_{\gD} \widehat E_T(d\gl)
\re(M^A(i-\gl)) \otimes I_\gotT\right) \cdot
(\gG^A_0 \wh\otimes I_\gotT)\right)\\
\end{split}
\ee
forms  a normalized boundary triplet for $S^*$ such that
   \bed
\wt S_0 := S^*\upharpoonright\ker(\wt \gG^S_0) = S_0 = A_0 \otimes
I_\gotT + I_\gotH \otimes T.
  \eed

\item[\;\;{\rm (ii)}]
The $\gamma$-field $\wt \gga^S(\cdot)$ and Weyl function $\wt M^S(\cdot)$ corresponding
to the normalized boundary triplet $\wt\Pi_S$ admit the following representations
   \be\la{3.10}
\wt\gga^S(z) = \int_{\gD} \left(\gga^A(z-\gl) \frac{1}{\sqrt{\im(M^A(i-\gl))}} \otimes
I_{\gotT}\right) \;\widehat E_T(d\gl), \;\; z \in \dC_\pm,
  \ee
and
  \be\la{3.11}
\begin{split}
\wt M^S(z) &= \int_{\gD} (L^A(z-\gl,i-\gl)\otimes I_{\gotT}) \widehat E_T(d\gl) \\
&=
\int_{\gD} \widehat E_T(d\gl) (L^A(z-\gl,i-\gl) \otimes I_{\gotT}), \qquad z \in
\dC_\pm,
\end{split}
  \ee
where
\be\la{eq:3.14} L^A (z,\zeta) := \frac{1}{\sqrt{\im(M^A(\zeta))}}(M^A(z)
-\re(M^A(\zeta)))\frac{1}{\sqrt{\im(M^A(\zeta))}}
   \ee
for $z \in \dC_\pm$, $\zeta \in \dC_+$.

\item[\rm (iii)]

If the Weyl function $M^A(\cdot)$  is of scalar type, $M^A(\cdot) =
m^A(\cdot)I_{\cH^A}$,  $m^A(\cdot)\in R[\dC]$,
then
   \be\la{3.13}
\wt M^S(z) = I_{\cH^A} \otimes \frac{m^A(z - T) -\re(m^A(i -T))}{\im(m^A(i
- T))}, \qquad z \in \dC_{\pm}.
  \ee
In particular, the latter happen  whenever  $n_\pm(A) = 1$.
\end{proposition}
  \begin{proof}
(i) By Theorem \ref{th:3.1}, $M^S(z) = \int_{\gD} M^A(z-\gl) \otimes I_{\cT} \widehat
E_T(d\gl)$, and hence for each $z\in \dC_+$
   \begin{equation}\label{Im_and_Re_Weyl_F}
   \begin{split}
\im(M^S(z)) &= \int_{\gD} \left(\im(M^A(z-\gl)) \otimes I_{\cT}\right) \widehat
E_T(d\gl)\quad \text{and}\\
\re(M^S(z)) &= \int_{\gD} \left(\re(M^A(z-\gl))\otimes
I_{\cT}\right) \widehat E_T(d\gl).
\end{split}
   \end{equation}
First we note that both integrals in \eqref{Im_and_Re_Weyl_F}  exist since the
operator-valued functions $\im(M^A(z-\cdot))$ and  $\re(M^A(z-\cdot)$ are Lipschitz (see
\cite{AdaMen96}).  Moreover, since the spectral measure  $\widehat E_T = I_\gotH \otimes
E_T$ commutes with $M^A(z-\gl) \otimes I_{\cT}$, both functions $\im(M^A(i -
\cdot))\otimes I_{\cT}$ and $\re(M^A(i-\cdot))\otimes I_{\cT}$ are $\widehat
E_T$-admissible. Noting that  $M^A(\cdot)$  is holomorphic  on $\dC_{+}$ and  $0\in
\rho(\im M(z))$ for $z\in \dC_{+}$, one easily concludes  that the operator-valued
functions $\im(M^A(i - \cdot))\otimes I_{\cT}$,  $\re(M^A(i-\cdot))\otimes I_{\cT}$, and
$(\im(M^A(i-\cdot))^{-1}\otimes I_{\cT}$ are  bounded on  the compact set $\gD$ and with
some constants $c_1, c_2 >0$ the following estimates hold
\bed
\begin{split}
&0 < c_1\le  \im(M^A(i - \gl))\otimes I_{\cT}  \le c_2\quad \text{and}\\
&c_2^{-1} \le
(\im(M^A(i- \gl))^{-1}\otimes I_{\cT}\le c_1^{-1}, \qquad \gl \in \gD.
\end{split}
\eed
Since  the function $\varphi(\cdot)=\sqrt{\cdot}$ is continuous on  $\dR_+$, then  in
accordance with  Proposition \ref{prop:A.-1}(iii) the compositions
$({\im(M^A(i-\gl)}))^{1/2} \otimes I_{\cT}$ and  $({\im(M^A(i-\gl)}))^{-1/2} \otimes
I_{\cT}$  are  $\widehat E_T$-admissible and
  \be\la{joh1}
  \begin{split}
R &:= \sqrt{\im(M^S(i))} = \int_{\gD} \left(\sqrt{\im(M^A(i-\gl))} \otimes I_{\gotT}\right) \widehat E_T(d\gl),  \\
R^{-1} &= \frac{1}{\sqrt{\im(M^S(i))}} =  \int_{\gD}
\left(\frac{1}{\sqrt{\im(M^A(i-\gl))}} \otimes I_{\gotT}\right) \widehat E_T(d\gl).
  \end{split}
  \ee
Combining  the second formula in  \eqref{joh1}  with  formula \eqref{weyl} and applying
Proposition \ref{prop:A.-2} one arrives at
     \be\label{4.30_product_im_M_and_re-M}
     \begin{split}
R^{-1}Q &:= R^{-1}{\re(M^S(i))}\\
& =  \int_{\gD}
\left(\frac{1}{\sqrt{\im(M^A(i-\gl))}}{\re(M^A(i-\gl))} \otimes I_{\gotT}\right)
\widehat E_T(d\gl).
\end{split}
  \ee

Now it follows from Lemma \ref{II.5}  (see formula \eqref{eq:2.10}) that  a triplet $\wt
\Pi_S = \{\cH_S,\wt \gG^S_0,\wt \gG^S_1\}$, where
\bed
\wt \gG^S_0 = \sqrt{\im(M^S(i))}\gG^S_0 \quad \text{and}\quad  \wt \gG^S_1 =
\frac{1}{\sqrt{\im(M^S(i))}}(\gG^S_1 - \re(M^S(i))\gG^S_0),
\eed
is a (normalized) boundary triplet for $S^*$. Combining these formulas with formulas
\eqref{joh1}  yields  \eqref{eq:3.10}.

(ii)
Combining \eqref{eq:3.7} with  the second identity in \eqref{joh1} and applying
Proposition \ref{prop:A.-2}  we arrive at
\bed
\begin{split}
&\wt \gga^S(z) = \gga^S(z) R^{-1}\\
&= \int_{\gD} \left(\gga^A(z - \gl) \otimes I_{\gotT}\right)\, \widehat E_T(d\gl) \cdot
\int_{\gD}
\left(\frac{1}{\sqrt{\im(M^A(i - \mu))}} \otimes I_{\gotT}\right) \widehat E_T(d\mu) \\
&= \int_{\gD} \left(\gga^A(z - \gl) \frac{1}{\sqrt{\im(M^A(i-\gl))}} \otimes
I_{\gotT}\right)\, \widehat E_T(d\gl), \qquad z \in \dC_\pm,
\end{split}
\eed
which proves \eqref{3.10}.

Similarly, combining formula  \eqref{weyl} with the third formula in \eqref{joh1} and
applying  Proposition \ref{prop:A.-2} implies
  \bed
\begin{split}
&\frac{1}{\sqrt{\im(M^S(i))}} \left( M^S(z) - \re(M^S(i)) \right) \frac{1}{\sqrt{\im(M^S(i))}}\\
&= \int_{\gD} \frac{1}{\sqrt{\im(M^A(i-\gl))}} \times\\
&\qquad\qquad\times\left( M^A(z-\gl) - \re(M^A(i-\gl)
\right) \frac{1}{\sqrt{\im(M^A(i-\gl))}} \widehat E_T(d\gl),
  \end{split}
  \eed
$z \in \dC_{\pm}$. This proves  \eqref{3.11}.  Moreover, inserting in \eqref{3.11} $z=i$ one easily gets
the equality $\wt M^S(i) = i(I_{\cH^A} \otimes I_{\gotT}) = iI_{\cH^S}$  meaning  that
the triplet $\wt \Pi_S$ is normalized.

(iii) Representation \eqref{3.13} is  immediate from \eqref{3.11}.
\end{proof}

\subsection{Unbounded case}

Let $A$ be a closed densely defined symmetric operator with equal deficiency indices in
$\gotH$ and let $T$ be an unbounded self-adjoint operator in  $\gotT$. First we
introduce an operator $S' := A \odot I_{\gotT} + I_{\gotH} \odot T$ by setting (cf.
\cite[Chapter 7.5.2]{Schmuedgen2012})
\bed
\begin{split}
S'f &:=  A \odot I_{\gotT}f  + I_{\gotH} \odot Tf :=
 \sum^l_{k=1}(Ag_k \otimes h_k) +  \sum^l_{k=1}(g_k \otimes  Th_k), \\
 f &= \sum^l_{k=1} g_k \otimes h_k \in
\dom(S'),\\
 \dom(S') :&:= \left\{f = \sum\limits_{k = 1}^l g_k \otimes h_k: g_k
\in \dom(A), \;\;h_k \in \dom(T),\;\;l\in \dN \right\}.
\end{split}
\eed
Clearly, $S'$
is a densely defined symmetric operator.
Further, we define  the operator $S:= A \otimes I_{\gotT} + I_{\gotH_A} \otimes T$ on
$\gotK := \gotH \otimes \gotT$ as the closure of $S'$, i.e.
   \bed
S := \overline {S'} := \overline{A \odot I_{\gotT} + I_{\gotH} \odot T}.
  \eed

Denote by $\gotH_{+}(A)$  the Hilbert space  obtained by equipping the domain $\dom(A)$
with the graph norm. Let $J_A: \gotH_+(A) \longrightarrow \gotH$ be the embedding
operator. Then
$\dom(A \otimes I_T) = (J_A \otimes I_{\gotT})(\gotH_{+}(A) \otimes \gotT).$
By  \cite[Proposition 7.26]{Schmuedgen2012},   $(A \otimes I_T)^{\ast} = A^{\ast}
\otimes I_T$  and
\bed \dom(A^{\ast} \otimes I_T) = (J_{A^*} \otimes I_{\gotT})(\gotH_{+}(A^{\ast})
\otimes \gotT). \eed
The operator $I_{\gotH} \otimes T = \overline{I_{\gotH} \odot T}$ is unbounded and
self-adjoint. Moreover.  one has
  \bed
S = \overline{A \otimes I_{\gotT} + I_{\gotH} \otimes T} \quad \text{and}\quad \dom(S)
\supseteq \cD := \dom(A \otimes I_{\gotT}) \cap \dom(I_{\gotH} \otimes T).
   \eed
Clearly, $\cD$ is a core for $S$, i.e.  $S = \overline{S\upharpoonright\cD}$.

Further,  setting  $T_n :=  E_T((n,n+1])T$ and $\gotT_n := E_T((n,n+1])\gotT$, $n \in
\dZ,$ one arrives at the  orthogonal decomposition
\bed
T = \bigoplus_{n\in\dZ}T_n, \quad \gotT := \bigoplus_{n\in\dZ}\gotT_n,
\eed
where $T_n = T_n^*\in \cB(\gotT_n)$.
 Let $\gotK_n := \gotH \otimes \gotT_n$, $n \in \dZ$.
Clearly, $\gotK := \gotH \otimes \gotT  = \bigoplus_{n \in \dZ} \gotK_n$. We set $S_n := A
\otimes I_{\gotT_n} + I_{\gotH} \otimes T_n$, $n \in \dZ$. For each $n \in \dZ$ the
operator $S_n$ is a  well-defined closed symmetric operator  in $\gotH_n$.

\bl\la{III.3}
Let $A$ and  $T$  be as above.
Let $T = \bigoplus_{n\in\dZ}T_n$ be an orthogonal decomposition of $T$ where $T_n =
T_n^*\in \cB(\gotT_n)$. Then
      \bed
S = \bigoplus_{n\in \dZ} S_n, \qquad   S_n := A \otimes I_{\gotT} + I_{\gotH_n} \otimes
T_n.
  \eed
In particular, if $T$ has a pure point spectrum, then
$S = \bigoplus_{n \in \dZ} S_n$  where   $S_n = A \otimes I_{\gotT_n} + \lambda_n
I_{\gotH_n},$    $\{\gl_n\}_{n\in\dZ}$ is  the sequence  of eigenvalues  of $T$,  and
$\gotH_n := \gotH \otimes \gotT_n$ with  ${\gotT_n} = E_T( \{ \lambda_n \} ) \gotT$.
\el
\begin{proof}
The proof is obvious.
\end{proof}
In general, for any self-adjoint extension $S_0$ of $S$ there is a boundary triplet
$\Pi_S = \{\cH^S,\wt \gG^S_0,\wt \gG^S_1\}$ such that $S_0 =
S^*\upharpoonright\ker(\gG^S_0)$. Moreover, in accordance with  Lemma \ref{II.5} it is
always possible starting with a $\Pi_S$ to define a normalized boundary triplet $\wt
\Pi_S$. In particular, we can find a boundary triplet $\Pi_S$ for $S^*$, $S = A \otimes
I_\gotT + I_\gotH \otimes T$, such that $S_0 := A_0 \otimes I_\gotT + I_\gotH \otimes
T$. However, in applications we need a special boundary triplet feeling a tensor
structure of the operators $S$ and $S^*$ and leading to simple forms of the
corresponding Weyl function  and $\gamma$-field.

Therefore in what follows  we choose another strategy. Let $\Pi_A$ be a boundary triplet
for  $A^*$ with the corresponding $\gamma$-field $\gga^A(\cdot)$ and  Weyl function
$M^A(\cdot)$. Starting with  this boundary triplet for  $A^*$ we construct a normalized
boundary triplet $\Pi_S = \{\cH^S,\gG^S_0,\gG^S_1\}$ for $S^*$ such that $S_0 =
S^*\upharpoonright \ker(\gG^S_0)$ and the corresponding  $\gamma$-field $\gga^S(\cdot)$
and  Weyl function $M^S(\cdot)$ can be  explicitly computed by means of $\gga^A(\cdot)$
and $M^A(\cdot)$ (cf. the proof of Theorem \ref{th:2.5}).
\bl\la{lem:3.4}
Let $A$ be a densely defined closed symmetric operator in $\gotH$. Let also $\Pi_A =
\{\cH^A,\gG^A_0,\gG^A_1\}$ be a boundary triplet for $A^*$ and let $M^A(\cdot)$ and
$\gamma(\cdot)$ be the corresponding Weyl function and $\gamma$-field, respectively.
Further, let $T$ be a  self-adjoint operator on $\gotT$ with spectral measure
$E_T(\cdot)$ and let $\widehat E_T(\cdot) := I_{{\mathcal H}^A} \otimes E_T(\cdot)$. Then
the following improper spectral integrals
\begin{align}
\begin{split}
G_0 f & := \int_\dR \widehat E_T(d\gl) \left(\sqrt{\im(M^A(i-\gl))} \otimes
I_{\gotT}\right)f
   \la{eq:3.47} \\
      & = \int_\dR \left(\sqrt{\im(M^A(i-\gl))} \otimes I_{\gotT}\right)  \widehat E_T(d\gl) f
                \end{split}\\
            \begin{split}
G_1 f & := \int_\dR \widehat E_T(d\gl) \left(\frac{1}{\sqrt{\im(M^A(i-\gl))}} \otimes I_{\gotT}\right) f
\la{eq:3.48} \\
      & = \int_\dR \left(\frac{1}{\sqrt{\im(M^A(i-\gl))}} \otimes I_{\gotT}\right) \widehat E_T(d\gl) f,
            \end{split}\\
            \begin{split}
G_2 f & := \int_\dR \widehat E_T(d\gl)
\left(\frac{1}{\sqrt{\im(M^A(i-\gl))}}\re(M^A(i-\gl)) \otimes I_\gotT\right)f
\la{eq:3.49}\\
      & = \int_\dR  \left(\frac{1}{\sqrt{\im(M^A(i-\gl))}}\re(M^A(i-\gl)) \otimes I_\gotT \right) \widehat E_T(d\gl) f
            \end{split}
\end{align}
exist for each  $f \in \dom(I_{\cH^A} \otimes T)$. Moreover, the following improper
spectral integrals
  \be\la{eq:3.41}
G(z)f  := \int_\dR \left(\gga^A(z-\gl)\ \frac{1}{\sqrt{\im(M^A(i-\gl))}} \otimes
I_{\gotT}\right) \widehat E_T(d\gl)f,
  \ee
and
  \be\la{eq:3.43}
\begin{split}
M(z)f  &:= \int_\dR \left(L^A(z-\gl,i-\gl)\otimes I_{\gotT}\right)\widehat E_T(d\gl)f\\
      &= \int_\dR \,\widehat E_T(d\gl) \left(L^A(z-\gl,i-\gl)\otimes I_{\gotT}\right)f, \qquad z \in \dC_\pm
\end{split}
\ee
exist for every  $f \in \cH^A \otimes \gotT$, where $L^A(z,\zeta)$, $z \in \dC_\pm$,
$\zeta \in \dC_+$, is given by \eqref{eq:3.14}.
  \el
\begin{proof}
We divide the proof in several steps.
(i)  Let $f\in \dom(I_{\cH^A} \otimes T)$. Then
  \bed
\int_{\dR} \lambda^2 d \parallel \widehat E_T(\gl) f \parallel^2 < \infty.
  \eed
Note that in accordance with \eqref{joh4},
\bed
\begin{split}
&\parallel (\im(M^A(i-\gl)))^{1/2} \otimes I_{\cT}
\parallel =  O(\vert \gl \vert) \quad \text{ and}\\
&\parallel (\im(M^A(i-\gl)))^{-1/2}
\otimes I_{\cT} \parallel = O(\vert \gl \vert) \quad\text{as}\quad \gl \to\infty.
\end{split}
\eed
Therefore the convergence of the integrals in \eqref{eq:3.47} and \eqref{eq:3.48} is
immediate  from Proposition \ref{prop:A.-3} with $\alpha =1$.

(ii)  To prove  \eqref{eq:3.49} it  suffices  to show that
  \begin{equation}\label{4.22_est-te_Im_times_Re}
\parallel  (\im(M^A(i-\gl)))^{-1/2} \re(M^A(i-\gl))\parallel = O(\vert \gl
\vert) \quad\text{as}\quad \gl \to\infty.
  \end{equation}
Noting that
  \bed
  \begin{split}
(\im&(M^A(i-\gl))^{-1/2} M^A(i-\gl) = \\
&(\im(M^A(i-\gl))^{-1/2} \re(M^A(i-\gl)) + i
(\im(M^A(i-\gl))^{1/2}
\end{split}
  \eed
and taking estimate \eqref{joh4}  into account one concludes  that  the required
estimate \eqref{4.22_est-te_Im_times_Re} is equivalent to the following one
 \begin{equation}\label{4.23_New_estimate}
\parallel(\im(M^A(i-\gl))^{-1/2} M^A(i-\gl) \parallel  = O(\vert \gl \vert)
\qquad\text{as}\qquad \gl \to\infty.
  \end{equation}
Further, in accordance with  \eqref{mlambda}
   \bed
   \begin{split}
\im(M^A(i-\gl)) &= - \im(M^A(-i-\gl)) = \gga^A(-i-\gl)^*\gga^A(-i-\gl)\\
& =
\gga^A(i-\gl)^*\gga^A(i-\gl), \qquad \gl \in \dR.
\end{split}
   \eed
Hence  there  exists  a family of isometries $V(\gl \pm i)$ mapping $\cH^A$ onto
$\cN_A(\pm i -\gl) = \ker(A^* + \gl  \mp i)$ and such that
  \be\la{eq:3.80}
V(\gl \pm i) (\im(M^A(i-\gl))^{1/2} = \gga^A(\pm i-\gl), \qquad \gl \in \dR.
  \ee
Using \eqref{mlambda}, we get
   \bed
   \begin{split}
M^A(i-\gl) - M^A(i)^* &= (2i -\gl)\gga^A(-i-\gl)^* \gga^A(-i) \\
&= (2i -\gl)
(\im(M^A(i-\gl))^{1/2} V(\gl)^*\gga^A(-i).
\end{split}
  \eed
Thus
   \bed
   \begin{split}
(\im&(M^A(i-\gl))^{-1/2} M^A(i-\gl)\\
& = (2i -\gl)V(\gl)^*\gga^A(-i) +
(\im(M^A(i-\gl))^{-1/2} M^A(-i).
\end{split}
  \eed
Combining this relation with estimate \eqref{joh4} yields \eqref{4.23_New_estimate} as
well as \ \
 $$
\parallel (\im(M^A(i-\gl)))^{-1/2} \re(M^A(i-\gl)) \parallel = O(\vert \gl
\vert).
$$
To prove \eqref{eq:3.49}  it remains  to apply Proposition \ref{prop:A.-3} with $\alpha
=1$.

(iii)  To  prove the convergence of the integral \eqref{eq:3.41} it suffices to show
that
   \be\label{4.27_est-te_gama-weyl(-1/2)}
\parallel \gga^A(z-\gl)(\im(M^A(i-\gl))^{-1/2} \parallel
\le \varkappa(z),  \qquad \gl \in \dR.
   \ee
with some positive constant $\varkappa(z)>0$. In accordance with  \eqref{2.5}
  \bed
\gga^A(z-\gl) = (A_0 + \gl - i)(A_0 + \gl - z)^{-1}\gga^A(i-\gl), \qquad z \in \dC_\pm,
\quad \gl \in \dR.
  \eed
Moreover, it follows from \eqref{eq:3.80} that
  \bed
(\im(M^A(i-\gl))^{-1/2} = (\gga^A(i-\gl))^{-1} V(i+\gl), \qquad \gl \in \dR.
  \eed
Combining these relations yields
  \be\label{4.29_Identity_for_gama-Weyl}
\gga^A(z-\gl)(\im(M^A(i-\gl))^{-1/2} = (A_0 + \gl - i)(A_0 + \gl - z)^{-1}V(i +\gl),
  \ee
$z \in \dC_\pm$, $\gl \in \dR$.
On the other hand
\bed
\parallel(A_0 + \gl - i)(A_0 + \gl - z)^{-1} \parallel  =  \parallel I+ (z - i)(A_0 + \gl - z)^{-1}\parallel
 \le 1 + \frac{|z - i|}{|\im z|} =: \varkappa(z).
\eed
Combining this estimate  with identity  \eqref{4.29_Identity_for_gama-Weyl}  we arrive
at the estimate \eqref{4.27_est-te_gama-weyl(-1/2)}. Proposition \ref{prop:A.-3} with
$\alpha =1$ completes the proof.

(iv)  To prove the existence of the integral \eqref{eq:3.43} it suffices to show that
for each fixed $z\in \dC_{\pm}$
  \be\label{estimate_for_LA}
\parallel L^A(z-\gl,i-\gl) \parallel = O(1)
\qquad\text{as}\qquad \gl \to\infty,
   \ee
and apply Proposition \ref{prop:A.-3}.   It follows from \eqref{eq:3.14} and identity
\eqref{mlambda}  that
   \bed
   \begin{split}
L^A&(z-\gl,i-\gl)\\
=&\frac{1}{\sqrt{\im(M^A(i-\gl))}}(M^A(z-\gl) -
M^A(i-\gl))\frac{1}{\sqrt{\im(M^A(i-\gl))}} + iI_{\cH} \\
 =& (z- i) \frac{1}{\sqrt{\im(M^A(i-\gl))}} \gga^A(-i-\gl)^* \gga^A(z
-\gl)\frac{1}{\sqrt{\im(M^A(i-\gl))}} + iI_{\cH},
\end{split}
   \eed
$z \in \dC_\pm$, $\gl \in \dR$. Inserting in this identity instead of $\gga^A(-i-\gl)^*$  its expression from
\eqref{eq:3.80} one gets
   \bed
   \begin{split}
L^A&(z-\gl,i-\gl)\\
&= (z- i) V(\gl-i)^* \gga^A(z -\gl)\frac{1}{\sqrt{\im(M^A(i-\gl))}} + iI_{\cH}.
\end{split}
   \eed
Finally, combining this identity with  \eqref{4.27_est-te_gama-weyl(-1/2)} implies
\eqref{estimate_for_LA}.
  \end{proof}
  \begin{remark}
    {\em
Combining estimates \eqref{4.22_est-te_Im_times_Re} and   \eqref{joh4}  we obtain
  \bed
  \begin{split}
\parallel&\re(M^A(i-\gl))\parallel le \parallel  (\im(M^A(i-\gl)))^{1/2}\parallel \times\\
&\times \parallel
(\im(M^A(i-\gl)))^{-1/2} \re(M^A(i-\gl))\parallel = O(\vert \gl \vert^2).
\end{split}
  \eed
{as} $\gl \to\infty.$   Simple examples show that even for  a scalar  Nevanlinna
function $f\in R[\dC]$ the function  $\parallel(\im(f(i-\gl)))^{-1/2}
\re(f(i-\gl))\parallel$ is not necessarily bounded.
}
\end{remark}

\bt\la{th:2.5}
Let $\Pi_A = \{\cH^A,\gG^A_0,\gG^A_1\}$ be a boundary triplet for $A^{*}$, let
$M^A(\cdot)$ and $\gamma^A(\cdot)$  be the corresponding  Weyl function and
$\gamma$-field, respectively.   Let also $T = T^*\in \cC(\gotT)\setminus \cB(\gotT)$  and
$S := A\otimes I_\gotT + I_{\gotH}\otimes T$. Then:

\item[\;\;\rm (i)]
There exists  a normalized boundary triplet $\wt \Pi_S = \{\wt \cH^S,\wt \gG^S_0,\wt
\gG^S_1\}$ for $S^*$ such that $\wt \cH^S := \cH^A \otimes \gotT$ and $S_0 :=
S^*\upharpoonright\ker(\wt \gG^S_0) = A_0 \otimes I_{\gotT} + I_{\gotH} \otimes T$, and
for any  $f \in \gotD := \dom(S^*) \cap \dom(I_{\gotH} \otimes T) (\subseteq \dom(S^*))$
\begin{align}\la{eq:3.57}
\wt \gG^S_0&f:=\left(\int_\dR \widehat E_T(d\gl) \sqrt{\im(M^A(i-\gl))} \otimes I_{\gotT}\right)\cdot(\gG^A_0
\wh \otimes I_\gotT)f,\nonumber\\
\wt \gG^S_1&f := \left(\int_\dR \widehat E_T(d\gl) \frac{1}{\sqrt{\im(M^A(i-\gl))}}
\otimes I_{\gotT}\right)\cdot  (\gG^A_1 \wh\otimes I_\gotT) f \\
  & - \left(\int_{\dR} \widehat E_T(d\gl)
\frac{1}{\sqrt{\im(M^A(i-\gl))}} \re(M^A(i-\gl)) \otimes I_\gotT\right) \cdot (\gG^A_0 \wh\otimes I_\gotT)f.  \nonumber
\end{align}

\item[\;\;\rm (ii)]
The  $\gamma$-field $\wt \gga^S(\cdot)$ and Weyl function $\wt M^S(\cdot)$ corresponding
to the triplet  $\wt \Pi_S$ are given by
   \be\la{eq:3.65}
\wt \gga^S(z) = G(z) \qquad \mbox{and} \qquad \wt M^S(z) = M(z), \quad z
\in \dC_\pm,
  \ee
where $G(\cdot)$ and $M(\cdot)$ are defined by \eqref{eq:3.41} and \eqref{eq:3.43},
respectively.

\item[\;\;\rm (iii)]
If $M^A(\cdot)$ is of scalar type, i.e. $M^A(\cdot) = m^A(\cdot)I_{\cH^A}$, then
representation \eqref{3.13} remains true.
\et
\begin{proof}
(i)
Clearly,  $f \in \dom(A^* \otimes I_\gotT)$. Let $\gD_n := [n,n+1)$, $n \in \dZ$. We set
$\gotT_n := E_T(\gD_n)\gotT$ and $T_n = TE_T(\gD_n)$, $n \in \dZ$. Notice that $\gotT =
\bigoplus_{n\in\dZ}\gotT_n$ and $T = \bigoplus_{n\in\dZ}T_n$. Let also  $R_{S_n} :=
\sqrt{\im(M^{S_n}(i))}$ and  $Q_{S_n} := \re(M^{S_n}(i))$, $n \in \dZ$.  Then, by
Proposition \ref{prop:3.2},   a triplet $\wt \Pi_{S_n} =
\{\cH^{S_n},\wt\gG^{S_n}_0,\wt\gG^{S_n}_1\}$ with
  \bed
  \begin{split}
\cH^{S_n} &:= \cH^A \otimes \gotT_n,\quad   \wt \gG^{S_n}_0  = R_{S_n}(\gG^A_0 \wh\otimes
I_{\gotT_n}), \quad \text{and}\\
& \wt \gG^{S_n}_1 =  R^{-1}_{S_n}\left(\gG^{S_n}_1 -
Q_{S_n}\gG^{S_n}_0\right) = R^{-1}_{S_n}\gG^{S_n}_1 - R^{-1}_{S_n} Q_{S_n}\gG^{S_n}_0,
\end{split}
   \eed
is  a boundary triplet for  $S^*_n$ for each   $n \in \dZ$. In turn,   Theorem
\ref{th:2.6}  ensures  that the direct sum $\wt\Pi_S := \bigoplus_{n\in\dZ}\wt \Pi_{S_n}
 = \{\wt \cH^S,\wt \gG^S_0,\wt \gG^S_1\}$ of boundary triplets  is an ordinary (normalized)
boundary triplet for $S^* = \bigoplus_{n\in \dZ} S^*_n$.

Setting $R := \bigoplus_n R_{S_n}$, applying  formula \eqref{joh1}  and noting that, by
Lemma \ref{lem:3.4}, the improper spectral integral \eqref{eq:3.47} exists one gets that
for any $h = \mathlarger{\mathlarger{\mathlarger{\oplus}}}_n h_n \in \dom(I_{\cH^A} \otimes T) = \oplus_n \dom(I_{\cH^A}
\otimes T_n)$
\be\label{4.51_for-la_for_R}
\begin{split}
Rh &=  \bigoplus_{n\in\dZ} R_{S_n}h_n = \bigoplus_{n\in\dZ} \sqrt{\im(M^{S_n}(i))}h_n  \\
&=
\bigoplus_{n\in\dZ} \int_{[n,n+1)}
\widehat E_{T_n}(\gl)\left(\sqrt{\im(M^A(i-\gl)} \otimes I_{\gotT_n}\right)h_n    \\
&=  \slim_{\begin{subarray}{l} p\to +\infty\\ q\to -\infty\end{subarray}}\int_{[q,p)}
\widehat E_{T}(\gl)\left(\sqrt{\im(M^A(i-\gl)} \otimes I_{\gotT}\right)h \\
&=  \int_{\dR}
\widehat E_{T}(\gl)\left(\sqrt{\im(M^A(i-\gl)} \otimes I_{\gotT}\right) = G_0 h.
\end{split}
\ee
Note that applying formula \eqref{joh1}  we have replaced the integral $\int_{[n,n+1]}$
by $\int_{[n,n+1)}$. The latter is possible  since $n+1 \not\in\gs_p(T_n)$ for each
$n \in \dZ$.

Next, similarly to \eqref{4.51_for-la_for_R} and using the convergence of the improper
spectral integral \eqref{eq:3.48} one gets from \eqref{joh1}
  \begin{equation}\label{4.53_formula_for_R-1}
  \begin{split}
R^{-1}h &= \bigoplus_{n\in\dZ} R_{S_n}^{-1}h_n = \bigoplus_{n\in\dZ} \left(\sqrt{\im(M^{S_n}(i))}
\right)^{-1}h_n \\
&= \int_\dR \widehat E_T(d\gl) \left(\frac{1}{\sqrt{\im(M^A(i-\gl))}}
\otimes I_{\gotT}\right) h = G_1h.
\end{split}
   \end{equation}
Further,  setting $Q := \bigoplus_n Q_{S_n} := \bigoplus_n {\re(M^{S_n}(i))}$, applying
formula   \eqref{4.30_product_im_M_and_re-M} with $\Delta_n$ in place of $\Delta$, and
noting that by Lemma \ref{lem:3.4} the improper spectral integral  \eqref{eq:3.49}
exists, we derive
\begin{align}\label{4.54_R-1_Q}
R^{-1}&Q h  = \bigoplus_{n\in\dZ} R^{-1}_{S_n} {\re(M^{S_n}(i))}h_n \\
 &=
\bigoplus_{n\in\dZ} \int_{[n,n+1)}
\widehat E_{T_n}(\gl)  \left(\frac{1}{\sqrt{\im(M^A(i-\gl))}}{\re(M^A(i-\gl))} \otimes I_{\gotT_n})\right)h_n    \nonumber\\
& =    \int_{\dR} \widehat E_{T}(\gl)
\left(\frac{1}{\sqrt{\im(M^A(i-\gl))}}{\re(M^A(i-\gl))} \otimes I_{\gotT})\right)h =
G_2h.\nonumber
\end{align}

Further,  let $f = \{f_n\}_{n\in\dZ} \in \gotD \subseteq \dom(A^* \otimes I_\gotT)$,
$f_n \in \gotH_A \otimes \gotT_n$,  $n \in \dZ$.  Note that $f \in \dom(\gG^A_0
\wh\otimes I_\gotT)\cap \dom(\gG^A_1 \wh\otimes I_\gotT)$  because  $f \in \dom(A^*
\otimes I_\gotT)$.  Hence
  \bed
(\gG^A_0 \wh\otimes I_\gotT) f = \bigoplus_{n\in\dZ}(\gG^A_0 \wh\otimes I_{\gotT_n}) f_n
\quad \mbox{and} \qquad (\gG^A_1 \wh\otimes I_\gotT) f = \bigoplus_{n\in\dZ}(\gG^A_1
\wh\otimes I_{\gotT_n}) f_n.
  \eed
On the other hand, by Theorem \ref{th:2.6} (see formula \eqref{2.16_reg_triplet})
\bed
\wt \gG^S_0 f = R(\gG^A_1 \wh\otimes I_\gotT) f \quad \text{and}\quad   \wt \gG^S_1 f =
R^{-1}(\gG^A_1 \wh\otimes I_\gotT) f  + R^{-1} Q (\gG^A_0 \wh\otimes I_\gotT) f,
\eed
$f \in \gotD$. Inserting in these relations   instead of $R$, $R^{-1}$, and $R^{-1}Q$ their expressions
from  \eqref{4.51_for-la_for_R}, \eqref{4.53_formula_for_R-1}, and \eqref{4.54_R-1_Q},
one arrives at formulas \eqref{eq:3.57}.

\noindent (ii) In accordance with  Proposition \ref{prop:3.2}(ii) the $\gamma$-field and
Weyl function  corresponding to the triplet $\wt \Pi_{S_n} =
\{\cH^{S_n},\wt\gG^{S_n}_0,\wt\gG^{S_n}_1\}$ are given by
\bed
\begin{split}
\wt\gga^{S_n}(z)  &=
\int_{[n,n+1)} \left(\gga^A(z-\gl) \frac{1}{\sqrt{\im(M^A(i-\gl))}} \otimes
I_{\gotT_n}\right) \widehat E_{T_n}(d\gl), \quad z \in \dC_\pm,
\end{split}
\eed
and
\bed
\begin{split}
\wt M^{S_n}(z) &= \int_{[n,n+1)} \left(L^A(z-\gl,i-\gl)\otimes I_{\gotT_n}\right)
\widehat E_{T_n}(d\gl) \\
&= \int_{[n,n+1)}  \widehat E_{T_n}(d\gl) \left(L^A(z-\gl,i-\gl)
\otimes I_{\gotT_n}\right), \quad z \in \dC_\pm,
\end{split}
\eed
respectively.   Here $L^A (z,\zeta)$ is given by \eqref{eq:3.14}. Further, applying
Theorem \ref{th:2.6} (see formula \eqref{W-fun_and_gam-field_for_dir_sum}) and taking
into account formulas \eqref{eq:3.41} and \eqref{eq:3.43}, we  arrive at
\eqref{eq:3.65}.

(iii) This statement is now immediate from formula \eqref{3.13}  and representation $T =
\bigoplus_{n\in\dZ}T_n$ with $T_n\in \cB(\gotT_n)$.
\end{proof}


%
\begin{remark}
{\em
\item[\;\;\rm (i)] If $T$ is pure point, $\gs(T) = \gs_{pp}(T) = \{\gl_k\}_{k\in\dZ}$, then the boundary space $\cH^S$ admits the representation
$\cH^S = \bigoplus_{k\in\dZ}\cH_k$, where $\cH_k = \cH^A \otimes \gotT_k$ and $\gotT_k$ is the eigenspace which corresponds to $\gl_k$. One easily checks that the Weyl function
admits the representation
\bed
M^S(z) = \bigoplus_{k\in\dZ} \left(L(z-\gl_k,i-\gl_k) \otimes I_{\gotT_k}\right), \quad z \in \dC_\pm.
\eed
\item[\;\;\rm (ii)] The set  $\gotD := \dom(S^*) \cap \dom(I_{\gotH_A} \otimes T) \subseteq \dom(S^*)$ is a core for $S^*$. Equivalently this means that
$\gotD$ regarded as a subset $\wh \gotD$ of $\gotH_+(S^*)$ is dense in the Hilbert space
$\gotH_+(S^*)$. Let $J_{S^*}: \gotH_+(S^*) \longrightarrow \gotH_S = \gotH_A \otimes
\gotT$ be the embedding operator. We set $\wh \gG^S_j := \gG^S_jJ_{S^*}: \gotH_+(S^*)
\longrightarrow \cH^S$, $j \in \{0,1\}$. The operator $\wh \gG^S_j$, $j \in \{0,1\}$, is
bounded. Hence
   \bed
\wh \gG^S_j = \overline{\gG^S_jJ_{S^*}\upharpoonright\wh \gotD}, \quad j\in
\{0,1\}.
  \eed
In other words, the closure of the operator $\gG^S_j\upharpoonright\gotD$, $j \in
\{0,1\}$, with respect to the topology of $\gotH_+(S^*)$ gives $\gG^S_j$, $j \in
\{0,1\}$.
}
\end{remark}

\begin{remark}\la{rem:3.9}
{\em
The case of a scalar type Weyl function can be slightly extended.
Let us assume that there is a boundary triplet $\Pi_A = \{\cH^A,\gG^A_0,\gG^A_1\}$ of $A^*$ such that
$\cH^A = \bigoplus^{n(A)}_{k=1}\cH^A_k$, $\cH^A_k := \dC$, $n(A) := n_\pm(A)$. With respect to this decomposition we suppose that
the Weyl function $M^A(\cdot)$ is diagonal, that is, it admits the representation
\bed
\begin{split}
M^A(z) &= \diag(m_1(z),m_2(z),\ldots,m_{n(A)}(z))\\
&=
\begin{pmatrix}
m^A_1(z)  & 0 & \cdots & \cdots \\
0 & m^A_2(z)  & \cdots & \cdots \\
\vdots & \vdots & \ddots & \vdots\\
\cdot & \cdot & \cdots & m^A_{n(A)}(z)
\end{pmatrix}:
\begin{matrix}
\cH^A_1\\
\oplus\\
\cH^A_2\\
\oplus\\
\vdots\\
\oplus\\
\cH^A_{n(A)}
\end{matrix}
\longrightarrow
\begin{matrix}
\cH^A_1\\
\oplus\\
\cH^A_2\\
\oplus\\
\vdots\\
\oplus\\
\cH^A_{n(A)}
\end{matrix}, \quad z \in \dC_\pm,
\end{split}
\eed
where $m_k(\cdot)$, $k =1,2,\ldots,{n(A)}$, are scalar Nevanlinna functions. If the Weyl function of a boundary triplet has this structure, then it is called of quasi scalar type.
We are going to compute the boundary triplet $\Pi_S$ as well $\gga$-field $\gga^S(\cdot)$ and Weyl function $M^S(\cdot)$ for the quasi scalar type  case. We set
\bed
\gG^A_{jk} := P^{\cH^A}_{\cH^A_k}\gGÂ_j: \dom(A^*) \longrightarrow \cH^A_k, \quad j = 0,1, \quad k = 1,2,\ldots,{n(A)}.
\eed
Obviously, we have
\bed
\gG^A_{1k}f_z = m_k(z)\gG^A_{0k}f_z, \qquad f_z \in \ker(A^*-z), \quad k = 1,2,\ldots,{n(A)}.
\eed
Let us introduce the operator $\gG^A_{jk} \wh\otimes I_\gotT : \dom(A^* \otimes I_\gotT) \longrightarrow \cH^S_k := \cH^A_k \otimes \gotT = \gotT$, $j = 0,1$, $k = 1,2,\ldots,{n(A)}$.
Notice that
\bed
\gG^A_j \wh\otimes I_\gotT =
\begin{pmatrix}
\gG^A_{j1} \wh\otimes I_\gotT\\
\gG^A_{j2} \wh\otimes I_\gotT\\
\vdots\\
\gG^A_{j{n(A)}} \wh\otimes I_\gotT
\end{pmatrix} : \dom(A^*\otimes I_\gotT) \longrightarrow
\begin{matrix}
\cH^S_1\\
\oplus\\
\cH^S_2\\
\oplus\\
\vdots\\
\oplus\\
\cH^S_{n(A)}
\end{matrix}\quad .
\eed
Notice that $\cH^S = \cH^A \otimes \gotT = \bigoplus^{n(A)}_{k=1}\cH^S_k$. Setting
$\gG^S_{jk} := P^{\cH^S}_{\cH^S_k}\gG^S_j$, $j \in \{0,1\}$, $k\in
\{1,2,\ldots,{n(A)}\}$, we get $\gG^S_j =
\transpose{(\gG^S_{j1},\gG^S_{j2},\ldots,\gG^S_{j{n(A)}})}$, $j \in \{0,1\}$. Using
\eqref{eq:3.57} we get
\be\la{eq:3.71}
\begin{split}
\gG^S_{0k} f &= \sqrt{\im(m_k(i-T))}(\gG^A_{0k} \wh{\otimes} I_\gotT)f\\
\gG^S_{1k}f & = \frac{1}{\sqrt{\im(m_k(i-T))}}\left(\gG^A_{1k} \wh{\otimes} I_\gotT - \re(m_k(i-T))(\gG^A_{0k}
\wh{\otimes} I_\gotT)\right)f,
\end{split}
\ee
$f \in \dom(A^*\otimes I_\gotT) \cap \dom(I_{\gotH_A} \otimes T)$, $k\in
\{1,2,\dots,{n(A)}\}$.

To compute the $\gamma$-field we set
\bed
\gga^A_k(z) := \gga^A(z)\upharpoonright \cH^A_k, \quad  \gga^A(z) = (\gga^A_1(z),\gga^A_2(z),\ldots,\gga^A_{n(A)}(z)),
\eed
$z \in \dC_\pm$, and
\bed
\gga^S_k(\cdot) = \gga^S(\cdot)\upharpoonright \cH^S_k, \quad \gga^S(z) = (\gga^S_1(z),\gga^S_2(z),\ldots,\gga^S_{n(A)}(z)),
\eed
$z \in \dC_\pm$, where $\cH^S_k := \cH^A_k \otimes \gotT = \gotT$, $k\in \{1,2,\ldots,{n(A)}\}$. From \eqref{eq:3.41} we find
  \be\la{eq:3.72}
\gga^S_k(z) = \gga^A_k(z -T)\frac{1}{\sqrt{\im(m_k(i-T))}}, \;\; z \in \dC_\pm, \;\;
k\in \{1,2,\ldots,{n(A)}\}.
  \ee
Finally, the Weyl function takes the form
\begin{align}\la{eq:3.73}
&M^S(z) =\\
&\diag\left(\frac{m^A_1(z-T) -
\re(m_1(i-T))}{\im(m_1(i-T))},\ldots,\frac{m^A_{n(A)}(z-T) -
\re(m_{n(A}(i-T))}{\im(m_{n(A)}(i-T))}\right) \nonumber
\end{align}
$z \in \dC_\pm$.
}
\end{remark}

\section{Sums of tensor products with  non-negative summands}

\subsection{Boundary triplets in the case of non-negative operators $A$ and $T$}
Here we complete previous results assuming the  operators $A$ and $T$ to be
non-negative. We denote by
$\widehat A_{\rm F}$ and
$\widehat A_{\rm K}$ the \emph{Friedrich's} and \emph{Krein's} extension of $A$, respectively.

\bt\la{th:5.1_Positive}
Let $A$ be a non-negative symmetric  operator in $\gotH$ and let $\Pi_A =
\{\cH^A,\gG^A_0,\gG^A_1\}$ be a boundary triplet for $A^{*}$ such that $A_0 :=
A^*\upharpoonright\ker(\gG_0^A) = {\widehat A}_F$. Let also $M^A(\cdot)$ and
$\gamma^A(\cdot)$ be the corresponding Weyl function and  $\gamma$-field, respectively.
Let also $T = T^* \in \cB(\gotT)$, $T\ge 0$ and let $S= A \otimes I_{\gotT} + I_{\gotH}
\otimes T$. Finally, let  $\widehat E_T(\cdot) := I_{\cH^A} \otimes E_T(\cdot)$, where
$E_T(\cdot)$ is the spectral measure of $T$. Then:
\begin{enumerate}

\item[\rm (i)]
$\Pi_S = \{\cH^S,\gG^S_0,\gG^S_1\} := \Pi_A \wh\otimes I_\gotT := \{\cH^A \otimes
\gotT,\gG^A_0 \wh\otimes I_\gotT,\gG^A_1 \wh\otimes I_\gotT\}$ is a boundary triplet for
$S^*$ such that
$$
S_0 := S^*\upharpoonright\ker(\gG_0^S) =  {\widehat S}_F = {\widehat A}_F \otimes
I_{\gotT} + I_{\gotH} \otimes T.
$$
\item[\rm (ii)]
The $\gamma$-field $\gga^S(\cdot)$ and  Weyl function $M^S(\cdot)$ of $\Pi_S$ admit the
following representations
  \be\la{eq:3.7New}
\gamma^S(z) = \int_{\gD} \left(\gamma^A(z - \lambda) \otimes I_{\gotT}\right) \widehat
E_T(d\lambda), \qquad z \in \dC\setminus \Delta,
  \ee
and
\be\label{weyl_New}
\begin{split}
M^S(z) &= \int_{\gD} \widehat E_T(d\lambda) \left(M^A(z - \lambda) \otimes I_{\gotT}\right)\\
&=  \int_{\gD}  \left(M^A(z - \lambda) \otimes I_\gotT\right)\widehat E_T(d\lambda),
\qquad z \in   \dC\setminus \Delta,
\end{split}
\ee
where $\gD$ is the smallest  closed  interval containing the spectrum $\gs(T)$.

\item[\rm (iii)]
If the Weyl function $M^A(\cdot)$  is of scalar type, $M^A(\cdot) =
m^A(\cdot)I_{\cH^A}$,  then
   \bed
M^S(z) = I_{\cH^A} \otimes m^A(z - T), \qquad z \in \dC_\pm.
   \eed
In particular, the latter holds  whenever  $n_\pm(A) = 1$.
\end{enumerate}
\et
  \begin{proof}
(i) It is immediate from the definition that $S_0 = S^*\upharpoonright\ker(\gG_0^S)  =
{A}_0 \otimes I_{\gotT} + I_{\gotH} \otimes T$. It remains to apply  Proposition
\ref{prop_5.5_Frid_and_Kr_tensor}.

Statements (ii) and (iii) are immediate from Theorem \ref{th:3.1}.
  \end{proof}
\bl\la{lem:3.4_Posit}
Let $A$ be a densely defined closed non-negative symmetric operator in $\gotH$ and let
$\Pi_A = \{\cH^A,\gG^A_0,\gG^A_1\}$ be a boundary triplet for $A^*$ and let $A_0\ge 0$.
Let also  $M^A(\cdot)$ and $\gamma^A(\cdot)$  be  the corresponding Weyl function and
$\gamma$-field, respectively. Further, let $T$ be a non-negative self-adjoint operator
on $\gotT$, let $E_T(\cdot)$ be its  spectral measure,  and let $\widehat E_T(\cdot) :=
I_{{\mathcal H}^A} \otimes E_T(\cdot)$.
Then  the following improper spectral integrals

\begin{align}
\begin{split}
G_0^+ f  := \int_{\dR_+} \widehat E_T(d\gl) \left(\sqrt{((M^A)'(a-\gl))} \otimes
I_{\gotT}\right)f, \qquad a<0, \quad  \la{eq:3.47_posit}
                \end{split}\\
            \begin{split}
G_1^+ f  := \int_{\dR_+} \widehat E_T(d\gl) \left(\frac{1}{\sqrt{((M^A)'(a-\gl))}}
\otimes I_{\gotT}\right) f,  \qquad a<0, \quad  \la{eq:3.48_posit}
            \end{split}\\
            \begin{split}
G_2^+ f  :=  \int_{\dR_+} \widehat E_T(d\gl)
\left(\frac{1}{\sqrt{(M^A)'(a-\gl)}}M^A(a-\gl) \otimes I_{\gotT}\right)f, \;\; a<0.
\la{eq:3.49_posit}
            \end{split}
\end{align}
exist for each  $f \in \dom(I_{\cH^A} \otimes T)$. Moreover, the following improper
spectral integrals
  \be\la{eq:3.41_Posit}
G(z)f  := \int_{\dR_+} \left(\gga^A(z - \gl)\ \frac{1}{\sqrt{(M^A)'(a -\gl)}} \otimes
I_{\gotT}\right) \widehat E_T(d\gl)f,
  \ee
$z\in \dC \setminus \dR_+$, $a<0$, and
  \be\la{eq:3.43_Posit}
\begin{split}
M(z)f  &:= \int_{\dR_+} \left(L^A(z - \gl, a-\gl)\otimes I_{\gotT}\right)\widehat E_T(d\gl)f\\
   &= \int_{\dR_+} \,\widehat E_T(d\gl) \left(L^A(z -\gl,a-\gl)\otimes I_{\gotT}\right)f, \quad z\in \dC \setminus \dR_+,
\end{split}
\ee
converge  for every  $f \in \cH^A \otimes \gotT$, where
     \be\la{eq:5.17_LA(x,a)}
L^A (z,a) := \frac{1}{\sqrt{(M^A)'(a)}}\left(M^A(z)
-M^A(a)\right)\frac{1}{\sqrt{(M^A)'(a)}},
   \ee
$z \in \rho(A_0)$, $a  \in \dR_-$.
   \el
  \begin{proof}
(i) First we prove the convergence of integral in  \eqref{eq:3.41_Posit}.
 It follows from \eqref{mlambda} that $(M^A)'(z) = \gamma^A(\overline
z)^*\gamma^A(z)$. Hence
 \begin{equation}\label{5.11_MA'-identity}
(M^A)'(a-\lambda) =\gamma^A(a-\lambda)^*\gamma^A(a-\lambda),   \qquad  \lambda\in\dR_+,
\quad a<0.
  \end{equation}
This identity implies the existence of an isometry $V(a-\lambda)$ mapping $\cH$ onto
$\gotN_{a-\lambda}(A)$ and such that
   \begin{equation}\label{5.8_V(a)}
V(a-\lambda)\sqrt{(M^A)'(a-\lambda)}= \gamma^A(a-\lambda), \qquad  \lambda\in\dR_+,
   \end{equation}
Further, in accordance with \eqref{2.5}
  \begin{equation}\label{5.9_gamma-field_connect}
  \begin{split}
\gamma^A(z-\lambda) &= (A_0 - a +\lambda)(A_0-z+\lambda)^{-1}\gamma^A(a-\lambda) \\
&=
U(a-\lambda, z-\lambda)\gamma^A(a-\lambda),
\end{split}
  \end{equation}
where $U(a-\lambda, z-\lambda):= (A_0 - a +\lambda)(A_0-z+\lambda)^{-1} \upharpoonright
\gotN_{a-\lambda}(A)$.  It is easily checked that  $U(a-\lambda, z-\lambda)$
isomorphically maps  $\gotN_{a-\lambda}(A)$ onto $\gotN_{z-\lambda}(A)$. Combining
relation  \eqref{5.9_gamma-field_connect}  with \eqref{5.8_V(a)}  yields
  \begin{eqnarray}\label{5.10_resol-estimate}
\|\gamma^A(z - \lambda)\bigl((M^A)'(a-\lambda)\bigr)^{-1/2}\| = \|U(a-\lambda, z-\lambda) V(a-\lambda)\|\nonumber  \\
= \|I + (z-a)(A_0-z+\lambda)^{-1}\| \le
\begin{cases}
1+|z-a|\cdot|\im z|^{-1},& z\in \dC \setminus \dR,  \\
1+|x-a|\cdot |x|^{-1},& x\in\dR_-.
\end{cases}
 \end{eqnarray}
Here we have taken into account that $|x|  \le | x - \lambda| = |x| + \lambda$. The
latter estimate implies  boundedness of the integrand in \eqref{eq:3.41_Posit} for each
$z\in \dC \setminus \dR_+$.   It remains  to apply Proposition \ref{prop:A.-3}
with $\alpha =0$.

(ii) Let us prove that
  \begin{equation}\label{5.14_estimate-for-LA}
C(z, a):= \sup_{\gl\in \dR_+}\|L^A(z - \gl, a-\gl)\|<\infty \quad \text{for each} \quad
z\in \dC \setminus \dR_+ \quad \text{and} \quad a<0.
   \end{equation}
 Combining identity \eqref{mlambda}   with \eqref{5.9_gamma-field_connect} yields
  \begin{eqnarray*}
M^A(z-\lambda)- M^A(a-\lambda)=(z-a)\gamma^A(a-\lambda)^*\gamma^A(z-\lambda)  \nonumber \\
= (z-a)\gamma^A(a-\lambda)^* U(a-\lambda,z-\lambda)\gamma^A(a-\lambda).
  \end{eqnarray*}
In turn,  inserting this identity in  \eqref{eq:5.17_LA(x,a)} and using \eqref{5.8_V(a)}
one derives
\be\label{5.11_difer_of_L}
\begin{split}
L^A(z - \gl,&a-\gl) =  (z-a)\frac{1}{\sqrt{(M^A)'(a - \gl)}}\gamma^A(a - \gl)^*\times\\
&\qquad\times U(a - \gl,z -\gl)\gamma^A(a - \gl)\frac{1}{\sqrt{(M^A)'(a - \gl)}}\\
&= (z-a)V(a-\lambda)^* U(a - \gl,z -\gl)V(a-\lambda), \quad \gl \in \dR_+.
\end{split}
  \ee
Noting that $V(a-\lambda)$ is an isometry for each $\gl\in \dR_+$ and using estimate
\eqref{5.10_resol-estimate}  one arrives at estimate \eqref{5.14_estimate-for-LA}. To
prove
 convergence of the  integral  \eqref{eq:3.43_Posit} for each  $f \in \cH^A \otimes
\gotT$,
it remains to apply Proposition \ref{prop:A.-3} with $\alpha =0$.

(iii) Let us prove convergence of integrals \eqref{eq:3.47_posit} and
\eqref{eq:3.48_posit}.
Since $A_0\ge 0$, integral representation \eqref{W-F+Integral_rep-n} implies
\begin{equation}\label{W-F_positive}
(M^A)'(a-\lambda) = \int_{\dR_+}\frac{d\Sigma_A(t)}{(t - a + \lambda)^2}, \qquad
\lambda\in\dR_+, \quad   a<0.
 \end{equation}
Using this representation instead of  \eqref{W-F_Imag_part_Int_rep} one proves the
following analog of  estimate  \eqref{joh4}
      \begin{equation}\label{joh4_posit}
C_1(1+|\lambda|^2)^{-1}  \im M(i)\le  (M^A)'(a-\lambda) \le  C_2(1+|\lambda|^2) \im
M(i),
      \end{equation}
$\lambda\in\dR_+$. Combining this estimate with inequality  $\int_{\dR_+} \lambda^2 d \parallel \widehat
E_T(\gl) f \parallel^2 < \infty$ characterizing  $f\in \dom(I_{\cH^A} \otimes T)$, and
applying Proposition \ref{prop:A.-3} with $\alpha = 1$ yields convergence of both
integrals \eqref{eq:3.47_posit} and \eqref{eq:3.48_posit}.

(iv) Due to Proposition \ref{prop:A.-3} (with $\alpha = 1)$  to prove
\eqref{eq:3.49_posit} it suffices to  show that
\begin{equation}\label{M'(-1/2)MA-estimate}
\|\left((M^A)'(a - \lambda)\right)^{-1/2}M^A(a-\lambda)\| =
O(|\lambda|)\quad\text{as}\quad \gl \to\infty.
\end{equation}
In accordance with  \eqref{mlambda}
   \begin{equation}
M^A(a-\lambda)=  M^A(a) - \lambda\gamma^A(a-\lambda)^*\gamma^A(a).
   \end{equation}
Combining this identity with \eqref{5.11_MA'-identity}  we  derive
   \be
   \begin{split}
\|\big((&M^A)'(a-\lambda)\big)^{-1/2}M^A(a-\lambda)\|\\
&\le \|M^A(a)\|\cdot
\|\left((M^A)'(a-\lambda)\right)^{-1/2}\| +
|\lambda|\cdot\|V^*(a-\lambda)\cdot\gamma^A(a)\|.
\end{split}
 \ee
Noting that $V(a-\lambda)$ is an isometry and taking  \eqref{joh4_posit} into account we
arrive at estimate  \eqref{M'(-1/2)MA-estimate}.
   \end{proof}
  \bt\la{th:2.5_Positive}
Let $\Pi_A = \{\cH^A,\gG^A_0,\gG^A_1\}$ be a boundary triplet for $A^{*}$, $A_0 :=
A^*\upharpoonright\ker(\gG^A_0)$,   let $M^A(\cdot)$ and $\gamma^A(\cdot)$  be the
corresponding  Weyl function and
$\gamma$-field, respectively.   Let also $T = T^*\in \cC(\gotT)$ be an unbounded
self-adjoint operator in $\gotT$ and  $S := A\otimes I_{\gotT} + I_{\gotH} \otimes T$.
Then:

\item[\;\;\rm (i)]
There exists  a boundary triplet $\wt \Pi_S = \{\wt \cH^S,\wt \gG^S_0,\wt \gG^S_1\}$ for
$S^*$ such that $\wt \cH^S := \cH^A \otimes \gotT$ and $S_0 :=
S^*\upharpoonright\ker(\wt \gG^S_0) = A_0 \otimes I_{\gotT} + I_{\gotH} \otimes T$. If
$f \in \gotD := \dom(S^*) \cap \dom(I_{\gotH} \otimes T) \subseteq \dom(S^*)$, then $f
\in \dom(S^*) \cap \dom(A^* \otimes I_\gotT)$ and
\begin{align}\la{eq:3.57_positive}
\wt \gG^S_0&f:=\left(\int_{\dR_+} \widehat E_T(d\gl) {\sqrt{(M^A)'(a -\gl))}}  \otimes I_{\gotT}\right)
\cdot(\gG^A_0 \wh \otimes I_\gotT)f,\nonumber\\
\wt \gG^S_1&f := \left(\int_{\dR_+} \widehat E_T(d\gl) \frac{1}{\sqrt{(M^A)'(a -\gl)}}  \otimes
I_{\gotT}\right)(\gG^A_1 \wh\otimes I_\gotT) f   \\
& - \left(\int_{\dR_+} \widehat E_T(d\gl) \left(\frac{1}{\sqrt{(M^A)'(a -\gl)}}\,
M^A(a-\gl) \otimes I_\gotT\right)\right) \cdot (\gG^A_0 \wh\otimes I_\gotT)f.\nonumber
\end{align}

\item[\;\;\rm (ii)]
The  $\gamma$-field $\wt \gga^S(\cdot)$ and Weyl function $\wt M^S(\cdot)$ corresponding
to $\wt \Pi_S$ are given by
   \be\la{eq:3.65_posit}
\wt \gga^S(z) = G(z) \qquad \mbox{and} \qquad \wt M^S(z) = M(z), \qquad z \in \rho(S_0),
  \ee
where $G(\cdot)$ and $M(\cdot)$ are defined by  \eqref{eq:3.41_Posit} and
\eqref{eq:3.43_Posit}, respectively.

\item[\;\;\rm (iii)]
If $M^A(\cdot)$ is a scalar type function, i.e. $M^A(\cdot) = m^A(\cdot)I_{\cH^A}$, then
representation \eqref{3.13} remains true.
\et
  \begin{proof}
(i) First  we let $\gD_n := [n-1,n)$,  $\gotT_n := E_T(\gD_n)\gotT$, and $T_n =
TE_T(\gD_n)$, $n \in \dN$. Clearly,  $\gotT = \bigoplus_{n\in\dZ}\gotT_n$ and $T =
\bigoplus_{n\in\dZ}T_n$. We also put $S_{n} := A\otimes I_{\gotT_n} + I_{\gotH}\otimes
T_{n}\in \cC({\gotH}\otimes {\gotT_n})$. Clearly,  each $T_n$  is bounded and
$\gs(T_n)\subset [n-1,n]$.

By Theorem \ref{th:5.1_Positive},   $\Pi_{S_n} = \{\cH^{S_n},\gG^{S_n}_0,\gG^{S_n}_1\}
:= \Pi_A \wh\otimes I_{\gotT_n} := \{\cH^A \otimes \gotT_n,\gG^A_0 \wh\otimes
I_{\gotT_n},\gG^A_1 \wh\otimes I_{\gotT_n}\}$ is a boundary triplet for $S^*_n$ such
that
$$
S_{0n} := S^*\upharpoonright\ker(\gG_0^{S_n}) =   { A}_0 \otimes I_{\gotT_n} + I_{\gotH}
\otimes T_n, \qquad n\in \dN.
$$
Let also $M^{S_n}(\cdot)$ be the corresponding Weyl function.   It follows from
\eqref{weyl_New} that
  \begin{equation}\label{weyl_derivative}
(M^{S_n})'(z) =  \int_{\gD_n} \left((M^A)'(z -\gl) \otimes I_{\gotT_n}\right) \widehat
E_T(d\gl), \qquad z\in \dC\setminus \gD_n.
  \end{equation}
Since the function $\varphi(\cdot)=\sqrt{\cdot}$ is continuous on  $\dR_+$, then  in
accordance with Proposition \ref{prop:A.-1}(iii) the compositions ${((M^A)'(a -
\gl)})^{1/2} \otimes I_{\gotT_n}$ and  ${((M^A)'(a - \gl)})^{-1/2} \otimes I_{\gotT_n}$
are $\widehat E_{T_n}$-admissible. Therefore  combining  representation
\eqref{weyl_derivative}  with Proposition \ref{prop:A.-1}(iii) yields
  \be\la{joh1_positive}
  \begin{split}
R_n &:= \sqrt{(M^{S_n})'(a)} =  \int_{\gD_n} \left(\sqrt{(M^A)'(a-\gl)} \otimes I_{\gotT_n}\right) \widehat E_T(d\gl),  \\
R_n^{-1} &= \frac{1}{\sqrt{(M^{S_n})'(a)}} =  \int_{\gD_n}
\left(\frac{1}{\sqrt{(M^A)'(a-\gl)}} \otimes I_{\gotT_n}\right) \widehat E_T(d\gl),
  \end{split}
  \ee
$a < 0$. Similarly,  using  representations  \eqref{weyl_New} and \eqref{joh1_positive}  and
applying Proposition \ref{prop:A.-2}  yields
  \be \la{5.16_Rn-1_M-positive}
  \begin{split}
R&_n^{-1}M_n^A(a) = \frac{1}{\sqrt{(M^{S_n})'(a)}}M_n^A(a) \\
& =  \int_{\gD_n}
\left(\frac{1}{\sqrt{(M^A)'(a-\gl)}}M^A(a-\gl) \otimes I_{\gotT_n}\right) \widehat
E_T(d\gl), \qquad a<0.
\end{split}
  \ee
Setting  $\cH_{S_n} := \cH^A \otimes \gotT_n$,
\be\label{tripl_posit}
\begin{split}
&\wt \gG^{S_n}_0 = \sqrt{(M^{S_n})'(a)}\gG^{S_n}_0 \qquad \text{and}\\
&\wt
\gG^{S_n}_1 = \frac{1}{\sqrt{(M^{S_n})'(a)}}(\gG^{S_n}_1 - M^{S_n}(a))\gG^{S_n}_0),
\end{split}
\ee
we obtain an ordinary boundary triplet  $\wt \Pi_{S_n} = \{\cH_{S_n},\wt \gG^{S_n}_0,\wt
\gG^{S_n}_1\}$ for $S_n^*$. Inserting formulas \eqref{joh1_positive} and
\eqref{weyl_derivative}  in \eqref{tripl_posit} yields \eqref{eq:3.57_positive} with
$\Delta_n$ in place of $\dR_+$.
Now applying  Proposition  \ref{cor_III.2.2_02}  (see formula \eqref{III.2.2_08}) one
gets that the direct sum $\wt\Pi_S := \bigoplus_{n\in\dN}\wt \Pi_{S_n}$ is an ordinary
boundary triplet for $S^*$. In particular, for any $f \in \gotD = \dom(S^*) \cap
\dom(A^* \otimes I_\gotT)$
   \begin{align}\la{Gamma_0-dir_sum}
\wt \gG^S_0 f &:=  \bigoplus_{n=1}^{\infty} \wt \gG^{S_n}_0 f \\
&= \bigoplus_{n=1}^{\infty}
\left(\int_{[n-1,n)}
\left({\sqrt{(M^A)'(a -\gl)}} \otimes I_{\gotT_n}\right)\widehat E_{T_n}(d\gl)\right)\cdot(\gG^A_0 \wh \otimes I_{\gotT_n})f \nonumber\\
&= \left(\int_{\dR_+} \left({\sqrt{(M^A)'(a -\gl)}} \otimes I_{\gotT}\right)\widehat
E_{T_n}(d\gl) \right)\cdot(\gG^A_0 \wh \otimes I_{\gotT})f,\nonumber
  \end{align}
which proves the first formula in \eqref{eq:3.57_positive}. Note that convergence of the
last integral for every $f \in \gotD$ (cf. \eqref{eq:3.47_posit})  is  guaranteed by
Lemma \ref{lem:3.4_Posit}.   Formula \eqref{eq:3.57_positive} for $\wt \gG^S_1$ is
proved similarly.

(ii)  It easily follows from \eqref{tripl_posit} that the Weyl function ${\wt
M}^{S_n}(\cdot)$ corresponding to the triplet  $\wt \Pi_{S_n}$ is
  \be\label{5.25_weyl_f-n_for_S_n}
  \begin{split}
 {\wt M}^{S_n}(z) &= R_n^{-1}\left( M^{S_n}(z) - M^{S_n}(a) \right)R_n^{-1} \\
 &=
\frac{1}{\sqrt{(M^{S_n})'(a)}} \left( M^{S_n}(z) - M^{S_n}(a) \right)
\frac{1}{\sqrt{(M^{S_n})'(a)}}
\end{split}
  \ee

Inserting  formulas \eqref{joh1_positive} and  \eqref{weyl_New}  into
\eqref{5.25_weyl_f-n_for_S_n} and applying Proposition \ref{prop:A.-2} we arrive at the
following representation
  \bed
  \begin{split}
{\wt M}^{S_n}(z) &= \int_{\gD_n} \Big(\frac{1}{\sqrt{(M^{A})'(a -\gl)}}\times\\
&\qquad\times ( M^A(z-\gl) - M^A(a-\gl))
\frac{1}{\sqrt{(M^{A})'(a -\gl)}}\Big)  \widehat E_T(d\gl),
\end{split}
  \eed
$z \in \dC_{\pm}$. Finally applying Proposition  \ref{cor_III.2.2_02} and taking notation
\eqref{eq:5.17_LA(x,a)} into account  we arrive at formula  for the Weyl function ${\wt
M}^{S}(\cdot)$ corresponding to $\wt\Pi_S$,
  \bed
  \begin{split}
{\wt M}^{S}(z)f &=  \bigoplus_{n\in\dN} {\wt M}^{S_n}(z)f =  \bigoplus_{n\in\dN}
\int_{\Delta_n} \left(L^A(z - \gl, a-\gl)\otimes I_{\gotT}\right)\widehat E_T(d\gl)f \\
&=
\int_{\dR_+} \left(L^A(z - \gl, a-\gl)\otimes I_{\gotT}\right)\widehat E_T(d\gl)f,  \quad z\in \dC \setminus \dR_+
\end{split}
   \eed
exist for every  $f \in \cH^A \otimes \gotT$ and any  $z\in \dC \setminus \dR_+$.
Note that Lemma \ref{lem:3.4_Posit}  ensures convergence of the last integral for every
$f \in \cH^A \otimes \gotT$.  Comparison with  \eqref{eq:3.43_Posit} proves the second
equality in   \eqref{eq:3.65_posit}. The first one is extracted by combining the first
formula in \eqref{W-fun_and_gam-field_for_dir_sum} with \eqref{joh1_positive}  and
applying Proposition \ref{prop:A.-2}.
  \end{proof}

 \subsection{Friedrichs and Krein extensions of $S := A\otimes
I_{\gotT}+I_{\gotH}\otimes T.$ }

In this section we assume that both a symmetric operator $A\in \cC(\mathfrak{H})$ and
the operator $T=T^*$ are non-negative.  Then the set $\Ext_A[0,\infty)$ of  non-negative
self-adjoint extensions of  $A$ is non-empty (see \cite{AG81,BS87,Ka76}). Moreover,
according to the Krein result \cite{K47} the set $\Ext_A[0, \infty)$ contains two
extremal extensions:  a maximal non-negative extension $\widehat A_{\rm F}$ (also called
\emph{Friedrichs'} or \emph{hard} extension) and  a minimal non-negative extension
$\widehat A_{\rm K}$ (\emph{Krein's} or \emph{soft} extension). The latter are uniquely
determined by  the following inequalities
$$
(\widehat  A_F+x)^{-1} \le (\wt A + x)^{-1} \le (\widehat A_K + x)^{-1}, \quad x\in
(0,\infty), \quad \wt A\in  \Ext_A(0,\infty),
$$
(for detail we refer the reader to \cite{AG81, Ka76}).

Recall the following statements.
   \begin{proposition}[\cite{DM91}]\label{prop_II.1.5_01}
Let $A\ge 0$ and let \ $\Pi=\{\cH,\Gamma_0,\Gamma_1\}$ be a boundary triplet for $A^*$
such that $A_0 (= A^*\upharpoonright \ker\Gamma_0)\ge 0$.
Let $M(\cdot)$ be the corresponding Weyl function.  Then $A_0= \widehat A_{\rm F}\ \
(A_0 = \widehat A_{\rm K})$ if and only if
      \begin{equation}\label{Fr1}
\lim_{x\downarrow-\infty}(M(x)f,f)=-\infty,\quad
\bigl(\lim_{x\uparrow0}(M(x)f,f)=+\infty\bigr),\quad f\in\cH\setminus\{0\}.
    \end{equation}
       \end{proposition}
Next we describe the Friedrichs extension $\widehat S_F$ of $S$ by means of the
extension $\widehat A_F$ of $A$.
We start with the following simple algebraic lemma.
  \begin{lemma}\label{lem_tensor_product}
Let $\{X_k\}^n_1$ be a sequence of positive definite operators in $\cH,$  $X_k\ge d
I_\cH >0$, $d > 0$, and let $E_T(\cdot)$ be a spectral measure of the operator $T =
T^*\in \cB(\gotT)$. Then for any partition $\{\Delta_k\}^n_1$ of $[a,b]$
$\bigl(\sigma(T)\subset [a,b]\bigr)$ one has
   \begin{equation}\label{Fr3}
X:=\sum_k X_k\otimes E_T(\Delta_k)\ge dI_{\cH \otimes \gotT}.
   \end{equation}
  \end{lemma}
\begin{proof}
Since $X_k\ge dI_\cH >0$, the operator $(X_k - d I_\cH)\otimes E_T(\Delta_k)$ is
non-negative. Hence
\bed
\begin{split}
X &= \sum_k X_k\otimes E_T(\Delta_k)\ge d \sum_k I_{\cH} \otimes E_T(\Delta_k) \\
&= dI_{\cH} \otimes \left(\sum_k E_T(\Delta_k)\right) = d I_{\cH}\otimes I_{\gotT} = d I_{\cH \otimes \gotT},
\end{split}
\eed
This inequality proves the result.
  \end{proof}

\begin{proposition}\label{prop_5.5_Frid_and_Kr_tensor}
Let $A$ be a non-negative symmetric operator in $\gotH$, let $T=T^*\ge 0$ and let $S := A\otimes
I_{\gotT}+I_{\gotH}\otimes T$. Then:
  \begin{equation}\label{Fr5}
\widehat S_F = \widehat A_F\otimes I_{\gotT}+I_{\gotH}\otimes T  \qquad \text{and}
\qquad \widehat S_K = \widehat A_K\otimes I_{\gotT}+I_{\gotH}\otimes T.
  \end{equation}
\end{proposition}
   \begin{proof}
(i) Assume for the beginning that $T$ is bounded, $T\in \cB(\gotT)$. Let $\Pi_A =
\{\cH^A,\gG^A_0,\gG^A_1\}$ be a boundary triplet for $A^*$ such that $A_0 = \widehat
A_F$. Then, by Theorem \ref{th:5.1_Positive},  $\Pi_S = \{\cH^S,\gG^S_0,\gG^S_1\} :=
\Pi_A \wh\otimes I_\gotT$  is a boundary triplet for $S^*$ satisfying   $S_0 :=
S^*\upharpoonright\ker(\gG_0^S) = A_0 \otimes I_{\gotT} + I_{\gotH} \otimes T$, and the
corresponding Weyl function $M^S(\cdot)$ is given by \eqref{weyl_New}.

To prove the first relation in \eqref{Fr5} it suffices to check condition \eqref{Fr1}
for $M^S(\cdot)$. Let $h := \sum^n_{j=1}h'_j\otimes h''_j$ where $h'_j\in\cH^A,
h''_j\in\gotT$, let $\cH_n^A := \text{span}\{h'_j:1\le j\le n\}$ and let $P_n$ be the
orthogonal projection on $\cH_n^A$ in $\cH^A$.

 Since $A_0 = \widehat A_F$, the Weyl function
$M^A(\cdot)$ satisfies condition \eqref{Fr1}. Setting $M^A_n(\cdot)=P_n
M(\cdot)\upharpoonright \cH_n^A$ we note that due  to the compactness of the
finite-dimensional ball condition \eqref{Fr1}  is uniform on each $\cH_n^A$. In other
words, for each $N>0$ there exists $x_N<0$ such that
   \begin{equation}\label{Fr6}
-M^A_n(x)\ge N  \qquad \text{for}\qquad  x\le x_N.
  \end{equation}
Since $A_0\ge 0$, Theorem \ref{th:5.1_Positive} ensures that  the Weyl function
$M^A(\cdot)$ being  a holomorphic in  $\dC\setminus \dR_+$ admits  the integral
representation  \eqref{weyl_New}  for any $z=x<0$ and $\gl>0$.  Let
$\pi=\{\Delta_k\}^p_1$ be a partition of $\Delta=[a,b]$,  let $\lambda_k\in\Delta_k$,
and let
    \begin{equation}\label{5.21_integ_sum_for_W_F}
S_p(\pi)=\sum_{k=1}^p M^A(x_N - \lambda_k)\otimes E_T(\Delta_k)
    \end{equation}
be an integral sum for the integral   \eqref{weyl_New}  with $x=x_N$. Setting $Y_k =
M^A_n(x_N - \lambda_k)$, $k\in \{1, \ldots, p\}$, one gets
    \begin{eqnarray}\label{Fr8}
(P_n\otimes I_\gotT)  S_p(\pi)h = \sum^p_{k=1} \sum^n_{j=1} P_n M^A(x_N - \lambda_k)h'_j\otimes
E_T(\Delta_k)h''_j   \nonumber  \\
= \sum^p_{k=1} \sum^n_{j=1} Y_k h'_j\otimes E_T(\Delta_k)h''_j = \sum^p_{k=1}  (Y_k \otimes
E_T(\Delta_k))h.
    \end{eqnarray}
Combining this relation  with \eqref{Fr6} and noting that $h\in \cH_n^A \otimes \gotT$
 and $x_N - \lambda_k  <x_N$  one gets from   Lemma \ref{lem_tensor_product}  that
      \bed
\bigl(S_p(\pi)h,h\bigr) = \bigl((P_n\otimes I_\gotT) S_p(\pi)h,h\bigr)  \le -N
      \eed
Passing here to the limit as the diameter $|\pi|$ of partition  $\pi$  tends to zero and
taking formula  \eqref{weyl_New}  for the Weyl function into account and setting
$M^S_n(\cdot)= (P_n\otimes I_\gotT) M(\cdot)\upharpoonright \cH_n^A\otimes I_\gotT$, one
derives
   \bed
\bigl(M^S(x)h,h\bigr) = \bigl(M_n^S(x)h,h\bigr)  \le -N \qquad \text{for}\qquad  x\le
x_N.
   \eed
Since finite  tensors $h = \sum^n_{j=1}h'_j\otimes h''_j$  are dense in $\cH^A
\otimes\gotT$, this inequality  yields   condition \eqref{Fr1} for $M(\cdot) =
M^S(\cdot)$ and arbitrary $h\in \cH^A\otimes \gotT$.

(ii) Let $T\in \cC(\cH)\setminus \cB(\cH)$. Then $T$ admits a decomposition
    \bed
T = \bigoplus_{n\in\dN} T_n,
    \eed
where $T_n := T E_T[n-1,n)\in\cB(\cH_n)$ and $\cH_n := E_T[n-1,n) \cH$.
Hence
   \begin{equation}\label{5.12_direct_sum}
S=\bigoplus_{n\in\dN} S_n \qquad \text{where}\quad  S_n := A\otimes I_{\cH_n} +
I_{\gotH}\otimes T_n.
   \end{equation}
Clearly, $S_n$ is a non-negative symmetric operator in $\gotH\otimes \cH_n.$ According
to \cite[Corollary 3.10]{MalNei2012}
    \begin{equation}\label{5.13_Frid_Krein_Ext}
{\widehat S}_F = \bigoplus_{n\in\dN} {\widehat S}_{n,F}\qquad \text{and}\qquad
{\widehat S}_K=\bigoplus_{n\in\dN} {\widehat S}_{n,K},
    \end{equation}
where ${\widehat S}_{n,F}$ and ${\widehat S}_{n,K}$ denote the Friedrichs' and Krein's
extensions of the symmetric non-negative operator $S_n$, respectively. Combining
representations \eqref{5.13_Frid_Krein_Ext}  with representations \eqref{Fr5} with
bounded $T_n\in \cB(\cH_n)$ in place of $T\in \cB(\cH)$ proved at the previous step,
implies
    \bed
\begin{split}
{\widehat S}_F &= \bigoplus_{n\in\dN} {\widehat S}_{n,F} =  \bigoplus_{n\in\dN}(
{\widehat A}_F\otimes I_{\cH_n} + I_{\gotH}\otimes T_n)    \\
&= {\widehat A}_F\otimes
I_{\cH} + \bigoplus_{n\in\dN}(I_{\gotH}\otimes T_n) =   {\widehat A}_F\otimes I_{\cH}
+ I_{\gotH}\otimes T.
  \end{split}
  \eed
The representation for $S_K$ is proved similarly.
   \end{proof}
Next we are going to discuss semibounded  extensions of the operator $S = A \otimes
I_{\gotT} + I_{\gotH} \otimes T$.   It is known that under the conditions of Proposition
\ref{prop_II.1.5_01}  the following implication holds: $\widetilde A = \widetilde A^* =
A_{\Theta}$ is semi-bounded below then $\Theta$ is semi-bounded  below. The equivalence
does not hold in general.
  \begin{definition}
Let $A\ge 0$ be a non-negative symmetric operator in $\gotH$ and let
$\Pi=\{\cH,\Gamma_0,\Gamma_1\}$ be a boundary triplet for $A^*$ such that
$A_0=\widetilde A_F$. We say that $A$ satisfies LSB-property (abbreviation of lower
semi-boundedness)  if the following equivalence holds:
 $$
 A_{\Theta} = A^*_{\Theta} \ \text{is lower semi-bounded} \Longleftrightarrow \Theta =
\Theta^* \  \text{is lower semi-bounded}.
$$
  \end{definition}
To describe the operators with LSB-property we introduce the following definition.

   \begin{definition}[\cite{DM91}]
It is said that $M(\cdot)$ uniformly tends to $-\infty$ (in symbols
$M(\cdot)\rightrightarrows -\infty)$ if for any $N>0$ there exists $x_N$ such that
  \begin{equation}\label{Fr2}
\bigl(M(x)h,h\bigr) \le -N\cdot\|h\|^2\qquad \text{for}\ x\le x_N,\quad h\in\cH.
   \end{equation}
  \end{definition}
Clearly, \eqref{Fr2} implies \eqref{Fr1} but not vice versa.
  \begin{proposition}[\cite{DM91}]\label{LSB-theorem}
Let $A\ge 0$ and let $\Pi=\{\cH,\Gamma_0,\Gamma_1\}$ be a boundary triplet for $A^*$
such that $A_0=\widetilde A_F$. Then the following statements are equivalent:

\item $(i)$  $A$ satisfies LSB  property;

\item $(ii)$  $M(x)\rightrightarrows-\infty$ as  $x\to-\infty$.
  \end{proposition}

\begin{proposition}
Let $A$ be a non-negative symmetric  operator in $\gotH$ and let $\Pi_A =
\{\cH^A,\gG^A_0,\gG^A_1\}$ be a boundary triplet for $A^{*}$ such that $A_0 :=
A^*\upharpoonright\ker(\gG_0^A) = {\widehat A}_F$.  Let also $T = T^* \in \cB(\gotT)$,
$T\ge 0$ and let $S= A \otimes I_{\gotT} + I_{\gotH} \otimes T$.  If  $A$ satisfies  the
LSB-property, then the operator $S$ also satisfies the LSB--property.
 \end{proposition}
  \begin{proof}
Consider a boundary triplet $\wt \Pi_S = \{\wt \cH^S,\wt \gG^S_0,\wt \gG^S_1\}$ for
$S^*$ given by  \eqref{eq:3.57_positive}.  By Theorem \ref{th:5.1_Positive}(i),
$$
S_0 = S^*\upharpoonright\ker(\gG_0^S) =  {\widehat S}_F = {\widehat A}_F \otimes
I_{\gotT} + I_{\gotH} \otimes T.
$$
Let also $M^A(\cdot)$ and $\gamma^A(\cdot)$ be the  Weyl function and $\gamma$-field,
respectively, corresponding to the triplet $\Pi_A$. Since  $A$ satisfies  the
LSB-property and $A_0 = {\widehat A}_F$, Theorem \ref{LSB-theorem} ensures  that the
Weyl function $M^A(\cdot)$ tends to $-\infty$ uniformly, i.e.
$M^A(x)\rightrightarrows-\infty$ as $x\to-\infty$. In other words, for each $N>0$ there
exists $x_N<0$ such that $-M^A(x)\ge N$  for $x\le x_N$.

By Theorem \ref{th:5.1_Positive}(i) the Weyl function  $M^S(\cdot)$ corresponding to
$\Pi_S$  is given by  \eqref{weyl_New}.  Let $\pi=\{\Delta_k\}^p_1$ be a partition of
$\Delta=[a,b]$ and   let $\lambda_k\in\Delta_k$. Then applying  Lemma
\ref{lem_tensor_product} to the integral sum \eqref{5.21_integ_sum_for_W_F}  we get
    \begin{equation}\label{Fr7_New}
-S_p(\pi)= - \sum_{k=1}^p M^A(x_N - \lambda_k)\otimes E_T(\Delta_k)  \ge N.
    \end{equation}
Passing here to the limit as  $|\pi|\to 0$ one obtains
\bed
-M^S(x) = \int_{\gD} \widehat E_T(d\lambda) \left(M^A(z - \lambda) \otimes
I_{\gotT}\right) \ge N \qquad  \text{for} \qquad  x \le x_N.
\eed
The latter amounts  to saying that $M^S(x)\rightrightarrows-\infty$ as $x\to-\infty$. By
Theorem \ref{LSB-theorem} this property implies (in fact is equivalent to) the
LSB-property of $S$.
       \end{proof}

\section{Examples}

In  what follows  the operator $T$ is  arbitrary (not necessarily bounded) self-adjoint
operator acting on a separable Hilbert space $\gotT$.

\subsection{Schr\"odinger operators and bosons in 1D}\la{sec:4.A}

\subsubsection{Schr\"odinger operators on half-lines}\la{subsec:4.A.1}

Let $v_r \in \dR$, $b \in \dR$, and let  $H_r = -\frac{d^2}{dx^2} + v_r$  denote a
minimal operator in $\gotH_r := L^2(\gD_r)$, $\gD_r = (b,\infty)$.  Clearly, $\dom(H_r)
= W^{2,2}_0(\gD_r) :=\{f \in W_2^2((b,\infty)): f(b) = f'(b) = 0 \}$ and  $H_r$ is a
closed  densely defined symmetric operator with $n_\pm(H_r) = 1$. The adjoint operator
is given by the same expression $H^*_r = -\frac{d^2}{dx^2} + v_r$ on the domain
$\dom(H^*_r) = W^{2,2}(\gD_r)$. One easily checks that a triplet $\Pi_{H_r} =
\{\cH^{H_r},\gG^{H_r}_0,\gG^{H_r}_1\}$  with
\bed
\cH^{H_r} := \dC, \quad \gG^{H_r}_0f = f(b), \quad \mbox{and} \quad \gG^{H_r}_1f =
f'(b), \quad f \in \dom({H_r}^*),
\eed
is a boundary triplet for $H_r^*$. The corresponding $\gamma$-field $\gga^{H_r}(\cdot)$
and Weyl function $M^{H_r}(\cdot)$ are given by
  \bed
(\gga^{H_r}(z)\xi)(x) = e^{i\sqrt{z-v_r}(x-b)}\xi, \quad \xi \in \dC, \quad x \in
\gD_r,\quad z \in \dC_\pm,
  \eed
and
  \be\la{eq:6.1a}
M^{H_r}(z) =  m^{H_r}(z) = i\sqrt{z - v_r}, \quad z \in \dC_\pm,
  \ee
respectively.   The function $\sqrt{\cdot}$ is defined on $\dC$ with the cut along
the positive semi-axis $\dR_+$. Its branch is fixed by the condition $\sqrt{1} =1$.
Clearly,  the Weyl function $M^{H_r}(\cdot)$ is a scalar function.

Let us consider the closed densely defined symmetric operator
  \be\label{6.1_S-L_oper_with_oper_Poten}
S_r = \overline{H_r \otimes I_\gotT + I_{\gotH_r} \otimes T}
  \ee
on the Hilbert space $\gotK_r := \gotH_r \otimes \gotT = L_2(\gD_r,\gotT)$. In the
following we use the notation $\vec{f}(x)$, $x \in \gD_r$ for elements of $\gotK_r =
L_2(\gD_r,\gotT)$.
 In accordance with Theorem \ref{th:2.5}(iii)
there is a boundary triplet $\Pi_{S_r} = \{\cH^{S_r},\gG^{S_r}_0,\gG^{S_r}_1\}$ for
${S_r}^*$ such that $\cH^{S_r} = \cH^{H_r} \otimes \gotT = \gotT$,
   \be\label{6.2_B_triplet_for_S-L_oper}
\begin{split}
\gG^{S_r}_0 \vec{f} &= \sqrt{\im(m^{H_r}(i - T))}\,\vec f(b),\\
\gG^{S_r}_1 \vec f &= \frac{1}{\sqrt{\im(m^{H_r}(i - T))}}\left(\vec{f}'(b) -
\re(m^{H_r}(i-T))\vec f(b)\right),
\end{split}
\ee
$\vec{f} \in \dom ({H_r}^* \otimes I_\gotT) \cap \dom(I_{\gotH_r} \otimes T) =
W^{2,2}(\gD_r,\gotT) \cap \dom(I_{\gotH_r} \otimes T) \subseteq \dom(S^*_r)$. The
corresponding $\gamma$-field $\gga^{S_r}(\cdot): \gotT \longrightarrow \gotK_r $ and
Weyl function $M^{S_r}(\cdot): \gotT \longrightarrow \gotT$ are given by
\bed
(\gga^{S_r}\xi)(x) = e^{i\sqrt{z -v_r
-T}(x-b)}\frac{1}{\sqrt{\im(m^{H_r}(i-T))}}\xi, \quad \xi \in \gotT, \quad x \in \gD_r,
\eed
and
\be\la{eq:6.4a}
M^{S_r}(z) = \frac{m^{H_r}(z -T) - \re(m^{H_r}(i-T))}{\im(m^{H_r}(i-T))},\qquad z
\in \dC_{\pm}.
  \ee

Of course, the considerations are similar for the interval $\gD_l = (-\infty,a)$, $a \in
\dR$. Let $H_l = -\frac{d^2}{dx^2} + v_l$, $v_l \in \dR$, with domain $\dom(H_l) :=
W^{2,2}_0(\gD_l)$ defined on $\gotH_l := L^2(\gD_l,\gotT)$. One checks that
$\Pi_{H_l} = \{\cH^{H_l},\gG^{H_l}_0,\gG^{H_l}_1\}$,
\bed
\cH^{H_l} := \dC, \quad \gG^{H_l}_0f = f(a), \quad \mbox{and} \quad \gG^{H_l}_1f =
-f'(a), \quad f \in \dom({H_l}^*),
\eed
is a boundary triplet for $H^*_l$. The $\gga$-field and Weyl function are given  by
\bed
(\gga^{H_l}(z)\xi)(x) = e^{i\sqrt{z-v_l}(a-x)}\xi, \quad \xi \in \dC, \quad x \in
\gD_l,\quad z \in \dC_\pm,
\eed
and
   \be\la{eq:6.4}
M^{H_l}(z) = m^{H_l}(z) = i\sqrt{z - v_l}, \quad z \in \dC_\pm.
  \ee
Let us consider the closed densely defined symmetric operator $S_l = \overline{H_l \otimes
I_\gotT + I_{\gotH_l} \otimes T}$ acting in $\gotK_l := \gotH_l \otimes \gotT =
L^2(\gD_l,\gotT)$. As above one finds
\be\la{eq:6.3}
\begin{split}
\gG^{S_l}_0 \vec{f} &= \sqrt{\im(m^{H_l}(i - T))}\,\vec f(a),\\
\gG^{S_l}_1 f &= \frac{1}{\sqrt{\im(m^{H_l}(i - T))}}\left(-\vec f'(a) -
\re(m^{H_l}(i-T))\vec f(a))\right)
\end{split}
\ee
$\vec{f} \in \dom ({H_l}^* \otimes I_\gotT) \cap \dom(I_{\gotH_l} \otimes T) =
W^{2,2}(\gD_l,\gotT) \cap \dom(I_{\gotH_l} \otimes T) \subseteq \dom(S^*_l)$ as well as
\bed
(\gga^{S_l}\xi)(x) = e^{i\sqrt{z -v_l
-T}(a-x)}\frac{1}{\sqrt{\im(m^{H_l}(i-T))}}\xi, \quad \xi \in \gotT, \quad x \in \gD_r,
\eed
and
\be\la{eq:6.7}
M^{S_l}(z) = \frac{m^{H_l}(z -T) - \re(m^{H_l}(i-T))}{\im(m^{H_l}(i-T))},\qquad  z
\in \dC_{\pm}.
  \ee

\subsubsection{Schr\"odinger operators on bounded intervals}

Let $\gD_c = (a,b)$ and $v_c\in \dR$. Consider a minimal Sturm-Liouville  operator $H_c$
in $\gotH_c = L^2(\gD_c)$ given by
\bed
\begin{split}
(H_cf)(x) &=-\frac{d^{2}}{dx^2}f(x) + v_cf(x),\quad x \in \gD_c, \\
f \in \dom(H_c) &=\left\{f \in W^{2,2}(\gD_c):
\ba{c}
f(a) = f(b) = 0\\
    f'(a) = f'(b) = 0
        \ea
        \right\}.
        \end{split}
\eed
Clearly,  $H_c$ is a closed symmetric operator  with the  deficiency indices  $n_\pm(A)
= 2$. Its adjoint  $H^*_c$ is given by
\bed
(H^*_cf)(x) = -\frac{d^2}{dx^2}f(x) +v_c f(x), \quad f \in \dom(H^*_c) = W^{2,2}(\gD_c).
\eed
Consider the extension (Dirichlet realization) $H^D_c$ of the minimal operator $H_c$
defined by
\bed
\begin{split}
 H^D_c =-\frac{d^{2}}{dx^2} + v_c, \qquad \dom(H^D_c) =\{f \in W^{2,2}(\gD_c):
f(a)=f(b)=0\}.
\end{split}
\eed
The Neumann extension (realization)  $H_N$ is fixed by
  \bed
\bspi H^N_c =-\frac{d^{2}}{dx^2} + v_c,  \quad  \dom(H^N_c) =\{f \in W^{2,2}(\gD_c):
f'(a)=f'(b)=0\}.
\end{split}
\eed
One easily checks that the triplet $\Pi_{H_c} := \{\cH^{H_c},\gG^{H_c}_0,\gG^{H_c}_1\}$
with
\bed
\cH^{H_c} := \dC^2, \quad \gG^{H_c}_0f =
\frac{1}{\sqrt{2}}\begin{pmatrix} f(a)+ f(b)\\
f(a) -f(b)
\end{pmatrix}, \quad \gG^{H_c}_1f =
\frac{1}{\sqrt{2}}\begin{pmatrix} f'(a) - f'(b)\\
f'(a) + f'(b)
\end{pmatrix},
\eed
$f \in \dom(H^*_c)$, is a boundary triplet for $H^*_c$. Clearly,  $H^D_c =
H^*_c\upharpoonright\ker(\gG^{H_c}_0)$ and $H^N_c =
H^*_c\upharpoonright\ker(\gG^{H_c}_1)$. The corresponding $\gamma$-field
$\gga^{H_c}(\cdot)$ and Weyl function $M^{H_c}(\cdot)$ are given by
\bed
(\gga^{H_c}(z)\xi)(x) = \frac{1}{\sqrt{2}}\left(\frac{\cos(\sqrt{z-v_c}(x - \nu))}{\cos(\sqrt{z-v_c}\,d)}, -\frac{\sin(\sqrt{z-v_c}(x - \nu))}{\sin(\sqrt{z-v_c}\,d)}\right) \cdot \begin{pmatrix} \xi_1 \\ \xi_2 \end{pmatrix},
\eed
$\quad z \in \dC_{\pm}$, $x \in (a,b)$, $\nu := \tfrac{a+b}{2}$, $d := \tfrac{b-a}{2}$, and
\bed
M^{H_c}(z) =
\begin{pmatrix}
m^{H_c}_1(z) & 0\\
0 & m^{H_c}_2(z)
\end{pmatrix}, \quad z \in \dC_{\pm},
\eed
where
\bed
\begin{split}
m^{H_c}_1(z) &:= \sqrt{z-v_c}\tan(\sqrt{z-v_c}\,d),\\
m^{H_c}_2(z) &:= -\sqrt{z-v_c}\cot(\sqrt{z-v_c}\,d),
\end{split}
\quad z \in \dC_{\pm}.
\eed
Notice that the Weyl function $M^{H_c}(\cdot)$ is of quasi scalar type.

We consider the closed densely defined symmetric operator
  \be\label{6.2.2_S-L_on_finite_int-l}
S_c := \overline{H_c \otimes I_{\gotT} + I_{\gotH_H} \otimes T}.
  \ee
defined on $\gotK_c := \gotH_c \otimes \gotT = L^2(\gD_c,\gotT)$. Elements of $ L^2(\gD_c,\gotT)$ are denoted by
$\vec{f}(x)$, $x \in \gD_c$. Obviously, the self-adjoint operators
$S^D_c := \overline{H^D_c \otimes I_\gotT + I_{\gotH_H} \otimes T}$ and $S^N_c := \overline{H^N_c \otimes I_\gotT + I_{\gotH_H} \otimes T}$ are self-adjoint extensions of $S_c$.

Let us introduce the subspaces $\cH^{H_c}_1 := \dC$ and $\cH^{H_c}_2 := \dC$. Notice that $\cH^{H_c} =  \transpose{(\cH^{H_c}_1 \oplus \cH^{H_c}_2)}$.
It follows from \eqref{eq:3.71} that there is a boundary triplet $\Pi_{S_c} = \{\cH^{S_c},\gG^{S_c}_0,\gG^{S_c}_1\}$ for $S^*_c$ such that
\bed
\cH^{S_c} = \cH^{H_c} \otimes \gotT =
\begin{matrix}
\cH^{H_c}_1 \otimes \gotT\\
\oplus\\
\cH^{H_c}_2 \otimes \gotT
\end{matrix} =
\begin{matrix}
\gotT\\
\oplus\\
\gotT
\end{matrix} =:
\begin{matrix}
\cH^{S_c}_1\\
\oplus\\
\cH^{S_c}_2
\end{matrix}
\eed
and
\begin{align*}
&\gG^{S_c}_0\vec{f} = \frac{1}{\sqrt{2}}
\begin{pmatrix}
\sqrt{\im(m^{H_c}_1(i-T)}(\vec f(a) + \vec{f}(b)\\
\sqrt{\im(m^{H_c}_2(i-T))}(\vec f(a) - \vec{f}(b))
\end{pmatrix}\,\\
&\gG^{S_c}_1(z)\vec{f} =\\
&\frac{1}{\sqrt{2}}
\begin{pmatrix}
\frac{1}{\sqrt{\im(m^{H_c}_1(i-T)}}\left(\vec f'(a) - \vec f'(b) - \re(m^{H_c}_1(i-T))(\vec{f}(a) + \vec{f}(b)\right)\\
\frac{1}{\sqrt{\im(m^{H_c}_2(i-T))}}\left(\vec f'(a) + \vec f'(b)) - \re(m^{H_c}_2(i-T))(\vec{f}(a) - \vec{f}(b))\right)
\end{pmatrix}\,\nonumber
\end{align*}
$\vec f \in \dom(H^*_c \otimes I_\gotT) \cap \dom(I_{\gotH_c} \otimes T) = W^{2,2}(\gD_c,\gotT) \cap \dom(I_{\gotH_c} \otimes T) \subseteq \dom(S^*_c)$.
From \eqref{eq:3.72} we get the $\gga$-field
$\gga^{S_c}(\cdot): \transpose{(\cH^{S_c}_1 \oplus \cH^{S_c}_2)} \longrightarrow \gotK_c$,
\bed
\begin{split}
(\gga^{S_c}(z)\vec{\xi})(x) =& \frac{\cos(\sqrt{z-T-v_c}(x - \nu))}{\sqrt{2}\cos(\sqrt{z-T-v_c}\,d)\sqrt{\im(m^{H_c}_1(i-T))}}\vec{\xi}_1\\
 &-\frac{\sin(\sqrt{z-T-v_c}(x - \nu))}{\sqrt{2}\sin(\sqrt{z-T-v_c}\,d)\sqrt{\im(m^{H_c}_2(i-T))}}\vec{\xi}_2,
\end{split}
\eed
$z \in \dC_{\pm}$. Finally, from \eqref{eq:3.73} the Weyl function $M^{S_c}(\cdot):  \transpose{(\cH^{S_c}_1 \oplus \cH^{S_c}_2)} \longrightarrow \transpose{(\cH^{S_c}_1 \oplus \cH^{S_c}_2)}$ is computed by
\bed
M^{S_c}(z) =
\begin{pmatrix}
\frac{m^{H_c}_1(z-T) - \re(m^{H_c}_1(i-T))}{\im(m^{H_c}_1(i-T))} &  0\\
0 & \frac{m^{H_c}_2(z-T) - \re(m^{H_c}_2(i-T))}{\im(m^{H_c}_2(i-T))}
\end{pmatrix}\;, \quad z \in \dC_\pm.
\eed
\begin{remark}
{\em
 Sturm-Liouville  operators  $S_c$ with  operator-valued potential $T = T^*\in \cC(\gotT)$
have  first  been treated on a finite interval in the pioneering paper by M.L. Gorbachuk
\cite{Gor71}. Clearly, the corresponding minimal operator $S_c$ admits representation
\eqref{6.2.2_S-L_on_finite_int-l}.
 In particular, a boundary triplet for $S_c^*$  was first constructed in \cite{Gor71}
(see also \cite{GG91}). A construction of a boundary triplet for $S_r^*$   in the case
of semi-axis has first been proposed in \cite[Section 9]{DM91}. However, our
construction   \eqref{6.2_B_triplet_for_S-L_oper}  of the boundary triplet for $S_r^*$
is borrowed from   \cite{MalNei2012} where a representation of $S$ as a direct
sum $S = \bigoplus_j S_j$ with bounded $T_j$ in place of unbounded $T$ was first proposed and
the regularization procedure for direct sums was invented  and applied to the
operator $S_r$.

 After appearance of the work \cite{Gor71} the spectral theory of self-adjoint and
dissipative extensions of $S_c$ in $L^2(\gD_c,\gotT)$ has intensively been investigated.
The results are summarized in  \cite[Chapter 4]{GG91} where one finds, in particular,
criteria for discreteness of the spectra, asymptotic formulas for the eigenvalues,
resolvent comparability results, etc. Spectral properties  of self-adjoint
extensions of $S_r$ was investigated  in  \cite{MalNei2012}, see also \cite{MalNei09}. In particular,  a criterion
for all self-adjoint extensions of  $S_r$ to have absolutely continuous non-negative
spectra was also obtained there.
}
\end{remark}

\subsection{Dirac operators and  bosons in 1D}

In the following we consider the Dirac operator instead of the Schr\"odinger operator,
cf. \cite{CarlMalPos2013}, \cite{BoiNeiPop2016}.

\subsubsection{Dirac operators on half-lines}\label{subsec:B.1}

In the Hilbert space $\gotD_r = L^2(\gD_r,\dC^2)$, $\gD_r = (b,\infty)$, let us consider the Dirac operator
\bed
\begin{split}
D_rf &:= \left(-ic\frac{d}{dx} \otimes \gs_1\right)f  + \left(\frac{c^2}{2} \otimes \gs_3\right)f, \\
f \in \dom(D_r) &:= W^{1,2}_0(\gD_r,\dC^2) :=\{f \in W^{1,2}(\gD_r,\dC^2): f(b) = 0\}.
\end{split}
\eed
Here
\bed
\gs_1 :=
\begin{pmatrix}
0 & 1\\
1 & 0
\end{pmatrix}
\quad \mbox{and} \quad
\gs_3 :=
\begin{pmatrix}
1 & 0\\
0 & -1
\end{pmatrix}.
\eed
Notice that
\bed
D^*_rf = \left(-ic\frac{d}{dx} \otimes \gs_1\right)f + \left(\frac{c^2}{2} \otimes \gs_3\right)f,
\;\;f \in \dom(D^*_r) = W^{1,2}(\gD_r,\dC^2).
\eed
One easily checks that $n_\pm(D_r) = 1$. From Lemma 3.3 of \cite{CarlMalPos2013}
we get that $\Pi_{D_r} = \{\cH^{D_r},\gG^{D_r}_0,\gG^{D_r}_1\}$,
\bed
\cH^{D_r} := \dC, \quad \gG^{D_r}_0f := f_1(b), \quad
\gG^{D_r}_1f:= ic f_2(b)
\eed
is a boundary triplet for $D^*_r$. The $\gga$-field $\gga^{D_r}(\cdot)$ and Weyl function are given by
\bed
(\gga^{D_r}(z)\xi)(x) =
\begin{pmatrix}
e^{ik(z)(x-b)}\xi\\
k_1(z)e^{ik(z)(x-b)}\xi
\end{pmatrix},
\quad x \in \gD_r, \quad z \in \dC_\pm, \quad \xi \in \cH^{D_r}.
\eed
and
\bed
M^{D_r}(z) = m^{D_r}(z)I_{\cH^{D_r}}, \quad m_r(z) := ic\, k_1(z), \quad z \in \dC_\pm.
\eed
Here
\bed
k(z) := \frac{1}{c}\sqrt{z^2 - \frac{c^4}{4}}, \quad z \in \dC,
\eed
where the branch of the multifunction $k(\cdot)$ is fixed by the condition $k(x) > 0$
for $x > \frac{c^2}{2}$. Notice that $k(\cdot)$ is holomorphic on $\dC \setminus
\left\{(-\infty,-\frac{c^2}{2}] \cup [\frac{c^2}{2},\infty)\right\}$. Further, let
\bed
k_1(z) := \frac{c\,k(z)}{z + \frac{c^2}{2}}, \quad z \in \dC.
\eed
Clearly, $k_1(\cdot)$  is also holomorphic on $\dC \setminus
\left\{(-\infty,-\frac{c^2}{2}] \cup [\frac{c^2}{2},\infty)\right\}$ and  admits
the representation
\bed
k_1(z) = \sqrt{\frac{z - \frac{c^2}{2}}{z + \frac{c^2}{2}}}, \quad z \in \dC,
\eed
where the branch of $\sqrt{\frac{z - \frac{c^2}{2}}{z + \frac{c^2}{2}}}$ is fixed by
the condition $\sqrt{\frac{x - \frac{c^2}{2}}{x + \frac{c^2}{2}}} > 0$ for $x > \frac{c^2}{2}$.

Let us consider the closed densely defined symmetric operator
\bed
S_r = \overline{D_r \otimes I_\gotT + I_{\gotD_r} \otimes T},
\eed
which is defined on
\bed
\gotK_r := \gotD_r \otimes \gotT = L^2(\gD_r,\transpose{(\gotT \oplus \gotT)}).
\eed
In the following we denote elements of $\gotK_r$ by $\vec{f} =
\transpose{(\vec{f}_1,\vec{f}_2)}$. Since $D_r$ is not semi-bounded from below the sum
$\overline{D_r \otimes I_\gotT + I_{\gotD_r} \otimes T}$ is also not semi-bounded from
below. Nevertheless, by Theorem \ref{th:2.5}\,(iii) a boundary triplet  $\Pi_{S_r} =
\{\cH^{S_r},\gG^{S_r}_0,\gG^{S_r}_1\}$ for $S^*_r$ can be chosen in the form
\bed
\cH^{S_r} = \cH^{D_r} \otimes \gotT = \gotT
\eed
and
\bed
\begin{split}
\gG^{S_r}_0\vec{f}  &= \sqrt{\im(m^{D_r}(i-T))}\,\vec{f}_1(b), \\
\gG^{S_r}_1\vec{f}  &= \frac{1}{\sqrt{\im(m^{D_r}(i-T))}}\left(ic\vec{f}_2(b) - \re(m^{D_r}_1(i-T))\vec{f}_1(b)\right),
\end{split}
\eed
$\vec f \in \dom(D^*_r \otimes I_{\gotT}) \cap \dom(I_{\gotD_r} \otimes T) =
W^{2,2}(\gD_r,\transpose{(\gotT \oplus\gotT)}) \cap \dom(I_{\gotD_r} \otimes T)
\subseteq \dom(S^*_r)$.
The $\gga$-field  $\gga^{L_r}(\cdot): \gotT \longrightarrow \gotK_r$ and Weyl function
$M^{S_r}(\cdot): \gotT \longrightarrow \gotT$ are given by
\bed
(\gga^{S_r}(z) \xi)(x) =
\begin{pmatrix}
e^{ik(z -T)(x-b)}\frac{1}{\im(m^{D_r}(i-T))} \xi\\
k_1(z-T)e^{ik(z-T)(x-b)}\frac{1}{\im(m^{D_r}(i-T))}\xi
\end{pmatrix},
\quad  \xi \in \gotT,
\eed
$x \in \gD_r$, $z \in \dC_\pm$ and
\bed
M^{S_r}(z) = \frac{m^{D_r}(z - T) - \re(m^{D_r}(i-T))}{\im(m^{D_r}(i-T))}, \quad z \in \dC_\pm.
\eed
The Dirac operator on the half-axis $(-\infty,a)$ is  treated similarly.

\subsubsection{Dirac operators on bounded intervals}

Let us consider the closed densely defined symmetric operator
\begin{align*}
D_cf &:= \left(-ic\frac{d}{dx} \otimes \gs_1\right)f + \left(\frac{c^2}{2} \otimes \gs_3\right)f, \\  
f \in \dom(D_c) &:= W^{1,2}_0(\gD_c,\dC^2) :=\{f \in W^{1,2}(\gD_c,\dC^2): f(a) = f(b) = 0\},\nonumber
\end{align*}
where $\gD_c = (a,b)$, acting in the Hilbert space $\gotD_c := L^2(\gD_c,\dC^2)$. Notice
that $n_\pm(D_c) = 2$. The adjoint operator $D^*_c$  is given by
\bed
D^*_cf = \left(-ic\frac{d}{dx} \otimes \gs_1\right)f + \left(\frac{c^2}{2} \otimes \gs_3\right)f, \qquad
f \in \dom(D^*_c) = W^{1,2}(\gD_c,\dC^2).
\eed
The triplet $\Pi_{D_c} = \{\cH^{D_c},\gG^{D_c}_0,\gG^{D_c}_1\}$, $\cH^{D_c} := \dC^2$,
\be\la{eq:6.2a}
\begin{split}
\gG^{D_c}_0
\begin{pmatrix}
f_1\\
f_2
\end{pmatrix}
 &:=
\frac{1}{\sqrt{2}}\begin{pmatrix}
f_1(a) + f_1(b)\\
f_1(a) - f_1(b)
\end{pmatrix},\\
\gG^{D_c}_1
\begin{pmatrix}
f_1\\
f_2
\end{pmatrix}
&:=
\frac{ic}{\sqrt{2}}\begin{pmatrix}
f_2(a) - f_2(b)\\
f_2(a) + f_2(b)
\end{pmatrix},
\end{split}
\ee
$f \in \dom(D^*_c)$, forms a boundary triplet for $D^*_c$. The $\gga$-field and the Weyl function are given by
\bed
\gga^{D_c}(z)
\begin{pmatrix}
\xi_1\\
\xi_2
\end{pmatrix} =
\frac{1}{\sqrt{2}}
\begin{pmatrix}
\frac{\cos(k(z)(x-\nu))}{\cos(k(z)d)} & \frac{\sin(k(z)(x-\nu))}{\sin(k(z)d)}\\[2mm]
ik_1(z)\frac{\sin(k(z)(x-\nu))}{\cos(k(z)d} & ik_1(z)\frac{\cos(k(z)(x-\nu))}{\sin(k(z)d)}
\end{pmatrix}
\begin{pmatrix}
\xi_1\\
\xi_2
\end{pmatrix}\,,
\eed
$z \in \dC_\pm$, and
\bed
M^{D_c}(z) =
\begin{pmatrix}
m^{D_c}_1(z) & 0\\
0 &  m^{D_c}_2(z)
\end{pmatrix}, \quad z \in \dC_\pm,
\eed
where
\bed
\begin{split}
m_1^{D_c}(z) &:= ck_1(z)\tan(k(z)d)\\
m^{D_c}_2(z) &:= -ck_1(z)\cot(k(z)d)
\end{split}, \quad z \in \dC_\pm,
\eed
and $d := \frac{b-a}{2}$, $\nu := \frac{b+a}{2}$. Notice that the Weyl function $M^{D_c}(\cdot)$ is of quasi scalar type.
The self-adjoint extension $D^{(1)}_c := D^*_c\upharpoonright\ker(\gG^{D_c}_0)$ has the domain
\bed
\dom(D^{(1)}_c) = \{f \in W^{1,2}(\gD_c,\dC^2): f_1(a) = f_1(b) = 0\}
\eed
while the extension $D^{(2)}_c := D^*_c\upharpoonright\ker(\gG^{D_c}_1)$ has the domain
\bed
\dom(D^{(2)}_c) = \{f \in W^{1,2}(\gD_c,\dC^2): f_2(a) = f_2(b) = 0\}.
\eed

We consider the closed symmetric operator
\bed
S_c := \overline{D_c \otimes I_\gotT + I_{\gotD_c} \otimes T}
\eed
which is defined on $\gotK_c := \gotD_c \otimes \gotT = L^2(\gD_c,\transpose{(\gotT \oplus \gotT)})$.
In the following we denote elements of of $\gotK_c$ by $\vec{f}$. In particular, we use the notation
\bed
\vec{f} =
\begin{pmatrix}
\vec{f}_1\\
\vec{f}_2
\end{pmatrix}, \quad \vec{f}_j \in L^2(\gD_c,\gotT), \quad j =1,2.
\eed

Let us construct the boundary triplet $\Pi_{S_c} = \{\cH^{S_c},\gG^{S_c}_0,\gG^{S_c}_1\}$ for $S^*_c$.
Since the  Weyl function $M^{D_c}(\cdot)$ is of quasi scalar type we follow Remark \ref{rem:3.9}.
To this end we introduce the subspaces $\cH^{D_c}_1 := \dC$ and $\cH^{D_c}_2 = \dC$. This yields $\cH^{S_c}_1 = \gotT$
and $\cH^{S_c}_2 = \gotT$ as well as
\bed
\cH^{S_c} =
\begin{matrix}
\cH^{S_c}_1\\
\oplus\\
\cH^{S_c}_2
\end{matrix} =
\begin{matrix}
\gotT\\
\oplus\\
\gotT
\end{matrix}\,.
\eed
Furthermore, we have
\bed
\gG^{S_c}_0\vec{f} =
\frac{1}{\sqrt{2}}
\begin{pmatrix}
\sqrt{\im(m^{D_c}_1(i-T))}(\vec{f}_1(a) + \vec{f}_1(b))\\
\sqrt{\im(m^{D_c}_2(i-T))}(\vec{f}_1(a) - \vec{f}_1(b))
\end{pmatrix}
\eed
and
\begin{align*}
\gG^{S_c}_1&\vec{f} =\\
&\frac{ic}{\sqrt{2}}
\begin{pmatrix}
\frac{1}{\sqrt{\im(m^{D_c}_1(i-T))}}\left(\vec{f}_2(a) - \vec{f}_2(b) - \re(m^{D_c}_1(i-T))(\vec{f}_1(a) + \vec{f}_1(b)\right)\\[1ex]
\frac{1}{\sqrt{\im(m^{D_c}_2(i-T))}}\left(\vec{f}_2(a) + \vec{f}_2(b) - \re(m^{D_c}_2(i-T))(\vec{f}_1(a) - \vec{f}_1(b)\right)
\end{pmatrix}\,,\nonumber
\end{align*}
$\vec{f} \in \dom(D^*_c \otimes I_\gotT) \cap \dom(I_{\gotD_c} \otimes T)$. The $\gga$-field $\gga^{S_c}(\cdot): \cH^{S_c} \longrightarrow \gotK_c$
is computed by
\begin{align*}
(\gga^{S_c}&(z)\vec{\xi})(x) =\\
&\frac{1}{\sqrt{2}}
\begin{pmatrix}
\frac{\cos(k(z - T)(x-\nu))}{\cos(k(z - T)d)\sqrt{\im(m^{D_c}_1(i-T))}} & -\frac{\sin(k(z - T)(x - \nu))}{\sin(k(z-T)d)\sqrt{\im(m^{D_c}_2(i-T))}}\\[2ex]
i\frac{k_1(z-T)\sin(k(z-T)(x-\nu))}{\cos(k(z-T)d)\sqrt{\im(m^{D_c}_1(i-T))}} & i\frac{k_1(z-T)\cos(k(z-T)(x-\nu))}{\sin(k(z-T)d)\sqrt{\im(m^{D_c}_2(i-T))}}
\end{pmatrix}
\begin{pmatrix}
\xi_1\\
\xi_2
\end{pmatrix},
\end{align*}
$z \in \dC_{\pm}$.
The Weyl function $M^{S_c}(\cdot): \cH^{S_c} \longrightarrow \cH^{S_c}$ is given by
\bed
M^{S_c}(z) =
\begin{pmatrix}
\frac{m^{D_c}_1(z-T) - \re(m^{D_c}_1(i-T))}{\im(m^{D_c}_1(i-T))} & 0\\
0 & \frac{m^{D_c}_2(z-T) - \re(m^{D_c}_2(i-T))}{\im(m^{D_c}_2(i-T))},
\end{pmatrix},  \quad z \in \dC_{\pm}.
\eed
\begin{remark}
{\em
Let us note that for the Dirac operator $D_c$  on a bounded interval the boundary triplet
$\Pi_{D_c}\wh{\otimes}I_\gotT = \{\wh{\cH^S},\wh{\gG^S_0},\wh{\gG^S_1}\}$ where
$\wh{\cH^S} := \dC^2 \otimes \gotT$, $\wh{\gG^S_0} := \gG^{D_c}_0 \wh\otimes I_\gotT$ and
 $\wh{\gG^S_1} := \gG^{D_c}_1 \wh\otimes I_\gotT$ is already a boundary triplet for $S^*_c$ where
$S_c = D_c \otimes I_\gotT + I_{\gotD_c}\otimes T$, see \cite{BoiNeiPop2016}. In other words it is not necessary to regularize
the boundary triplet $\Pi_{D_c}\wh{\otimes}I_\gotT$.
}
\end{remark}

\section{A model for electronic transport through a boson cavity}

Let us propose a simple model describing the electronic transport through an optical cavity, cf. \cite{NeidWilhelmZag2014,NeidWilhelmZag2015}. We
consider the Hilbert space $\gotH = L^2(\dR_-) \oplus L^2(\dR_+)$ where $\dR_- =
(-\infty,0)$ and $\dR_+ = (0,\infty)$. On the subspaces $\gotH_l := L^2(\dR_-)$ and
$\gotH_r := L^2(\dR_+)$ we consider the closed symmetric operators $H_l =
-\frac{d^2}{dx^2} + v_l$ and  $H_r = -\frac{d^2}{dx^2} + v_r$ of Subsection
\ref{subsec:4.A.1}. We set $\gotH := \gotH_l \oplus \gotH_r = L^2(\dR)$ and $H := H_l
\oplus H_r$. Notice that $H$ can be regarded as the symmetric operator
\bed
A = -\frac{d^2}{dx^2} + v(x), \qquad v(x) :=
\begin{cases}
v_l & x \in \dR_-\\
v_r & x \in \dR_+
\end{cases}
\eed
with domain $\dom(A) = W^{2,2}_0(\dR) := \{f \in W^{2,2}(\dR): f(0) = f'(0) = 0\}$. The operator $A$ is symmetric and
has deficiency indices $n_\pm(A) = 2$. For simplicity we assume that
\bed
0 \le v_r \le v_l.
\eed
Another one is the extension $A^D = H^D_l \oplus H^D_r$ where $H^D_l$ and $H^D_r$ are
the extensions of $H_l$ and $H_r$, respectively, with Dirichlet boundary conditions at
zero. One easily checks that the triple $\Pi_A = \{\cH^A,\gG^A_0,\gG^A_1\}$ with
\be\la{eq:7.1}
\cH^A := \begin{matrix}
\cH^{H_l}\\
\oplus\\
\cH^{H_r}
\end{matrix} =
\begin{matrix}
\dC \\
\oplus \\
\dC\; ,
\end{matrix} \qquad
\gG^A_0 f:=
\begin{pmatrix}
f(-0)\\
f(+0)
\end{pmatrix}, \qquad
\gG^A_1 :=
\begin{pmatrix}
-f'(-0)\\
f'(+0)
\end{pmatrix},
\ee
defines a boundary triplet for $A^*$, cf. Subsection \ref{subsec:4.A.1}. The Weyl function $M^A(z)$ of the boundary triplet $\Pi_A$ is given by
\bed
M^A(z) =
\begin{pmatrix}
m^{H_l}(z) & 0\\
0 & m^{H_r}(z)
\end{pmatrix} =
\begin{pmatrix}
i\sqrt{z-v_l} & 0\\
0 & i\sqrt{z - v_r}
\end{pmatrix}, \qquad z \in \rho(H^D).
\eed
where $A^D = A_0 := A^*\upharpoonright\ker(\gG^A_0)$.

Any other self-adjoint extension of $A$ is given by a self-adjoint relation $\gT = G(B)$ in $\dC^2$, where $B$ is a self-adjoint operator given by
\bed
B =
\begin{pmatrix}
\ga & \gga\\
\overline{\gga} & \gb
\end{pmatrix},
\eed
cf. Proposition \ref{prop2.1}. The self-adjoint extension $A_B := A^*\upharpoonright\ker(\gG^A_1 - B\gG^A_0)$ corresponds to the boundary conditions
\bed
\begin{split}
f'(-0) &= -\ga f(-0) - \gga f(+0)\\
f'(+0) &= \overline{\gga} f(-0) + \gb f(+0)
\end{split}, \quad f \in \dom(A^*) = W^{2,2}(\dR_-) \oplus W^{2,2}(\dR_+).
\eed
If the matrix $B$ is diagonal, then there is no coupling between the left and right quantum system, i.e.
the operator $A_B$ decomposes into a direct sum of two self-adjoint operators acting on $\gotH_l$ and
$\gotH_r$, respectively. If $\gga \not= 0$, then in some sense the left and right system interact.

Let us view the point zero as a quantum dot or quantum cavity. In particular, the Hilbert space $\cH^A = \dC^2$ is viewed as the  state space of the quantum dot and the self-adjoint operator $B$ as the Hamiltonian of the dot. The Hamiltonian $B$ describes a two level system to which we are going to couple to bosons. The state space of the bosons is the Hilbert space $\gotT = l_2(\dN_0)$, $\dN_0 := \{0,1,2\ldots\}$. The boson operator $T$ is given by
\bed
\begin{split}
T\vec \xi &= T\{\xi_k\}_{k\in\dN_0} = \{k\xi_k\}_{k \in \dN_0},\\
\vec \xi &= \{\xi_k\}_{k\in\dN_0}  \in \dom(T) := \{\{\xi_k\}_{k\in\dN_0} \in l_2(\dN_0): \{k\xi_k\}_{k \in \dN_0}\in l_2(\dN_0)\}.
\end{split}
\eed
The Hamiltonian $T$ describes a system of bosons which do not interact mutually.
The number of bosons is not fixed and varies from zero to infinity.
The Hilbert space $\gotT$ has a natural basis given by $e_k = \{\gd_{kj}\}_{j\in\dN_0}$.
Let us introduce the creation and annihilation operator
$b^*$ and $b$, respectively, defined by
\bed
b^*e_k = \sqrt{k+1}e_{k+1}, \quad k \in \dN_0, \quad \mbox{and} \quad be_k = \sqrt{k}e_{k-1}, \quad k \in \dN_0,
\eed
where $e_{-1} = 0$. One easily checks that $T = b^*b$.

Let us consider the  compound system $\{\cH^C,C\}$ consisting of the two level quantum
system $\{\cH^A,B\}$ and of the boson system $\{l_2(\dN_0),T\}$. Its state space is
given by $\cH^C := \cH^A \otimes l_2(\dN_0) = l_2(\dN_0,\cH^A) = l_2(\dN_0,\dC^2)$ and
the compound Hamiltonian  $C := \overline{B \otimes I_\gotT + I_{\cH^A}\otimes T}$
where obviously $\dom(C) = \dom(I_{\cH^A}\otimes T)$. The Hamiltonian $C$ does not
describe any interaction between the system $\{\cH^A,B\}$ and the bosons. To introduce
such an interaction we consider the so-called Jaynes-Cumming Hamiltonian,
cf. \cite{jaynes1963}. To this end  we consider  the eigenvalues $\gl^B_0$ and $\gl^B_1$
of B and the corresponding normalized eigenvectors $e^B_0$ and $e^B_1$. We assume that
$\gl^B_0 < \gl^B_1$. Notice that $B$ admits the representation
\bed
B = \gl^B_0 (\cdot,e^B_0)e^B_0  + \gl^B_1(\cdot,e^B_1)e^B_1.
\eed
Let us define the matrices
\bed
\begin{split}
\gs^B_+e^B_0 &= e^B_1, \qquad \gs^B_+e^B_1 = 0\\
\gs^B_-e^B_0 &= 0, \qquad \gs^B_-e^B_1 = e^B_0.
\end{split}
\eed
One easily checks that
\bed
B = \gl^B_1 \gs^B_+\gs^B_- + \gl^B_0 \gs^B_-\gs^B_+.
\eed
We set
\bed V_{JC} := \gs^B_+ \otimes b  + \gs^B_- \otimes b^*, \quad \dom(V_{JC}) =
\dom(I_{\cH^A} \otimes \sqrt{T}), \eed
and define the Jaynes-Cummings Hamiltonian $C_{JC}$ by  setting
\bed
C_{JC} = B \otimes I_\gotT+ I_{\cH^A} \otimes T + \gt V_{JC}, \quad \gt \in \dR.
\eed
One easily checks that the perturbation $V_{JC}$ is infinitesimally small with respect
to $I_{\cH^A} \otimes T$ which yields that $C_{JC}$ is self-adjoint with domain
$\dom(C_{JC}) = \dom(I_{\cH^A} \otimes T)$.

Let us consider the closed symmetric operator
\bed 
S := A \otimes I_\gotT + I_\gotH \otimes T
\eed
in the Hilbert space $\gotK := \gotH \otimes \gotT$.  Setting
\bed
\begin{split}
\gotK_l &:= \gotH_l \otimes \gotT, \qquad S_l := H_l \otimes I_\gotT + I_{\gotH_l} \otimes T, \\
\gotK_r &:= \gotH_r \otimes \gotT, \qquad S_r := H_r \otimes I_\gotT + I_{\gotH_r} \otimes T,
\end{split}
\eed
we obtain
\bed
\gotK = \gotK_l \oplus \gotK_r \quad \mbox{and} \quad S = S_l \oplus S_r.
\eed
It is desirable  to define the  operator $\wt S = \wt S^*$ describing the point
contact of the quantum system $\{\gotH,A_B\}$  with  $A_B :=
A^*\upharpoonright\ker(\gG^A_1 - B\gG^A_0)$, to the boson reservoir  treating  $C_{JC}$
as the boundary operator  with respect to a triplet $\Pi_S = \Pi_A \wh \otimes I_\gotT$,
i.e. by setting  $\wt S := S_{C_{JC}}$. However,  the triplet $\Pi_S = \Pi_A \wh\otimes
I_\gotT = \Pi_{H_l} \wh\otimes I_\gotT \oplus \Pi_{H_r} \wh\otimes I_\gotT$ is in general not a
boundary triplet for $S^*$ if $T$ is unbounded. Therefore the above treatment is
incorrect. To get a boundary triplet for $S^*$ one  has  to regularize the triplet $\Pi_S$ in accordance
with Theorem~\ref{th:2.5}. This  leads to the boundary triplet $\wt \Pi_S = \{\cH^S,\wt
\gG^S_0,\gG^S_1\} = \Pi_{S_l} \oplus \Pi_{S_r}$, where $\Pi_{S_r}$ and $\Pi_{S_l}$ are
the normalized boundary triplets \eqref{6.2_B_triplet_for_S-L_oper} and \eqref{eq:6.3},
respectively. The corresponding Weyl function $M^S(\cdot)$ is given by
\bed
M^S(z) =
\begin{pmatrix}
M^{S_l}(z) & 0\\
0 & M^{S_r}(z)
\end{pmatrix}, \quad z \in \rho(S_0), \quad S_0 = S^D_l \oplus S^D_r,
\eed
where $M^{S_r}(\cdot)$ and $M^{S_l}(\cdot)$ are defined by \eqref{eq:6.4a} and
\eqref{eq:6.7}, respectively. Moreover,  $S^D_l = H^D_l \otimes I_\gotT + I_{\gotH_l}
\otimes T$ and $S^D_r = H^D_r \otimes I_\gotT + I_{\gotH_l} \otimes T$.

The boundary operator  $\wt C_{JC}$  corresponding to the extension  $\wt S :=
S_{C_{JC}}$  in the boundary triplet $\wt \Pi_S$  is the following  regularization of
$C_{JC}:$
   \be\la{eq:6.1}
\wt C_{JC} := R^{-1}(C_{JC} - Q)R^{-1}, \quad \dom(\wt C_{JC}) = \dom(T^{3/2}),
   \ee
with
  \be\la{eq:7.2_Q}
 Q :=
\begin{pmatrix}
Q_l & 0\\
0& Q_r
\end{pmatrix}
=
\begin{pmatrix}
\re(m^{H_l}(i-T)) & 0\\
0 & \re(m^{H_r}(i-T))
\end{pmatrix}
\ee
and
\be\la{eq:7.3_R} R :=
\begin{pmatrix}
R_l & 0\\
0& R_r
\end{pmatrix} =
\begin{pmatrix}
\sqrt{\im(m^{H_l}(i-T))} & 0\\
0 & \sqrt{\im(m^{H_r}(i-T))}
\end{pmatrix}\; ,
\ee
where $m^{H_r}(\cdot)$ and $m^{H_l}(\cdot)$ are given by \eqref{eq:6.1a} and \eqref{eq:6.4}.

Next we show  that the operator $\wt C_{JC}$ is well-defined and self-adjoint.
We set
\bed
\wt C := R^{-1}(C - Q)R^{-1} \quad \mbox{and} \quad \wt V_{JC} := R^{-1} V_{JC} R^{-1}
\eed
and consider the representation  $\wt C_{JC} = \wt B + \wt T + \gt\wt V_{JC} =\wt C +
\gt\wt V_{JC}$ where
\bed
\wt C := \wt B + \wt T \quad \mbox{and} \quad \wt B :=R^{-1} (B \otimes I_\gotT) R^{-1},
\quad \wt T := R^{-1}(I_{\cH^A} \otimes T - Q)R^{-1}.
  \eed
It follows with account of  \eqref{eq:7.2_Q} and \eqref{eq:7.3_R}  that
  \bed
\begin{split}
\wt T := R^{-1}&(I_{\cH^A} \otimes T - Q)R^{-1} \\
&=
\begin{pmatrix}
R^{-1}_l(T - Q_l)R^{-1}_l & 0\\
0 & R^{-1}_r(T - Q_r)R^{-1}_r
\end{pmatrix}
=:
\begin{pmatrix}
\wt T_l & 0\\
0 & \wt T_r
\end{pmatrix}.
\end{split}
\eed
Let us introduce the operators
\bed
\begin{split}
Z_l &:= \sqrt{\sqrt{I_\gotT + (T+v_l)^2}+ T+ v_l} \ge I_\gotTß, \\
Z_r &:=\sqrt{\sqrt{I_\gotT + (T+v_r)^2}+ T+ v_r} \ge I_\gotT\,.
\end{split}
\eed
Clearly, these operators are  self-adjoint and
 $\dom(Z_l)= \dom(Z_r)= \dom(T^{1/2})$. A straightforward computation shows that
\bed
R =\frac{1}{\sqrt[4]{2}}
\begin{pmatrix}
Z^{-1/2}_l & 0\\
0 & Z^{-1/2}_r
\end{pmatrix}
\eed
and
\bed
Q = -\frac{1}{\sqrt{2}}
\begin{pmatrix}
 Z_l & 0\\
0 &  Z_r
\end{pmatrix}.
\eed
It follows that
\bed
\wt T =
\begin{pmatrix}
\sqrt{2}T Z_l + Z^2_l & 0\\
0 & \sqrt{2}T Z_r + Z^2_r
\end{pmatrix}
\eed
It is easily seen that $\dom(\wt T) = \dom(T^{3/2}) \oplus \dom(T^{3/2})$ and  $\wt T =
\wt T^*$. Moreover, it satisfies $\wt T \ge I_{\cH^A \otimes \gotT}$. The operator $\wt
B = R^{-1} (B \otimes I_\gotT) R^{-1}$ takes the form
\bed
\wt B = \sqrt{2}
\begin{pmatrix}
\ga Z_l & \gga Z^{1/2}_l Z^{1/2}_r\\
\overline{\gga}Z^{1/2}_r Z^{1/2}_l & \gb Z_r
\end{pmatrix}\,.
\eed
This operator is symmetric on the natural domain $\dom(T^{1/2}) \oplus
\dom(T^{1/2})$. Moreover, one easily checks that $\wt B$ is infinitesimally small with
respect to $\wt T$. Hence the sum $\wt C := \wt T + \wt B$ is self-adjoint on the domain
$\dom(T^{3/2}) \oplus \dom(T^{3/2})$. Furthermore, a straightforward computation
shows that the operator  $\wt V_{JC}$ is symmetric on $\dom(T)$. Moreover,  $\wt V_{JC}$
is infinitesimally small with respect to $\wt T$. Therefore the operator
sum $\wt C_{JC} = \wt B + \wt T + \wt V_{JC}$ is well-defined on $\dom(T^{3/2}) \oplus
\dom(T^{3/2})$ and self-adjoint.

According to \eqref{eq:3.57} the boundary triplet $\wt \Pi_S := \{\cH^S,\wt \gG^S_0,\wt \gG^S_1\}$
takes the form
   \be\la{eq:7.4_b_trip}
\wt\gG^S_0 f = R (\;\gG^H_0 \wh \otimes I_\gotT) f, \quad \mbox{and} \quad \wt \gG^S_1f
= R^{-1}(\gG^H_1 \wh \otimes I_\gotT - Q\; \gG^H_0 \wh \otimes I_\gotT)f,
  \ee
for $f \in \gotD = \dom(S^*) \cap \dom(I_\gotH \otimes T)$.
Hence the Hamiltonian $\wt S$ describing the contact to the reservoir is given by
\be\la{eq:7.5}
\begin{split}
\wt S = S_{\wt C_{JC}} &:= S^*\upharpoonright\dom(S_{\wt C_{JC}}),\\
\dom(S_{\wt C_{JC}}) &:= \{f\in\dom(S^*): \wt \gG^S_1f = \wt C_{JC}\wt \gG^S_0f\}.
\end{split}
\ee
Inserting \eqref{eq:7.4_b_trip} into \eqref{eq:7.5} one rewrites the last relation as
\be\la{eq:6.2}
\dom(S_{\wt C_{JC}}) \cap \gotD =\{ f\in\gotD: \gG^H_1 \wh \otimes I_\gotT f = C_{JC}\;\gG^H_0 \wh \otimes I_\gotT f\}
\ee
which coincides with the original  boundary condition.

The system $\{\gotK,S_{\wt C_{JC}}\}$ can be regarded as the Jaynes-Cummings model coupled to leads. It describes the electronic transport through a dot or cavity where the electrons interact with bosons. Notice that the interaction  of electrons to bosons is restricted only to one point.

\begin{remark}
{\em
Let us make some comments.

\item[\;\;(i)]
The system $\{\gotK,S_0\}$, $S_0 := S^*\upharpoonright\ker(\gG^S_0) = A^D \otimes I_\gotT + I_{\gotH_H}\otimes T$ describes a situation where the left system $\{\gotK_l,S^D_l\}$, $S^D_l := S^*_l\upharpoonright\ker(\gG^{S_l}_0)$, and right system $\{\gotK_r,S^D_r\}$, $S^D_r := S^*_r\upharpoonright\ker(\gG^{S_r}_0)$ are completely decoupled.

\item[\;\;(ii)]

If $\gt = 0$, then $\wt C_{JC} = \wt C$. The system $\{\gotK,S_{\wt C}\}$ describes a
situation where the left and right systems $\{\gotK_l,S^D_l\}$ and $\{\gotK_r,S^D_r\}$
are coupled by a point interaction at zero but not by a boson-electron interaction. In
particular, if $B$ is diagonal, then again the left and right systems
$\{\gotK_l,S^D_l\}$ and $\{\gotK_r,S^D_r\}$ are decoupled.

\item[\;\;(iii)]  If $\gt \not= 0$, then the system $\{\gotK,S_{\wt C_{JC}}\}$ can be viewed as fully
coupled: the left and right systems  are coupled by point interaction at zero as well as
by a boson-electron interaction at zero.

\item[\;\;(iv)] If $\gt \not= 0$ and $B$ is diagonal,  then the left and right systems are only coupled by the boson-electron interaction at zero.

\item[\;\;(v)] The model above can be viewed as a simple model of a solar cell or a light emitting diode (LED).

\item[\;\;(vi)] Using the result of subsection \ref{subsec:B.1} one can introduce a similar model where the Schr\"odinger operators $H_l$ and $H_r$ are replaced by Dirac operators $D_l$ and $D_r$, respectively.

\item[\;\;(vii)] Finally, we  mention  that the self-adjoint operators $C_{JC}$ and  $\wt C_{JC}$
have  Jacobi structure  in a basis formed by the orthogonal systems  $\{e^B_0 \otimes
e_k,e^B_1 \otimes e_k\}_{k\in\dN_0}$ in $\cH^A \otimes \gotT$. To see this we note first
that the subspaces spanned by $\{e^B_0 \otimes e_{k+1}\}_{k \in \dN_0}$ and $\{e^B_1
\otimes e_k\}_{k \in \dN_0}$, respectively,  leave invariant the self-adjoint operators
$T + \gt V_{JC}$ and $\wt T + \gt \wt V_{JC}$. Secondary, the operators $B \otimes
I_\gotT$ and $\wt B$ leave invariant the subspace spanned by the systems  $\{e^B_0
\otimes e_k\}_{k \in \dN_0}$ and $\{e^B_1 \otimes e_k\}_{k \in \dN_0}$,  respectively.
Both statements  ensures the Jacobi structure of both operators $C_{JC}$ and
$\wt C_{JC}$.
}
\end{remark}

\section*{Acknowledgments}
This work was financially supported by the Ministry
of Education and Science of the Russian Federation
(the Agreement number No. 02.A03.21.0008),
by  the Government
of the Russian Federation (grant 074-U01), by grant MK-5161.2016.1
of the President of the Russian Federation,  by DFG grants NE~1439/3-1 and BR~1686/3-1, and
by grant 16-11-10330 of the Russian Science Foundation.
The preparation of the paper was supported by the European Research Council via ERC-2010-AdG no 267802
(``Analysis of Multiscale Systems Driven by Functionals'').
I.~P. and A.~B. thank WIAS (Berlin)
for hospitality and I.~P. also TU Clausthal. H.~N. thanks  the ITMO University of St. Petersburg (Russia)
for financial support and hospitality.


\def\cprime{$'$} \def\cprime{$'$} \def\cprime{$'$} \def\cprime{$'$}
  \def\cprime{$'$} \def\cprime{$'$} \def\cprime{$'$}
  \def\lfhook#1{\setbox0=\hbox{#1}{\ooalign{\hidewidth
  \lower1.5ex\hbox{'}\hidewidth\crcr\unhbox0}}} \def\cprime{$'$}
  \def\cprime{$'$} \def\cprime{$'$} \def\cprime{$'$} \def\cprime{$'$}
  \def\cprime{$'$} \def\cprime{$'$} \def\cprime{$'$} \def\cprime{$'$}
  \def\cprime{$'$} \def\lfhook#1{\setbox0=\hbox{#1}{\ooalign{\hidewidth
  \lower1.5ex\hbox{'}\hidewidth\crcr\unhbox0}}} \def\cprime{$'$}

\section*{Addresses:}

A.A.~Boitsev,\\ 
St. Petersburg National Research University
of Information Technologies, Mechanics and Optics,
49 Kronverkskiy, St. Petersburg 197101, Russia\\
E-mail: boitsevanton@gmail.com\\

\noindent
J.F.~Brasche\\
Institut f\"ur Mathematik, 
TU Clausthal,
Erzstr. 1, D-38678 Clausthal-Zellerfeld, Germany\\
E-mail:johannes.brasche@tu-clausthal.de\\

\noindent
M.M.~Malamud\\
Institute of Applied Mathematics and Mechanics,
NAS Ukraine,\\
Peoples Friendship University of Russia (RUDN University),
6 Miklukho-Maklaya Street, Moscow, 117198, Russia\\
E-mail: malamud3m@gmail.com\\

\noindent
H.~Neidhardt,
Weierstrass Institute of Applied Analysis and Stoachstics,
 Mohrenstr. 39, D-10117 Berlin, Germany\\
E-mail: hagen.neidhardt@wias-berlin.de\\

\noindent
I.Yu.~Popov\\
St. Petersburg National Research University
of
Information Technologies, Mechanics and Optics,
49 Kronverkskiy, St. Petersburg 197101, Russia\\
E-mail: popov1955@gmail.com

\end{document}